%% file: paper1.tex
\documentclass[12pt,a4paper]{article}
\usepackage[margin=3cm]{geometry}

\usepackage{amsmath,amsthm,amssymb}
\usepackage{color}         
\usepackage{graphicx}    
\usepackage{calc}
\usepackage{docmute}
\usepackage{array}
\usepackage{xspace}
\usepackage[numbers,sort]{natbib}
\usepackage{subfloat} 
\usepackage{mathtools}
\usepackage{textcomp}
\usepackage[font={it}]{caption}
\usepackage{bbm}            
\usepackage{enumerate}  
\usepackage{verbatim}      
\usepackage{etoolbox}      
\usepackage{hyperref}       
\usepackage{thm-restate}
\usepackage{blkarray}      

\newtheorem{theorem}{Theorem}
\newtheorem{cor}[theorem]{Corollary}
\newtheorem{proposition}[theorem]{Proposition}

\theoremstyle{definition}
\newtheorem{definition}[theorem]{Definition}
\newtheorem{varremark}[theorem]{Remark}

\usepackage{tikz}
\usetikzlibrary{calc, decorations.markings,decorations.pathmorphing,intersections,shapes}
\usetikzlibrary{fit,shapes}

\pgfdeclarelayer{back}
\pgfdeclarelayer{front}
\pgfsetlayers{back,main,front}

\makeatletter
\pgfkeys{%
  /tikz/on layer/.code={
    \pgfonlayer{#1}\begingroup
    \aftergroup\endpgfonlayer
    \aftergroup\endgroup
  },
  /tikz/node on layer/.code={
     \gdef\node@@on@layer{%
      \setbox\tikz@tempbox=\hbox\bgroup\pgfonlayer{#1}\unhbox\tikz@tempbox\endpgfonlayer\egroup}
    \aftergroup\node@on@layer
  },
  /tikz/end node on layer/.code={
    \endpgfonlayer\endgroup\endgroup
  }
}

\def\node@on@layer{\aftergroup\node@@on@layer}
\renewcommand{\-}[0]{\nobreakdash-\hspace{0pt}}

\makeatletter
\def\calign@preamble{%
   &\hfil\strut@
    \setboxz@h{\@lign$\m@th\displaystyle{##}$}%
    \ifmeasuring@\savefieldlength@\fi
    \set@field
    \hfil
    \tabskip\alignsep@
}
\let\cmeasure@\measure@
\patchcmd\cmeasure@{\divide\@tempcntb\tw@}{}{}{}
\patchcmd\cmeasure@{\divide\@tempcntb\tw@}{}{}{}
\patchcmd\cmeasure@{\ifodd\maxfields@
  \global\advance\maxfields@\@ne
  \fi}{}{}{}    
\newenvironment{calign}
{%
  \let\align@preamble\calign@preamble
  \let\measure@\cmeasure@
  \align
}
{%
  \endalign
}  
\makeatother

\newcommand\ignore[1]{}

\newenvironment{tz}
{\begin{aligned}\begin{tikzpicture}}
{\end{tikzpicture}\end{aligned}\,}
\tikzset{blob/.style={draw, circle, fill=white, inner sep=1pt, minimum width=15pt, font=\scriptsize, node on layer=front, line width=0.7pt}}
\tikzset{greenregion/.style={fill=green, fill opacity=0.3, draw=none}}
\tikzset{redregion/.style={fill=red, fill opacity=0.3, draw=none}}
\tikzset{blueregion/.style={fill=blue, fill opacity=0.3, draw=none}}
\tikzset{yellowregion/.style={fill=yellow, fill opacity=0.5, draw=none}}
\tikzset{cyanregion/.style={fill=cyan, fill opacity=0.3, draw=none}}
\tikzset{orangeregion/.style={fill=orange, fill opacity=0.6, draw=none}}
\tikzset{solidgreenregion/.style={fill=green!30, fill opacity=1, draw=none}}
\tikzset{solidredregion/.style={fill=red!30, fill opacity=1, draw=none}}
\tikzset{solidblueregion/.style={fill=blue!30, fill opacity=1, draw=none}}
\tikzset{solidyellowregion/.style={fill=yellow!30, fill opacity=1, draw=none}}
\tikzset{string/.style={line width=0.7pt}}
\tikzset{zig/.style={decoration={zigzag,segment length=3, amplitude=0.5}}}
\tikzset{bnd/.style={draw,string}}   
\tikzset{projector/.style={circle, draw, font=\scriptsize, inner sep=-5pt, minimum width=0.35cm, string, fill=white}}
\tikzset{dimension/.style={font=\scriptsize, inner sep=1pt}}

\newcounter{jvcommcounter}
\newcounter{drcommcounter}
\newcommand{\Tr}{\mathrm{Tr}}
\newcommand{\End}{\mathrm{End}}

\newcommand\cat[1]{\ensuremath{\mathbf{#1}}}
\newcommand{\UEB}{\ensuremath{\mathrm{UEB}}}
\newcommand{\QLS}{\ensuremath{\mathrm{QLS}}}
\newcommand{\HAD}{\ensuremath{\mathrm{Had}}}
\newcommand\C{\ensuremath{\mathbb{C}}\xspace}
\newcommand{\minus}{\ensuremath{\text{-}}}
\newcommand\superequals[1]{\stackrel {\makebox[0pt]{\tiny\eqref{#1}}} =}
\newcommand\Dita{Di\c{t}\u{a}\xspace}
\newcommand{\ket}[1]{\left|#1\right\rangle}
\newcommand{\bra}[1]{\left\langle#1\right|}
\newcommand{\braket}[2]{\left\langle#1|#2\right\rangle}
\newcommand\eqgap{\hspace{10pt}}
\newcommand\N{\mathbb{N}}
\def\hrt{0.707107}
\def\side{0.7}
\def\xdelta {0.2}
\def\ydelta{0.2}

\newcommand\grid[1]{\ensuremath{\def\arraystretch{1.4}\begin{array}{|c|c|c|c|c|c|c|c|}\hline#1\\\hline\end{array}}}

\renewcommand{\to}[1][]{\ensuremath{\xrightarrow{#1}}}
\newcommand{\To}[1][]{\ensuremath{\xRightarrow{#1}}}

\allowdisplaybreaks[1]

\newcommand\vc[1]{\begin{tabular}{@{}c@{}}#1\end{tabular}}

\def\figuretopsuck{\vspace{-10pt}}
\def\figurecaptionsuck{\vspace{-15pt}}
\def\figurecaptionpostsuck{\vspace{-15pt}}

\begin{document}

\title{Biunitary constructions in quantum information}

\author{
\begin{tabular}{cc}
David J. Reutter\thanks{Department of Computer Science, University of Oxford} & Jamie Vicary\thanks{School of Computer Science, University of Birmingham} \footnotemark[1]
\\
\texttt{david.reutter@cs.ox.ac.uk}
&
\texttt{j.o.vicary@bham.ac.uk}
\end{tabular}}

\date{\today}

\maketitle

\begin{abstract}
We present an infinite number of construction schemes involving unitary error bases, Hadamard matrices, quantum Latin squares and controlled families, many of which have not previously been described. Our results rely on biunitary connections,  algebraic objects which play a central role in the theory of planar algebras. They have an attractive graphical calculus which allows simple correctness proofs for the constructions we present. We apply these techniques to construct a unitary error basis that cannot be built using any previously known method. 
\end{abstract}


\section{Introduction}
\textit{Biunitary connections} (or simply \textit{biunitaries}) were introduced by Ocneanu~\cite{Ocneanu:1989} in 1989 as a central tool in the study and classification of subfactors. Here, we use an approach to biunitaries developed by Jones and others~\cite{Jones:1999,Jones:2013,Morrison:2014} within the theory of \textit{planar algebras}, which studies the linear representation theory of algebraic structures in the plane. We can describe a biunitary informally as a planar algebra element $U$ with two inputs and two outputs, drawn below and above the vertex respectively, which is \textit{vertically unitary}~\eqref{eq:biunitaryverticallyunitary}, and which is \textit{horizontally unitary} up to a scalar factor~$\lambda$~\eqref{eq:biunitaryhorizontallyunitary}:
\def\bigangle{150}
\def\smallangle{30}
\def \sidew {0.5}
\def \scl{0.69}
\begin{align}
\label{eq:biunitaryverticallyunitary}
\begin{tz}[string,scale=\scl]
\path[redregion] (0.25-\sidew,0) rectangle (1,3.5);
\path[blueregion] (1,0) rectangle (1.75+\sidew,3.5);
\path[fill=white] (0.25,0) to [out=90, in=-135] (1,1) to [out=-45, in=90] (1.75,0);
\path[greenregion,draw] (0.25,0) to [out=90, in=-135] (1,1) to [out=-45, in=90] (1.75,0);
\path[fill=white] (0.25,3.5) to [out=-90, in=135] (1,2.5) to [out=45, in=-90] (1.75,3.5);
\path[greenregion,draw] (0.25,3.5) to [out=-90, in=135] (1,2.5) to [out=45, in=-90] (1.75,3.5);
\path[fill=white] (1,2.5) to [out=-135, in=90] (0.35, 1.75) to [out=-90, in=135] (1,1) to [out=45, in=-90] (1.65,1.75) to [out=90, in=-45] (1,2.5);
\path[yellowregion,draw] (1,2.5) to [out=-135, in=90] (0.35, 1.75) to [out=-90, in=135] (1,1) to [out=45, in=-90] (1.65,1.75) to [out=90, in=-45] (1,2.5);
\node[blob] at (1,1) {$U$};
\node[blob] at (1,2.5) {$U^\dagger$};
\end{tz}
&=
\begin{tz}[string,scale=\scl]
\path[redregion] (0.25-\sidew,0) rectangle (0.25,3.5);
\path[blueregion] (1.75,0) rectangle (1.75+\sidew,3.5);
\path[greenregion] (0.25,0) rectangle (1.75,3.5);
\draw (0.25,0) to (0.25,3.5);
\draw (1.75,0) to (1.75,3.5);
\end{tz}
&
\begin{tz}[string,scale=\scl]
\path[redregion] (0.25-\sidew,0) rectangle (1,3.5);
\path[blueregion] (1,0) rectangle (1.75+\sidew,3.5);
\path[fill=white] (0.25,0) to [out=90, in=-135] (1,1) to [out=-45, in=90] (1.75,0);
\path[yellowregion,draw] (0.25,0) to [out=90, in=-135] (1,1) to [out=-45, in=90] (1.75,0);
\path[fill=white] (0.25,3.5) to [out=-90, in=135] (1,2.5) to [out=45, in=-90] (1.75,3.5);
\path[yellowregion,draw] (0.25,3.5) to [out=-90, in=135] (1,2.5) to [out=45, in=-90] (1.75,3.5);
\path[fill=white] (1,2.5) to [out=-135, in=90] (0.35, 1.75) to [out=-90, in=135] (1,1) to [out=45, in=-90] (1.65,1.75) to [out=90, in=-45] (1,2.5);
\path[greenregion,draw] (1,2.5) to [out=-135, in=90] (0.35, 1.75) to [out=-90, in=135] (1,1) to [out=45, in=-90] (1.65,1.75) to [out=90, in=-45] (1,2.5);
\node[blob] at (1,1) {$U^\dagger$};
\node[blob] at (1,2.5) {$U$};
\end{tz}
&=
\begin{tz}[string,scale=\scl]
\path[redregion] (0.25-\sidew,0) rectangle (0.25,3.5);
\path[blueregion] (1.75,0) rectangle (1.75+\sidew,3.5);
\path[yellowregion] (0.25,0) rectangle (1.75,3.5);
\draw (0.25,0) to (0.25,3.5);
\draw (1.75,0) to (1.75,3.5);
\end{tz}
\\[5pt]
\label{eq:biunitaryhorizontallyunitary}
\begin{tz}[string,scale=\scl]
\path[redregion] (3.25,2) to [out=-90, in=45] (2.5,1) to [out=-45, in=90] (3.25,0) to(3.25+\sidew,0) to (3.25+\sidew,2);
\path[redregion] (0.25,0) to [out=90, in=-135] (1,1) to [out=135, in=-90] (0.25,2) to (0.25-\sidew,2) to (0.25-\sidew,0);
\path[blueregion] (2.5,1) to [out=-135, in=0] (1.75,0.3)  to [out=180, in=-45] (1,1) to [out=45, in=180] (1.75,1.7) to [out=0, in=135] (2.5,1);
\path[greenregion,draw] (3.25,0) to [out=90, in=-45] (2.5,1)to [out=-135, in=0] (1.75,0.3)  to [out=180, in=-45] (1,1) to [out=-135, in=90] (0.25,0);
\path[yellowregion,draw] (0.25,2) to [out=-90, in=135] (1,1) to [out=45, in=180] (1.75,1.7) to [out=0, in=135] (2.5,1) to [out=45, in=-90] (3.25,2);
\node[blob] at (1,1) {$U$};
\node[blob]at (2.5,1) {$U_*$};
\end{tz}
&=
\lambda\,
\begin{tz}[string,scale=\scl] 
\path[redregion] (3.25,0) to [out=90, in=0] (1.75,0.7)  to [out=180, in=90] (0.25,0) to (0.25-\sidew,0) to (0.25-\sidew,2) to (0.25,2) to [out=-90, in=180] (1.75,1.3) to [out=0 , in=-90] (3.25,2) to (3.25+\sidew,2) to (3.25+\sidew,0);
\path[yellowregion,draw] (3.25,2) to [out=-90, in=0] (1.75,1.3)  to [out=180, in=-90] (0.25,2);
\path[greenregion,draw] (3.25,0) to [out=90, in=0] (1.75,0.7)  to [out=180, in=90] (0.25,0);
\end{tz}
&
\begin{tz}[string,scale=\scl]
\path[blueregion] (0.25,0) to [out=90, in=-135] (1,1) to [out=135, in=-90] (0.25,2) to (0.25-\sidew,2) to (0.25-\sidew,0);
\path[redregion] (2.5,1) to [out=-135, in=0] (1.75,0.3)  to [out=180, in=-45] (1,1) to [out=45, in=180] (1.75,1.7) to [out=0, in=135] (2.5,1);
\path[blueregion] (3.25,2) to [out=-90, in=45] (2.5,1) to [out=-45, in=90] (3.25,0) to(3.25+\sidew,0) to (3.25+\sidew,2);
\path[greenregion,draw] (3.25,0) to [out=90, in=-45] (2.5,1)to [out=-135, in=0] (1.75,0.3)  to [out=180, in=-45] (1,1) to [out=-135, in=90] (0.25,0);
\path[yellowregion,draw] (0.25,2) to [out=-90, in=135] (1,1) to [out=45, in=180] (1.75,1.7) to [out=0, in=135] (2.5,1) to [out=45, in=-90] (3.25,2);
\node[blob] at (1,1) {$U_*$};
\node[blob] at (2.5,1) {$U$};
\end{tz}
&=
 \lambda\,
\begin{tz} [string,scale=\scl]
\path[blueregion] (3.25,0) to [out=90, in=0] (1.75,0.7)  to [out=180, in=90] (0.25,0) to (0.25-\sidew,0) to (0.25-\sidew,2) to (0.25,2) to [out=-90, in=180] (1.75,1.3) to [out=0 , in=-90] (3.25,2) to (3.25+\sidew,2) to (3.25+\sidew,0);
\path[yellowregion,draw] (3.25,2) to [out=-90, in=0] (1.75,1.3)  to [out=180, in=-90] (0.25,2);
\path[greenregion,draw] (3.25,0) to [out=90, in=0] (1.75,0.7)  to [out=180, in=90] (0.25,0);
\end{tz} 
\end{align}
In this paper, diagrams of this sort represent simple linear algebra data: regions are labelled by indexing sets, and wires and vertices are labelled by indexed families of finite-dimensional Hilbert spaces and linear maps, respectively.\footnote{Formally this is a common generalization of the tensor~\cite[Example 2.6]{Jones:1999} and spin model~\cite[Example 2.8]{Jones:1999} planar algebras, corresponding to a fragment of the monoidal 2-category~\cat{2Hilb}~\cite{Baez:1997}. However, our exposition will be elementary, and we will not assume knowledge of these ideas.} Blank regions correspond to the trivial indexing set. In concrete terms, a biunitary therefore comprises a family of linear maps satisfying some algebraic properties.

The \textit{type} of a biunitary is the shading pattern which surrounds the vertex. We show in \autoref{sec:biunitarity} that a variety of structures in quantum information theory correspond exactly to biunitaries of particular types. Some important examples are given in \autoref{fig:biunitarycharacterizations}.\footnote{Note that some of the inputs or outputs of the biunitary may in general be composite wires. For example, in \autoref{fig:biunitarycharacterizations}(c) the first input is composite, and \autoref{fig:biunitarycharacterizations}(d) the first input and second output are composite.} In the lower-right image,  we see that the notation is 3\-dimensional, with the blue sheet lying beneath the yellow sheet; the colours do not convey mathematical information, but rather make the geometry easier to understand. Rotations by a quarter-turn, and reflections about the horizontal or vertical axes, preserve the given interpretations in terms of quantum structures.

Some of these characterizations are already known: complex Hadamard matrices were characterized by Jones as biunitaries with alternating shaded and unshaded regions~\cite{Jones:1999}, and unitary error bases were characterized by the second author as biunitaries with one shaded and three unshaded regions~\cite{Vicary:2012, Vicary:2012hq}. Here we show that quantum Latin squares can be characterized as biunitaries with two adjacent shaded regions and two adjacent unshaded regions. We also show that controlled families can be described by adding an additional shaded region in a certain way; in \autoref{fig:biunitarycharacterizations}, we illustrate one application of this idea, giving a biunitary characterization of a controlled family of Hadamard matrices. 
\begin{figure}[t]
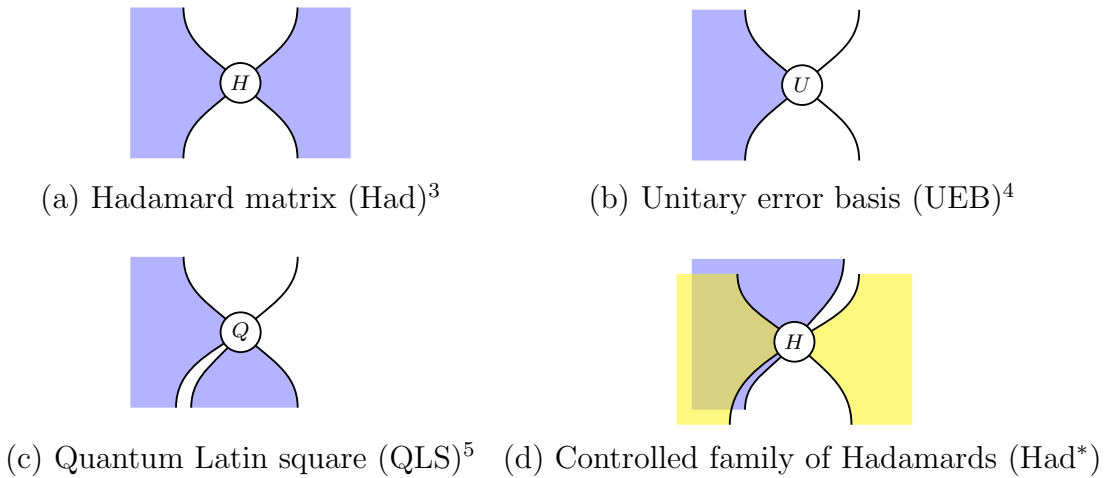

\vspace{-5pt}
\begin{calign}
\nonumber
\begin{tz}[string, scale=1.0]
\node [blob] at (1,1) {$H$};
\path [blueregion] (0.25,0) to [out=90, in=-135] (1,1) to [out=135, in=-90] (0.25,2) to (0.25-\side,2) to (0.25-\side,0);
\draw [string,bnd] (0.25,0) to [out=90, in=-135] (1,1) to [out=135, in=-90] (0.25,2);
\path [blueregion] (1.75,0) to [out=90, in=-45] (1,1) to [out=45, in=-90] (1.75,2) to (1.75+\side,2) to (1.75+\side,0);
\draw [string,bnd] (1.75,0) to [out=90, in=-45] (1,1) to [out=45, in=-90] (1.75,2);
\end{tz}
&
\begin{tz}[string]
\clip (0.25-\side,0) rectangle (1.75+\side,2);
\draw (1,1) to [out=45, in=down] (1.75,2);
\draw (1.75,0) to [out=90, in=-45] (1,1);
\path [blueregion] (0.25,0) to [out=90, in=-135] (1,1) to [out=135, in=-90] (0.25,2) to (0.25-\side,2) to (0.25-\side,0);
\draw[bnd] (0.25,0) to [out=90, in=-135] (1,1) to [out=135, in=-90] (0.25,2);
\node [blob] at (1,1) {$U$};
\end{tz}
\\
\nonumber
\text{(a) Hadamard matrix (\HAD)}\footnotemark
&
\text{(b) Unitary error basis (UEB)}\footnotemark
\\[5pt]
\nonumber
\begin{tz}[xscale=-1, string]
\clip (0.25-\side,0) rectangle (1.75+\side,2);
\node [blob] at (1,1) {$Q$};
\path [blueregion] (0.25,0) to [out=up, in=-135] (1,1) to [out=-55, in=up] (1.65,0);
\path [blueregion] (1.75,2) to [out=down, in=45] (1,1) to [out=-35, in=up] (1.85,0) to (1.75+\side,0) to (1.75+\side,2);
\draw [string,bnd] (1.75,2) to [out=-90, in=45] (1,1) to [out=-35, in=up] (1.85,0);
\draw [string] (0.25,2) to [out=-90, in=135] (1,1);
\draw [string,bnd] (0.25,0) to [out=up, in=-135] (1,1) to [out=-55, in=up] (1.65,0);
\end{tz}
&
\begin{tz}[string]
\clip (0.15-\side,-0.1) rectangle (2.05+\side,2.3);
\path [blueregion] (0.35,0.1) to [out=90, in=-125] (1,1) to [out=55, in=-90] (1.65,2.1) to (0.35-\side,2.1) to (0.35-\side,0.1);
\draw[bnd] (0.35,0.1) to [out=90, in=-125] (1,1) to [out=55, in=-90] (1.65,2.1);
\path [yellowregion] (0.15,-0.1) to [out=90, in=-145] (1,1) to [out=135, in=-90] (0.25,1.9) to (0.15-\side,1.9) to (0.15-\side,-0.1);
\draw[bnd] (0.15,-0.1) to [out=90, in=-145] (1,1) to [out=135, in=-90] (0.25,1.9);
\path [yellowregion] (1.75,-0.1) to [out=90, in=-45] (1,1) to [out=35, in=-90] (1.85,1.9) to (1.85+\side,1.9) to (1.85+\side,-0.1);
\draw[bnd] (1.75,-0.1) to [out=90, in=-45] (1,1) to [out=35, in=-90] (1.85,1.9);
\node [blob] at (1,1) {$H$};
\end{tz}
\\
\nonumber
\text{(c) Quantum Latin square (QLS)}\footnotemark
& \text{(d) Controlled family of Hadamards ($\HAD^*$)}
\end{calign}
\vspace{-20pt}
\caption{Biunitary characterizations of quantum structures.\label{fig:biunitarycharacterizations}}
\vspace{-15pt}
\end{figure}
\addtocounter{footnote}{-2}\footnotetext{A (complex) \textit{Hadamard matrix} is a square complex matrix with entries of modulus 1, which is proportional to a unitary matrix. Fundamental structures in quantum information, they are central in the theories of mutually unbiased bases, quantum key distribution, and other phenomena~\cite{Durt:2010}.}\stepcounter{footnote}\footnotetext{A \textit{unitary error basis} is a basis of unitary operators on a finite-dimensional Hilbert space, orthogonal with respect to the trace inner product. They provide the basic data for quantum teleportation, dense coding and error correction procedures~\cite{Werner:2001, Knill:1996_2, Shor:1996}.}\stepcounter{footnote}\footnotetext{A \textit{quantum Latin square}~\cite{Musto:2015} is a square grid of vectors in a finite-dimensional Hilbert space, such that every row and every column is an orthonormal basis. They are quantum generalizations of classical Latin squares.}%

Our main results are based on the simple fact that the \textit{diagonal} composite of two biunitaries is again biunitary. We show in \autoref{sec:composition} that, given the description of quantum combinatorial structures in terms of biunitaries as summarized above, one can  immediately write down a large number of schemes for the construction of certain quantum structures from others. We give some examples in \autoref{fig:introbinary}; note that the biunitaries are connected diagonally in each case, as required.
\begin{minipage}{\linewidth}
\begin{calign}
\nonumber
\begin{tz}[xscale=-1,string, scale=1.5]
\node (1) [blob] at (1,1) {$H_2$};
\node (2) [blob] at (1.5,1.5) {$H_1$};
\path [blueregion,  string,bnd] (0.25,0) to [out=90, in=-135] (1,1) to [out=-45, in=90] (1.75,0);
\path [blueregion, string] (\hrt+1.75,0) to [out=90, in=-45] (1.5,1.5) to [out=45, in=-90] (2.25,2.5) to (2.25+\side,2.5) to (2.25+\side,0);
\draw [string,bnd] (0.5 *2^0.5+1.75,0) to [out=90, in=-45] (1.5,1.5) to [out=45, in=-90] (2.25,2.5);
\path [blueregion, bnd] (0.75,2.5) to [out=-90, in=135] (1.5,1.5) to (1,1) to [out=135, in=-90] (0.75-0.5*2^0.5, 2.5);
\end{tz}
&
\begin{tz}[xscale=-1,string, scale=1.5]
\node (1) [blob]at (1,1) {$U_2$};
\node (2) [blob] at (1.5,1.5) {$U_1$};
\path [blueregion] (\hrt+1.75,0) to [out=90, in=-45] (1.5,1.5) to [out=45, in=-90] (2.25,2.5) to (2.25+\side,2.5) to (2.25+\side,0);
\draw [string,bnd] (\hrt+1.75,0) to [out=90, in=-45] (1.5,1.5) to [out=45, in=-90] (2.25,2.5);
\path [blueregion,  string,bnd] (0.25,0) to [out=90, in=-135] (1,1) to [out=-45, in=90] (1.75,0);
\draw [string] (0.75,2.5) to [out=-90, in=135] (1.5,1.5) to (1,1) to [out=135, in=-90] (0.75-0.5* 2^0.5, 2.5);
\end{tz}  
\\*[-1pt]\nonumber
\text{(a) }\HAD + \HAD{\,\leadsto\,}\QLS
&
\text{(b) }\UEB + \UEB{\,\leadsto\,}\QLS
\\[5pt]\nonumber
\begin{tz}[string, scale=1.5]
\node (1) [blob] at (1,1) {$H$};
\node (2) [blob] at (1.5,1.5) {$Q$};
\path [blueregion] (0.75,2.6) to [out=-90, in=135] (2.center) to [out=-155, in=65] (1.center) to [out=-125, in=90] (0.35,0.1) to (0.35-\side,0.1) to (0.35-\side,2.6);
\draw [string,bnd]  (0.75,2.6) to [out=-90, in=135] (2.center) to [out=-155, in=65] (1.center) to [out=-125, in=90] (0.35,0.1);
\path [yellowregion] (0.15,-0.1) to [out=90, in=-145] (1.center) to [out=135, in=-90] (0.75-\hrt, 2.4) to (0.15-\side,2.4) to (0.15-\side,-0.1);
\draw [string,bnd] (0.15,-0.1) to [out=90, in=-145] (1.center) to [out=135, in=-90] (0.75-\hrt, 2.4);
\path [yellowregion, string,bnd] (1.75,0) to [out=90, in=-45] (1,1) to [out=25, in=-115] (1.5,1.5) to [out=-45, in=90]  (\hrt+1.75,0);
\draw [string] (1.5,1.5) to [out=45, in=-90] (2.25,2.5);
\end{tz}
&
\begin{tz}[string, scale=1.5, xscale=1]
\node (1) [blob] at (1,1) {$H_1$};
\node (2) [blob] at (1.5,1.5) {$H_2$};
\path [blueregion] (0.75,2.6) to [out=-90, in=135] (1.5,1.5) to [out=-155, in=65] (1,1) to [out=-125, in=90] (0.35,0.1) to (0.15-\side,0.1) to (0.15-\side,2.6);
\draw [string,bnd] (0.75,2.6) to [out=-90, in=135] (1.5,1.5) to [out=-155, in=65] (1,1) to [out=-125, in=90] (0.35,0.1) node {};
\path [yellowregion] (0.15,-0.1) node {} to [out=90, in=-145] (1,1) to [out=135, in=-90] (0.75-\hrt,2.4) to (0.65-2*\hrt,2.4) to (0.65-2*\hrt,-0.1);
\draw [string,bnd] (0.15,-0.1) to [out=90, in=-145] (1,1) to [out=135, in=-90] (0.75-\hrt, 2.4);
\path [blueregion] (2.35,2.6) to [out=-90, in=35] (1.5,1.5) to [out=-45, in=90] (1.75+\hrt,0.1) to (1.75+\hrt+ \side,0.1) to (1.75+\hrt+\side,2.6);
\draw [string,bnd] (2.35,2.6) to [out=-90, in=35] (1.5,1.5) to [out=-45, in=90] (1.75+\hrt,0.1) ;
\path [yellowregion] (1.75,-0.1) to [out=90, in=-45] (1,1) to [out=25, in=-115] (1.5,1.5) to [out=55, in=-90] (2.15,2.4) to (1.55+\hrt+\side,2.4) to (1.55+\hrt+\side,-0.1);
\draw [string,bnd] (1.75,-0.1) to [out=90, in=-45] (1,1) to [out=25, in=-115] (1.5,1.5) to [out=55, in=-90] (2.15,2.4) ;
\end{tz}
\\*[-1pt]\nonumber
\text{(c) }\HAD^* + \QLS {\,\leadsto\,} \UEB
&
\text{(d) }\HAD^* + \HAD^* {\,\leadsto\,} \HAD
\end{calign}
\vspace{-20pt}

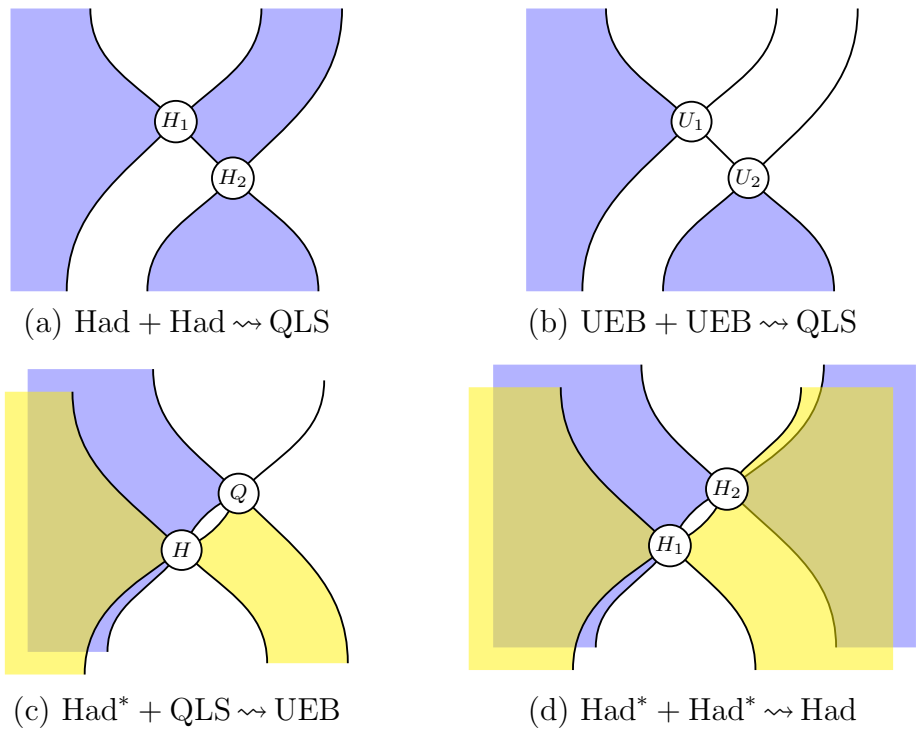
\captionof{figure}{Some biunitary composites of arity 2.}\label{fig:introbinary}%
\vspace{-2pt}
\end{minipage}

\vspace{10pt}
\noindent
We explain each of these constructions briefly. \autoref{fig:introbinary}(a) gives a way to combine two Hadamard matrices to produce a quantum Latin square, generalizing a known construction.\footnote{When both Hadamard matrices are the same, this agrees with a known construction of a quantum Latin square from a single Hadamard matrix~\cite[Definition 10]{Musto:2015}.} (Note that the wires terminating near the upper-right of \autoref{fig:introbinary}(a) are interpreted as a \textit{single} composite wire for the purpose of identifying it as having the basic quantum Latin square type of \autoref{fig:biunitarycharacterizations}, a method that we use repeatedly, and motivate formally with bracketings in \autoref{thm:composition}.) \autoref{fig:introbinary}(b) describes a procedure for combining two unitary error bases to yield a quantum Latin square, a construction we believe to be new. \autoref{fig:introbinary}(c) combines a controlled family of Hadamard matrices and a quantum Latin square to give a unitary error basis, recovering the {quantum shift-and-multiply} construction~\cite[Definition 18]{Musto:2015}. In \autoref{fig:introbinary}(d), two families of Hadamard matrices are combined to produce a single Hadamard matrix, recovering a construction of Hosoya and Suzuki \cite[Section 1]{Hosoya:2003} which generalizes a construction of \Dita \cite[Section 4]{Dita:2004}. These constructions can of course be iterated; for example, combining the constructions of Figures~\ref{fig:introbinary}(a)~and~\ref{fig:introbinary}(c) gives a way to combine a controlled family of Hadamard matrices and two further Hadamard matrices to produce a single unitary error basis, again a new construction.

In all these cases, correctness of the construction follows immediately from the type-theoretic structure (that is, the shading pattern) of the diagram, relying only on diagonality of the composition; no further details need to be checked. Our approach therefore offers advantages even for those constructions that are already known, since the traditional proofs of correctness are nontrivial. To emphasize this point we compare our graphical techniques to traditional methods, in which constructions are defined using tensor notation. For example, the construction of \autoref{fig:introbinary}(c) would traditionally be written as follows~\cite[Definition 18]{Musto:2015}, where $U_{ab,c,d}$ is the $(c,d)$th matrix entry of the $(a,b)\text{th}$ element of the unitary error basis, $Q_{b,d,c}$ is the coefficent of $\ket c$ in the $(b,d)\text{th}$ position of the quantum Latin square, and $H^b _{a,d}$ is the $(a,d)\text{th}$ coefficient of the $b$th Hadamard matrix:
\begin{equation} \label{eq:qsm}U_{ab,c,d} := H^b_{a,d}\, Q_{b,d,c}^{}
\end{equation}
It is not trivial to write down correct expressions of this form, and to show that this indeed defines a unitary error basis requires a calculation of several lines~\cite[Theorem 20]{Musto:2015} that invokes the distinct algebraic properties of the tensors $Q_{b,d,c}$ and $H^b_{a,d}$. In contrast, in our new approach, it would be easy to discover this construction by considering all ways the basic components can be diagonally composed; correctness is immediate, and all algebraic properties are subsumed by the single concept of biunitarity. Nonetheless, expression~\eqref{eq:qsm} can be immediately read off from the form of the biunitary composite.

Higher-arity constructions can also be described, such as those given in \autoref{fig:introhigher}. Both of these we believe to be new.
\begin{figure}[t!]
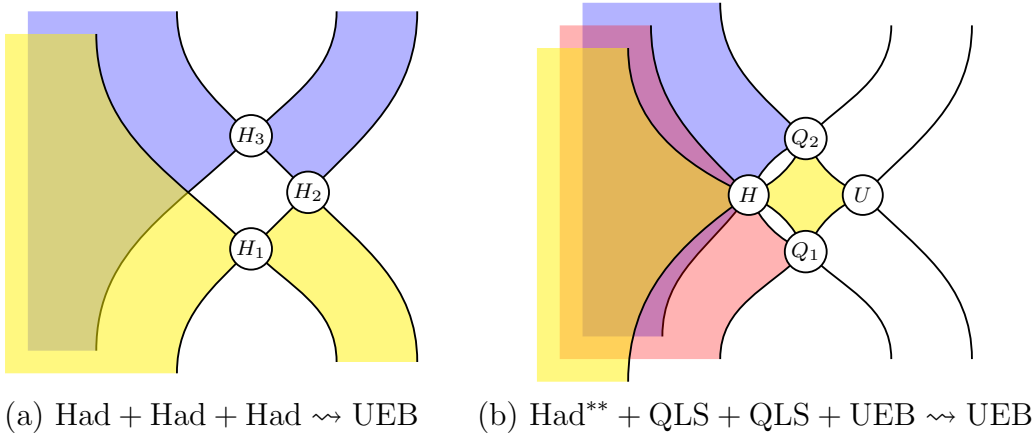

\begin{calign}
\nonumber
\def\sidew {0.6}
\begin{tz}[string,scale=1.5]
\node (1) [blob] at (1.75,1.25) {$H_1$};
\node (2) [blob] at (1.75,2.25) {$H_3$};
\node (3) [blob] at (2.25,1.75) {$H_2$};
\path [blueregion,bnd] (2.5,3.35) to [out=-90, in=45] (2.center) to (3.center) to [out=45, in=-90] (2.5+\hrt,3.35);
\path [blueregion] (1.1,3.35) to [out=-90, in=135] (2.center) to [out=-135, in=up] (1.1-\hrt,0.35) to (1.1-\hrt-\sidew,0.35) to (1.1-\hrt-\sidew,3.35);
\draw [bnd]  (1.1,3.35) to [out=-90, in=135] (2.center) to [out=-135, in=up] (1.1-\hrt,0.35);
\path [yellowregion, bnd] (2.5,0.25) to [out=90, in=-45] (1.center) to (3.center)  to [out=-45, in=90] (2.5+\hrt,0.25);
\path [yellowregion] (1.1,0.15) to [out=90, in=-135] (1.center) to [out=135, in=-90] (1.1-\hrt,3.15) to (0.9-\hrt-\sidew,3.15) to (0.9-\hrt-\sidew,0.15);
\draw [bnd] (1.1,0.15) to [out=90, in=-135] (1.center) to [out=135, in=-90] (1.1-\hrt,3.15);
\end{tz}
&
\def \sidew {0.7}
\begin{tz}[string, scale=1.5]
\node (Q) [blob] at (1,1) {$Q_1$};
\node (V) [blob]at (1.5,1.5) {$U$};
\node (H) [blob] at (0.5,1.5) {$H$};
\node (P) [blob] at (1,2) {$Q_2$};
\def\leftb{0.45-\hrt-\sidew}
\path [blueregion] (0.25,3.2) to [out=-90, in=135] (P.center) to [out=-155, in=65] (0.5,1.5) to [out=-125, in=90, looseness=1] (0.45-\hrt, 0.25) to (\leftb,0.25) to (\leftb,3.2);
\draw [bnd] (0.25,3.2) to [out=-90, in=135] (1,2) to [out=-155, in=65] (0.5,1.5) to [out=-125, in=90, looseness=1] (0.45-\hrt,0.25);
\path [redregion] (0.25,0.05) to [out=90, in=-135] (Q.center) to [out=155, in=-65] (H.center) to [out=135, in=-90] (0.35-\hrt,3.0) to (\leftb-0.2,3.0) to (\leftb-0.2,0.05);
\draw [bnd] (0.25,0.05) to [out=90, in=-135] (Q.center) to [out=155, in=-65] (H.center) to [out=135, in=-90] (0.35-\hrt,3.0);
\path [yellowregion] (0.15-\hrt, 2.8) to [out=-90, in=155] (H.center) to [out=-145, in=90] (0.15-\hrt, -0.15) to (\leftb-0.4,-0.15) to (\leftb-0.4, 2.8);
\draw [bnd] (0.15-\hrt, 2.8) to [out=-90, in=155] (0.5,1.5) to [out=-145, in=90] (0.15-\hrt, -0.15);
\path [yellowregion, bnd] (H.center) to [out=25, in=-115] (P.center) to [out=-65, in=155] (V.center) to [out=-155, in=65] (Q.center) to [out=115, in=-25] (H.center);
\draw [string] (1,1) to [out=-45, in=90] (1.75,0.05);
\draw [string] (1.75+\hrt,0.05) to [out=90, in=-45] (1.5,1.5) to [out=45, in=-90] (1.75+\hrt,3);
\draw [string] (1,2) to [out=45, in=-90] (1.75,3);
\end{tz}
\\*\nonumber
\text{(a) }\HAD + \HAD + \HAD \leadsto \UEB
&
\text{(b) }\HAD^{**} + \QLS + \QLS + \UEB \leadsto \UEB
\end{calign}
\vspace{-20pt}
\caption{Some biunitary composites of arities 3 and 4.\label{fig:introhigher}}
\end{figure}
In \autoref{fig:introhigher}(a), arising as a consequence of the constructions of Figures~\ref{fig:introbinary}(a) and \ref{fig:introbinary}(c), three Hadamard matrices combine to produce a unitary error basis, an elegant construction which we believe to be new.\footnote{When all three Hadamard matrices are the same, this agrees with a known construction of a unitary error basis from a single Hadamard matrix~\cite[Definition 33]{Musto:2015} which we believe to be folklore.} In \autoref{fig:introhigher}(b), which does not arise as a consequence of lower-arity constructions, we combine a double-controlled family of Hadamard {matrices~$(H)$,} two quantum Latin squares $(Q_1,Q_2)$ and a unitary error basis $(U)$ to produce a new unitary error basis. While the first example is simple and elegant, the second example is indicative of the more complex constructions our technique can produce. Further complex examples are given in Figures~\ref{fig:ternary}, \ref{fig:higher} and~\ref{fig:infinitude}.

For unitary error bases, we illustrate all constructions of arities 2 and 3 that arise from our methods, and we give examples of constructions of arities~4 and~8.  Furthermore, in \autoref{sec:infinity} we show that our methods give rise to an infinite family of logically independent constructions, none of which factor through any simpler construction between Hadamard matrices, unitary error bases, quantum Latin squares and controlled families thereof.\footnote{In a subsequent article we analyze the space of all such UEB constructions, and unify them in terms of a single binary composite involving new quantum combinatorial structures.}

In \autoref{sec:equivalence} we consider the problem of equivalence from our new perspective. Hadamard matrices, unitary error bases, quantum Latin squares and controlled families all come with standard notions of equivalence. We give a new generic definition of equivalence for biunitaries, broader than the one used traditionally in the planar algebra literature, and show that it recovers precisely these usual notions of equivalence for each of the quantum structures we consider.

Finally, in \autoref{sec:newUEB} we use the 4-fold composite of \autoref{fig:introhigher}(b) to produce a unitary error basis on an 8\-dimensional Hilbert space. We show that it cannot be produced by the two known UEB construction methods---algebraic, and quantum shift-and-multiply---even up to equivalence. This is a proof of principle that the biunitary methods we propose can give rise to genuinely new quantum structures.
\paragraph{Significance.}
Hadamard matrices and unitary error bases provide the mathematical foundation for an extremely rich variety of quantum computational phenomena, amongst them the study of mutually unbiased bases, quantum key distribution, quantum teleportation, dense coding and quantum error correction~\cite{Durt:2010,Werner:2001,Shor:1996,Klappenecker:2003,Knill:1996_2}. Nevertheless their general structure is notoriously difficult to understand; 
in dimension $n$, Hadamard matrices have only been classified up to \mbox{$n=5$}~\cite{Szollosi:2011,Tadej:2006}, and the general structure of unitary error bases is virtually unknown for $n > 2$. Quantum Latin squares have been introduced much more recently~\cite{Musto:2015, Banica:2007, Musto:2016}, generalizing classical Latin squares which have a wide range of applications in classical and quantum information~\cite{Shannon:1949,Meyer:2015,Benoist:2016}.

By unifying these quantum structures as special cases of the single notion of biunitary, and providing simple type-theoretical  tools to understand the intricate interplay between them, we unify several already-known and seemingly-unrelated constructions~\cite{Musto:2015,Werner:2001,Hosoya:2003,Dita:2004,Banica:2007}, uncover an infinite number of new constructions, and produce novel, concrete examples. These new tools may lead to further progress in questions of classification and applications of Hadamard matrices, unitary error bases and quantum Latin squares, and perhaps move us closer to full classification results for these important structures. 

On the other hand, biunitaries are central tools in the study and classification of subfactors~\cite{Jones:1999,Jones:2013,Morrison:2014,Ocneanu:1989, Popa:1983}, a highly significant activity in the theory of von Neumann algebras. We hope that our work leads to the development of further connections between subfactor theory and quantum computation.

\subsection{Related work}

\paragraph{Quantum constructions.}As well as producing a number of new constructions, our methods encompass and unify several constructions from the literature.

\begin{itemize}
\item The \textit{Hadamard method}~\cite[Definition 33]{Musto:2015}, believed to be folklore, which produces a unitary error basis from a single Hadamard matrix. In {expression~\eqref{eq:HadamardMethod}} we give a biunitary presentation of a new generalization, in which three Hadamard matrices produce a unitary error basis.
\item The method given in \cite[Definition 2.3]{Banica:2007} and~\cite[Definition 10]{Musto:2015}, which produces a quantum Latin square from a single Hadamard matrix. {In \autoref{fig:binarydetailed1}(a) we give a biunitary presentation of a new generalization, in which two Hadamard matrices produce a quantum Latin square.}
\item Werner's \textit{shift-and-multiply construction}~\cite{Werner:2001} which produces a unitary error basis from a family of Hadamard matrices and a Latin square. This is a special case of the quantum shift-and-multiply construction discussed below.
\item The \textit{quantum shift-and-multiply construction} due to Musto and the second author~\cite[Definition 18]{Musto:2015} which produces a unitary error basis from a family of Hadamard matrices and a quantum Latin square. {We give a biunitary description in \autoref{fig:binarydetailed2}(b).}
\item \textit{\Dita's construction}~\cite[Section 4]{Dita:2004}, which produces a Hadamard matrix from a Hadamard matrix and a family of Hadamard matrices. This method sees wide use in the literature~\cite{Dita:2004,Tadej:2006,Matolcsi:2007,Nicoara:2010}, and is a special case of Hosoya's and Suzuki's construction, described below. {We give a biunitary description in \autoref{fig:binarydetailed1}(d).}
\item \textit{Hosoya's and Suzuki's construction}~\cite[Section 1]{Hosoya:2003}, which produces a Hadamard matrix from two families of Hadamard matrices. {We give a biunitary description in \autoref{fig:binarydetailed1}(c).}
\end{itemize}
There are many known constructions which are beyond our methods. For unitary error bases, we do not know a biunitary characterization of Knill's algebraic construction~\cite{Knill:1996}. For Hadamard matrices, an analogue of Knill's construction are the Fourier matrices arising from finite abelian groups. Other examples include Petrescu's construction of continuous families of Hadamard matrices in prime dimension~\cite{Petrescu:1997}, Wocjan's and Beth's construction~\cite{Wocjan:2004} and its generalization by Musto~\cite{Musto:2016}, or several other less-general constructions which only work in specific dimensions~\cite{Haagerup:1996, Szollosi:2012, Szollosi:2011, Matolcsi:2007}. In all of these cases, the methods are not purely compositional; they make use of some additional group-theoretic or algebraic structure which is out of reach of the biunitary approach.

\paragraph{Biunitary connections and planar algebras.} 
Biunitary connections were introduced by Ocneanu in 1989 in terms of \textit{paragroups} \cite{Ocneanu:1989} in an attempt to better understand the rich combinatorial structures arising in subfactor theory, a branch of the theory of von Neumann algebras.
In 1999, Jones introduced the theory of planar algebras \cite{Jones:1999} and with it the modern graphical formulation of biunitarity as used in this paper. 

Recently, Jaffe, Liu and Wozniakowski have described a related planar algebraic approach to quantum information based on \textit{planar para algebras}~\cite{Jaffe:2016d, Jaffe:2016a,Jaffe:2017}.

The relation between Hadamard matrices and von Neumann algebras predates the notion of biunitarity and can be traced back to Popa's commuting squares~\cite{Popa:1983} (later shown to be equivalent to biunitarity~\cite{JonesSunder:1997}) and the statistical-mechanical spin models of Jones~\cite{Jones:1989,Jones:1999}; there is a significant literature on the interplay between Hadamard matrices, planar algebras and subfactors~\cite{Haagerup:1996,Nicoara:2006,Nicoara:2010}. Quantum Latin squares appear under the name \textit{magic bases} in a Hopf algebraic approach to subfactor theory by Banica and others \cite{Banica:2007,Banica:2007_2,Banica:2009}; however, their biunitary characterization does not seem to have been written down. 
Unitary error bases were shown to be characterized in terms of biunitaries by the second author~\cite{Vicary:2012hq, Vicary:2012}.

Traditionally, biunitary connections are used in the classification of amenable subfactors of the hyperfinite $\mathrm{II}_1$ factor \cite{JonesSunder:1997}.
We give a brief overview of this research programme.
Every such subfactor induces a pair of principal graphs, and a flat biunitary connection in the associated graph planar algebra. Conversely, every biunitary connection induces a planar algebra of flat elements, and the original subfactor can always be recovered from this planar algebra~\cite{Evans:1998}. Therefore, to classify subfactors it is sufficient to classify possible principal graphs and flat connections in their graph planar algebras. This programme has lead to a classification of subfactors up to index $5+\frac{1}{4}$~\cite{Jones:2013, Morrison:2014,Izumi:2015, Afzaly:2015}.


\paragraph{Categorical quantum mechanics.} 
This work builds on the programme of categorical quantum mechanics, initiated by Abramsky and Coecke~\cite{Abramsky:2004} and developed by them and others~\cite{Abramsky:2009, Coecke:2006, Coecke:2008, Coecke:2017, Kissinger:2015, Coecke:2014, Coecke:2012, Backens:2014, Selinger:2007}, which uses monoidal categories with duals to provide a high-level syntax for quantum information flow. It was shown by the second author that these ideas can be extended to a higher categorical setting~\cite{Vicary:2012, Vicary:2012hq}, developing the work of Baez on a categorified notion of Hilbert space~\cite{Baez:1997}. The key advantage of this approach is that the notion of Frobenius algebra, used in the monoidal category setting to describe classical information, is no longer needed and replaced by the simpler notion of dualizable 1\-morphism in \cat{2Hilb}. This is the setting employed in this paper. While the results we prove could in principle be translated back into the language of Frobenius algebras, they would lose their simplicity and power. In this sense, the current work serves as an advertisement for the essential role that higher category theory can play in quantum information theory. 

\subsection{Notations and conventions}

We denote the $n$\-element set $\{1,\ldots,n\}$ by $[n]$. The letters $a,b,d,e,f,g,h,i,j,k,r,s$ are used to denote indices, the letters $n,m,p,q$ are used to denote dimensions. We use the following shorthand notations to refer to sets of quantum structures:
\begin{itemize}
\item $\UEB_n$ is the set of $n$-dimensional unitary error bases;
\item $\QLS_n$ is the set of $n$-dimensional quantum Latin squares;
\item $\HAD_n$ is the set of $n$-dimensional Hadamard matrices;
\item For $\mathrm X \in \{ \UEB_n, \QLS_n, \HAD_n \}$, $\mathrm X^{p_1,\ldots, p_k}$ is the set of lists of quantum structures of type $X$ controlled by indices in $[p_1],~[p_2],~\ldots,[p_k]$.
\end{itemize}

\noindent
For example, $\UEB_{n^2 m} ^{n,p}$ is the set of lists of $n^2 m$-dimensional unitary error bases, controlled by indices taking values in $[n]$ and $[p]$.  \\

\subsection*{Acknowledgements}

We are grateful to Bruce Bartlett, Andr\'e Henriques, Scott Morrison, Benjamin Musto and Dominic Verdon for useful discussions. Early versions of these results were presented at WIP 2016, QIP 2017, CALCO 2017 and QPL 2017, and we are grateful to the organizers of these events, and to the audience for useful feedback. 

\section{Biunitarity}
\label{sec:biunitarity}

In \autoref{sec:foundations} we introduce our formalism, and give the definition of biunitarity. In \autoref{sec:characterizing} we recall the biunitary characterizations of Hadamard matrices and unitary error bases, and give new biunitary characterizations of quantum Latin squares and controlled families.
\subsection{Mathematical foundations}
\label{sec:foundations}

The graphical calculus for describing composition of multilinear maps was proposed by Penrose~\cite{Penrose:1971}, and is today widely used in a range of areas~\cite{Selinger:2010, Abramsky:2009, Coecke:2006, Joyal:1991, Orus:2014}. In this scheme, wires represent Hilbert spaces and vertices represent linear maps between them, with wiring diagrams representing composite linear maps. For example, given linear maps $A:W\otimes H \otimes J\to L \otimes M\otimes R$ and $B:V \to H \otimes J$, we can describe a composite linear map $V\otimes W \to L \otimes M\otimes R$ graphically as follows:
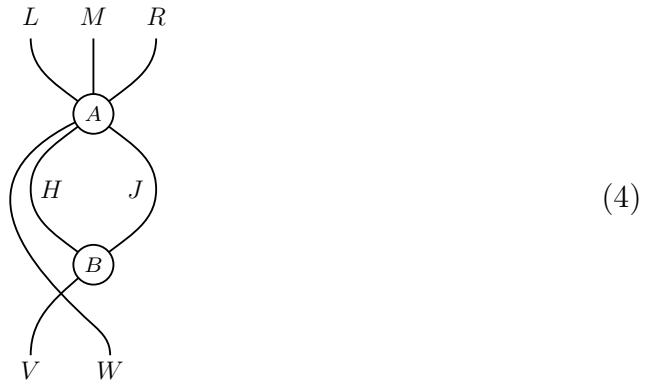
\begin{equation}
\label{eq:penroseexample}
\begin{tz}[string,xscale=1.1]
\path[draw](1,-1) to [out= 135, in = down] (0.25,0) to [out= up, in = -135] (1,1) to [out = -45, in = up] (1.75,0) to [out= down, in = 45] (1,-1);
\draw (1,1) to (1,2);
\draw (1,-1) to[out= -135, in =up] (0.25,-2.2);
\draw  (0.25,2) to [out= down, in= 135] (1,1) to [out =-155, in = 135,looseness=1.5] (1,-1.8) to [out= -45, in =up] (1.2,-2.2) ;
\draw (1,1) to [out = 45, in =down] (1.75,2);
\node[blob] at (1,1) {$A$};
\node[blob] at (1,-1) {$B$};
\node[scale=0.8] at (1.75,2.3) {$R$};
\node[scale=0.8] at (0.5,0) {$H$};
\node[scale=0.8] at (1.5, 0) {$J$};
\node[scale=0.8] at (1, 2.3) {$M$};
\node[scale=0.8] at (0.25, -2.4) {$V$};
\node[scale=0.8] at (1.2,-2.4) {$W$};
\node[scale=0.8] at (0.25,2.3) {$L$};
\end{tz}
\end{equation}
\noindent

In this article we use a generalized calculus that involves \emph{regions}, as well as wires and vertices. This is an instance of the graphical calculus for monoidal 2\-categories~\cite{Barrett:2012, Bartlett:2014, Hummon:2012, SchommerPries:2011} applied to the 2\-category\footnote{Here and throughout, we use the term `2-category' to refer to the fully weak structure, which is sometimes called `bicategory'.} of finite-dimensional 2\-Hilbert spaces~\cite{Baez:1997}. The 2\-category of 2\-Hilbert spaces can be described as follows~\cite{Elgueta:2007,Vicary:2012hq}:
\begin{itemize}
  \item objects are natural numbers $n,m,\ldots$;
  \item 1-morphisms $n\to m$ are $m\times n$-matrices of finite-dimensional Hilbert spaces (\autoref{fig:2Hilb}(a));
  \item 2-morphisms are matrices of linear maps (\autoref{fig:2Hilb}(b)).
\end{itemize}
\begin{figure}[]\figuretopsuck 
\begin{calign}
\nonumber
\left(\begin{array}{ccc} H_{11} & \cdots & H_{1n}\\\vdots & \ddots & \vdots \\ H_{m1} &\cdots & H_{mn}\end{array}\right)
& 
 \left(\begin{array}{ccc} H_{11}\to[\phi_{11}]H'_{11} & \ldots & H_{1n} \to[\phi_{1n}] H'_{1n}\\\vdots & \ddots & \vdots\\H_{m1} \to[\phi_{m1}]H'_{m1}  &\ldots& H_{mn} \to[\phi_{mn}] H'_{mn} \end{array} \right)
 \\[0pt]\nonumber
 \text{(a) A 1-morphism $H:n\to m$} &\text{(b) A 2-morphism $\phi: H \To[] H'$}
\end{calign}

\figurecaptionsuck 
\caption{The 1- and 2-morphisms in the 2-category of 2-Hilbert spaces.}
\label{fig:2Hilb}%
\figurecaptionpostsuck
\end{figure}
Composition of 1\-morphisms is given by `matrix multiplication' of matrices of Hilbert spaces, with addition and multiplication of complex numbers replaced by direct sum and tensor product, respectively. Composition of 2-morphisms is given by componentwise composition of linear maps.

\begin{figure}[b]
\figuretopsuck
\figuretopsuck
\begin{calign}
\def\spacing{15pt}
\nonumber
 \vc{
$\phi: \begin{pmatrix}J_1\\J_2\end{pmatrix} \To  \begin{pmatrix} H_{11} & H_{12} & H_{13} \\ H_{21} & H_{22} & H_{23} \end{pmatrix} \circ \begin{pmatrix}K_1 \\ K_2 \\ K_3 \end{pmatrix}$
\\[-2pt]
\text{(a) A 2-morphism $\phi$}\\[\spacing]
$\begin{pmatrix} J_1 \to[\phi_1] (H_{11}\otimes K_1) \oplus (H_{12}\otimes K_2) \oplus (H_{13} \otimes K_3) \\ J_2 \to [\phi_2] (H_{21}\otimes K_1) \oplus (H_{22} \otimes K_2) \oplus (H_{23} \otimes K_3) \end{pmatrix}$\\[1pt] 
\text{(b) The 2-morphism $\phi$ as a matrix of linear maps}\\[\spacing]
$\phi_{i,j} : J_i \to H_{i,j} \otimes K_j$\hspace{0.2cm} \text{ for $i \in [2]$ and $j \in [3]$} \\[1pt]
\vc{(c) The 2-morphism $\phi$ as a family of linear\\maps indexed by its adjacent regions}}
&
\vc{
$\begin{tz}[string, scale=1.35, xscale=1.45]
\path[blueregion] (-1,0) to (-1,2) to (-0.5,2) to [out=down, in=135] (0,1) to (0,0) to (-1,0);
\path[redregion] (-0.5,2) to [out=down, in=135] (0,1) to [out=45, in=down] (0.5,2);
\draw (0,0) to (0,1) to [out=135, in=down] (-0.5, 2);
\draw[string] (0.5,2) to [out=down, in=45] (0,1);
\node[blob] at (0,1) {$\phi$};
\node[scale=0.8] at (0,-0.4) {$\textstyle\begin{pmatrix} J_1 \\ J_2\end{pmatrix}$};
\node[scale=0.8] at (-0.5, 2.4) {$\textstyle \begin{pmatrix} H_{11} & H_{12} & H_{13} \\ H_{21} & H_{22} & H_{23} \end{pmatrix}$};
\node[scale=0.8] at (0.5,2.4) {$\textstyle \begin{pmatrix}K_1 \\ K_2 \\ K_3 \end{pmatrix}$};
\node[scale=0.9] at (-0.5,0.75) {$2$};
\node[scale=0.9] at (0, 1.75) {$3$};
\end{tz}$
\\[-5pt]
\vc{(d) Graphical representation of the\\2-morphism $\phi$}}
\end{calign}

\figurecaptionsuck 
\caption{Translating between equivalent expressions for 2-morphisms.}
\label{fig:equivalentexpression}%
\figurecaptionpostsuck
\end{figure}
In the graphical calculus, regions, wires and vertices represent objects, 1\-morphisms and 2\-morphisms, respectively. The 2\-category has a monoidal structure, acting on objects as multiplication, and on 1\- and 2\-morphisms as the Kronecker product of matrices of Hilbert spaces and linear maps, respectively; this is represented graphically by `layering' one diagram above another.

\paragraph{Elementary description.}
To help the reader understand these concepts, we also give a direct account of the formalism  in elementary terms, that can be used without reference to the higher categorical technology. In \autoref{fig:equivalentexpression} we indicate how to translate between the categorical language presented above and the more elementary language used here.

In this direct perspective,  shaded regions are labelled by \textit{finite sets}, indexed by a parameter; we write $i{:}n$ to indicate that the parameter $i$ varies over the set $[n]$.\footnote{For simplicity we will often omit these labels.} Wires and vertices now represent \textit{families} of Hilbert spaces and linear maps respectively, indexed by the parameters of all adjoining regions.
A composite surface diagram represents a family of composite linear maps, indexed by the parameters of all open regions, with closed regions being summed over. This is illustrated by the following example:

\begin{equation}
\label{eq:translation}
\begin{tz}[string]
\clip (-1,-2.6) rectangle (1.9, 2.6);
\path[blueregion] (-0.8,-2) to(0.25,-2) to[out = up, in = -135] (1,-1) to [out= 135, in = down]  (0.25,0) to [out= up, in = -135] (1,1) to (1,2) to (-0.8,2);
\path[redregion, draw](1,-1) to [out= 135, in = down] (0.25,0) to [out= up, in = -135] (1,1) to [out = -45, in = up] (1.75,0) to [out= down, in = 45] (1,-1);
\draw (1,1) to (1,2);
\draw (1,-1) to[out= -135, in =up] (0.25,-2);
\draw (1,1) to [out = 45, in =down] (1.75,2);
\path[yellowregion](0.25,1.8) to [out= down, in= 135] (1,1) to [out =-155, in = 135,looseness=1.5] (1,-1.6) to [out= -45, in =up] (1.4,-2.2) to (-1, -2.2) to (-1,1.8) to (0.25,1.8);
\draw  (0.25,1.8) to [out= down, in= 135] (1,1) to [out =-155, in = 135,looseness=1.5] (1,-1.6) to [out= -45, in =up] (1.4,-2.2) ;
\node[blob] at (1,1) {$A$};
\node[blob] at (1,-1) {$B$};
\node[scale=0.8] at (1.75,2.2) {$R$};
\node[scale=0.8] at (0.5,0) {$H$};
\node[scale=0.8] at (1.5, 0) {$J$};
\node[scale=0.8] at (1, 2.2) {$M$};
\node[scale=0.8] at (0.25, -2.4) {$V$};
\node[scale=0.8] at (1.4,-2.4) {$W$};
\node[scale=0.8] at (0.25,2.2) {$L$};
\node[dimension, below left] at (0.9,1.95) {$i{:}n$};
\node[dimension] at (1,0) {$j{:}p$};
\node[dimension,below left] at (0.25,1.75) {$k{:}m$};
\end{tz}
\hspace{25pt}\leftrightsquigarrow\hspace{25pt}\sum_{j\in [p]}\hspace{-25pt}
\begin{tz}[string,xscale=1.1]
\clip (-1,-2.6) rectangle (1.9, 2.6);
\path[draw](1,-1) to [out= 135, in = down] (0.25,0) to [out= up, in = -135] (1,1) to [out = -45, in = up] (1.75,0) to [out= down, in = 45] (1,-1);
\draw (1,1) to (1,2);
\draw (1,-1) to[out= -135, in =up] (0.25,-2.2);
\draw  (0.25,2) to [out= down, in= 135] (1,1) to [out =-155, in = 135,looseness=1.5] (1,-1.8) to [out= -45, in =up] (1.2,-2.2) ;
\draw (1,1) to [out = 45, in =down] (1.75,2);
\node[blob,minimum width=21pt] at (1,1) {$A_{ijk}$};
\node[blob,minimum width=21pt] at (1,-1) {$B_{ij}$};
\node[scale=0.8] at (1.75,2.3) {$R$};
\node[scale=0.8] at (0.5,0) {$H_{ij}$};
\node[scale=0.8] at (1.5, 0) {$J_{j}$};
\node[scale=0.8] at (1, 2.3) {$M_{i}$};
\node[scale=0.8] at (0.25, -2.4) {$V_{i}$};
\node[scale=0.8] at (1.2,-2.4) {$W_k$};
\node[scale=0.8] at (0.25,2.3) {$L_k$};
\end{tz}
\end{equation}The diagram on the left represents an entire family of composite linear maps. The maps which comprise this family are  given by the right-hand diagrams for different values of $k$ and $i$, which index the open regions.  The closed region labelled $j:p$ is summed over.
 
 In some situations, particularly when dealing with equations of shaded diagrams with different connectivity, we may need to have multiple parameters $i:n$, $i':n$ labelling the same region. In this case, we need an auxiliary rule that says the corresponding linear map is zero when $i \neq i'$.
 
In this paper we will not encode any information in the colour of a region, with colour used only as a way to bring out the geometry of the diagram. Regions which are drawn in different colours may represent the same object, and regions which are the same colour may represent different objects. At all times, when it matters, we indicate the object being encoded with the parameter label (like $i:n$).

Given this interpretation of diagrams $D$ as families of linear maps $D_i$, we define two diagrams $D,D'$ to be equal when all the corresponding linear maps $D_i,D_i'$ are equal, and the scalar product $\lambda D$ as the family of linear maps $\lambda D_i$.

\paragraph{Duality.}We define the linear maps \mbox{$\eta:\mathbb{C} \to \C^n \otimes \C^n$} and \mbox{$\epsilon: \C^n \otimes \C^n \to \mathbb{C}$} as follows:
\begin{calign} 
\label{eq:vanillacup}
\begin{tz}[yscale=1.1, yscale=-1,string]
\draw (0,0) to [out=up, in=up, looseness=2] (1,0);
\node[scale=0.8] at (0,-0.2) {$\C^n$};
\node[scale=0.8] at (1,-0.2) {$\C^n$};
\end{tz}
&
\begin{tz}[yscale=1.1,string]
\draw (0,0) to [out=up, in=up, looseness=2] (1,0);
\node[scale=0.8] at (0,-0.2) {$\C^n$};
\node[scale=0.8] at (1,-0.2) {$\C^n$};
\end{tz}
\\*[-5pt]
\nonumber
\eta:\,1\mapsto \sum_{i}\, \ket{i} \otimes \ket{i} 
&
\epsilon:\,\ket{i}\otimes \ket{j} \mapsto \delta_{ij}
\end{calign}
The notation is justified, since the following equations can be demonstrated:
\begin{calign}
\label{eq:vanillasnake}
\begin{tz}[string, scale=1.5]
\draw (0,0.25) to (0,1) to [out=up, in=up, looseness=2] (-0.5,1) to [out=down, in=down, looseness=2] (-1,1) to (-1,1.75);
\end{tz}
\eqgap=\eqgap
\begin{tz}[scale=1.5]
\draw [string] (0,0.25) to (0,1.75);
\end{tz}
\eqgap=\eqgap
\begin{tz}[string,xscale=-1, scale=1.5]
\draw (0,0.25) to (0,1) to [out=up, in=up, looseness=2] (-0.5,1) to [out=down, in=down, looseness=2] (-1,1) to (-1,1.75);
\end{tz}
\end{calign}
It can easily be verified that for a linear map $f:\C^n \to \C^n$, we have the following:
\begin{calign} 
\label{eq:trace}
\begin{tz}[yscale=1.1,string]
\node (f) [blob] at (0,-0.25){$f$};
\draw [string] (f.center) to (f.north) to [out=up, in=up, looseness=2] (1,0 |- f.north) to (1,-0.5 |- f.south) to [out=down, in=down, looseness=2] (f.south) to (f.center);
\end{tz}
\eqgap=\eqgap
\begin{tz}[yscale=1.1,string, xscale=-1]
\node (f) [blob] at (0,-0.25){$f$};
\draw [string] (f.center) to (f.north) to [out=up, in=up, looseness=2] (1,0 |- f.north) to (1,-0.5 |- f.south) to [out=down, in=down, looseness=2] (f.south) to (f.center);
\end{tz}
\eqgap=\eqgap
\Tr(f) 
&
\begin{tz}[string]
\draw (0,0) to [out=up, in=up, looseness=2] (1,0) to [out= down, in=down , looseness=2] (0,0);
\end{tz}
\eqgap=\eqgap n
\end{calign}
Since wires in our framework correspond to indexed families of Hilbert spaces, and assuming for simplicity that  all Hilbert spaces are chosen to be of the form $\C^n$ for some $n \in \N$, we can introduce the following notation for families of linear maps of the form~\eqref{eq:vanillacup}:
\begin{calign}
\label{eq:cupsandcaps}
\begin{tz}[yscale=1.1, yscale=-1]
\draw [redregion, string, draw] (0,0) to [out=up, in=up, looseness=2] (1,0);
\draw [blueregion] (0,0) to (-0.5,0) to (-0.5,1) to (1.5,1) to (1.5,0) to (1,0) to [out=up, in=up, looseness=2] (0,0);
\end{tz}
&
\begin{tz}[yscale=1.1]
\draw [blueregion, string, draw] (0,0) to [out=up, in=up, looseness=2] (1,0);
\draw [redregion] (0,0) to (-0.5,0) to (-0.5,1) to (1.5,1) to (1.5,0) to (1,0) to [out=up, in=up, looseness=2] (0,0);
\end{tz} 
&
\begin{tz}[yscale=1.1, yscale=-1]
\draw [blueregion, string, draw] (0,0) to [out=up, in=up, looseness=2] (1,0);
\draw [redregion] (0,0) to (-0.5,0) to (-0.5,1) to (1.5,1) to (1.5,0) to (1,0) to [out=up, in=up, looseness=2] (0,0);
\end{tz}
&
\begin{tz}[yscale=1.1]
\draw [redregion, string, draw] (0,0) to [out=up, in=up, looseness=2] (1,0);
\draw [blueregion] (0,0) to (-0.5,0) to (-0.5,1) to (1.5,1) to (1.5,0) to (1,0) to [out=up, in=up, looseness=2] (0,0);
\end{tz}
\end{calign}
Then the following hold as a direct consequence of equations~\eqref{eq:vanillasnake}:
\begin{calign}
\label{eq:colouredcupscaps}
\begin{tz}[scale=1,xscale=-1]
\draw [blueregion] (0,0.25) to (0,1) to [out=up, in=up, looseness=2] (-0.5,1) to [out=down, in=down, looseness=2] (-1,1) to (-1,1.75) to (0.5,1.75) to (0.5,0.25);
\path[redregion] (0,0.25) to (0,1) to [out=up, in=up, looseness=2] (-0.5,1) to [out=down, in=down, looseness=2] (-1,1) to (-1,1.75) to (-1.5,1.75) to (-1.5, 0.25);
\draw [string] (0,0.25) to (0,1) to [out=up, in=up, looseness=2] (-0.5,1) to [out=down, in=down, looseness=2] (-1,1) to (-1,1.75);
\end{tz}
=
\begin{tz}[scale=1,xscale=-1]
\draw [blueregion] (0,0.25) to (0,1.75) to (0.5,1.75) to (0.5,0.25);
\draw[redregion] (0,0.25) to (0,1.75) to (-0.5,1.75) to (-0.5,0.25);
\draw [string] (0,0.25) to (0,1.75);
\end{tz}
=
\begin{tz}[scale=1,scale=-1]
\draw [blueregion] (0,0.25) to (0,1) to [out=up, in=up, looseness=2] (-0.5,1) to [out=down, in=down, looseness=2] (-1,1) to (-1,1.75) to (0.5,1.75) to (0.5,0.25);
\path[redregion] (0,0.25) to (0,1) to [out=up, in=up, looseness=2] (-0.5,1) to [out=down, in=down, looseness=2] (-1,1) to (-1,1.75) to (-1.5,1.75) to (-1.5, 0.25);
\draw [string] (0,0.25) to (0,1) to [out=up, in=up, looseness=2] (-0.5,1) to [out=down, in=down, looseness=2] (-1,1) to (-1,1.75);
\end{tz}
&
\begin{tz}[scale=1]
\draw [blueregion] (0,0.25) to (0,1) to [out=up, in=up, looseness=2] (-0.5,1) to [out=down, in=down, looseness=2] (-1,1) to (-1,1.75) to (0.5,1.75) to (0.5,0.25);
\path[redregion] (0,0.25) to (0,1) to [out=up, in=up, looseness=2] (-0.5,1) to [out=down, in=down, looseness=2] (-1,1) to (-1,1.75) to (-1.5,1.75) to (-1.5, 0.25);
\draw [string] (0,0.25) to (0,1) to [out=up, in=up, looseness=2] (-0.5,1) to [out=down, in=down, looseness=2] (-1,1) to (-1,1.75);
\end{tz}
=
\begin{tz}[scale=1]
\draw [blueregion] (0,0.25) to (0,1.75) to (0.5,1.75) to (0.5,0.25);
\draw[redregion] (0,0.25) to (0,1.75) to (-0.5,1.75) to (-0.5,0.25);
\draw [string] (0,0.25) to (0,1.75);
\end{tz}
=
\begin{tz}[scale=1,yscale=-1]
\draw [blueregion] (0,0.25) to (0,1) to [out=up, in=up, looseness=2] (-0.5,1) to [out=down, in=down, looseness=2] (-1,1) to (-1,1.75) to (0.5,1.75) to (0.5,0.25);
\path[redregion] (0,0.25) to (0,1) to [out=up, in=up, looseness=2] (-0.5,1) to [out=down, in=down, looseness=2] (-1,1) to (-1,1.75) to (-1.5,1.75) to (-1.5, 0.25);
\draw [string] (0,0.25) to (0,1) to [out=up, in=up, looseness=2] (-0.5,1) to [out=down, in=down, looseness=2] (-1,1) to (-1,1.75);
\end{tz}
\end{calign}

\paragraph{Dagger structure.}
Given a family of linear maps, we define its \textit{adjoint} (or \textit{dagger}) to be the family consisting of the adjoints of the linear maps:
\begin{calign}
\label{eq:dagger}
\begin{tz}[string]
\clip (0.25-\side,-0.6) rectangle (1.75+\side,2.6);
\path[greenregion] (0.25-\side,0) to (0.25,0) to [out= up, in=-135] (1,1) to (1,2) to (0.25-\side,2);
\path[blueregion] (1.75+\side,0) to (1.75,0) to [out= up, in=-45] (1,1) to (1,2) to (1.75+\side,2);
\path[redregion, draw] (0.25,0) to [out= up, in=-135] (1,1) to [out=-45, in=up] (1.75,0);
\draw (1,1) to (1,2);
\node[blob] at (1,1) {$A$};
\node[scale=0.8] at (0.25,-0.27) {$H$};
\node[scale=0.8] at (1.75, -0.27) {$V$};
\node[scale=0.8] at (1, 2.27) {$W$};
\node[dimension,below right] at (0.25-\side, 1.95) {$i{:}n$};
\node[dimension,below left] at (1.75+\side, 1.95) {$k{:}p$};
\node[dimension,above] at (1, 0.05) {$j{:}m$};
\end{tz}
\,\,\leftrightsquigarrow \hspace{-5pt}
\begin{tz}[string]
\clip (0,-0.6) rectangle (2,2.6);
\draw (0.25,0) to [out= up, in=-135] (1,1) to [out=-45, in=up] (1.75,0);
\draw (1,1) to (1,2);
\node[blob] at (1,1) {$A_{ijk}$};
\node[scale=0.8] at (0.25,-0.27) {$H_{ij}$};
\node[scale=0.8] at (1.75, -0.27) {$V_{jk}$};
\node[scale=0.8] at (1, 2.27) {$W_{ik}$};
\end{tz}& 
\begin{tz}[string,yscale=-1]
\clip (0.25-\side,-0.6) rectangle (1.75+\side,2.6);
\path[greenregion] (0.25-\side,0) to (0.25,0) to [out= up, in=-135] (1,1) to (1,2) to (0.25-\side,2);
\path[blueregion] (1.75+\side,0) to (1.75,0) to [out= up, in=-45] (1,1) to (1,2) to (1.75+\side,2);
\path[redregion, draw] (0.25,0) to [out= up, in=-135] (1,1) to [out=-45, in=up] (1.75,0);
\draw (1,1) to (1,2);
\node[blob] at (1,1) {$A^\dagger$};
\node[scale=0.8] at (0.25,-0.27) {$H$};
\node[scale=0.8] at (1.75, -0.27) {$V$};
\node[scale=0.8] at (1, 2.27) {$W$};
\node[dimension,above right] at (0.25-\side, 1.95) {$i{:}n$};
\node[dimension,above left] at (1.75+\side, 1.95) {$k{:}p$};
\node[dimension,below] at (1, 0.05) {$j{:}m$};
\end{tz}
\,\,\leftrightsquigarrow \hspace{-5pt}
\begin{tz}[string,yscale=-1]
\clip (0,-0.6) rectangle (2,2.6);
\draw (0.25,0) to [out= up, in=-135] (1,1) to [out=-45, in=up] (1.75,0);
\draw (1,1) to (1,2);
\node[blob] at (1,1) {$A^\dagger_{ijk}$};
\node[scale=0.8] at (0.25,-0.27) {$H_{ij}$};
\node[scale=0.8] at (1.75, -0.27) {$V_{jk}$};
\node[scale=0.8] at (1, 2.27) {$W_{ik}$};
\end{tz}
\end{calign}
Graphically, we can think of the adjoint as a reflection about a horizontal axis. This is justified, since the following holds:
\begin{calign}
\left(
\raisebox{-1pt}{$\begin{tz}[yscale=1.1, yscale=-1]
\draw [redregion, string, draw] (0,0) to [out=up, in=up, looseness=2] (1,0);
\draw [blueregion] (0,0) to (-0.5,0) to (-0.5,1) to (1.5,1) to (1.5,0) to (1,0) to [out=up, in=up, looseness=2] (0,0);
\end{tz}$}
\right)^\dagger
 = 
\begin{tz}[yscale=1.1]
\draw [redregion, string, draw] (0,0) to [out=up, in=up, looseness=2] (1,0);
\draw [blueregion] (0,0) to (-0.5,0) to (-0.5,1) to (1.5,1) to (1.5,0) to (1,0) to [out=up, in=up, looseness=2] (0,0);
\end{tz}
&
\left(
\raisebox{-1pt}{$\begin{tz}[yscale=1.1, yscale=-1]
\draw [blueregion, string, draw] (0,0) to [out=up, in=up, looseness=2] (1,0);
\draw [redregion] (0,0) to (-0.5,0) to (-0.5,1) to (1.5,1) to (1.5,0) to (1,0) to [out=up, in=up, looseness=2] (0,0);
\end{tz}$}
\right)^\dagger
 = 
\begin{tz}[yscale=1.1]
\draw [blueregion, string, draw] (0,0) to [out=up, in=up, looseness=2] (1,0);
\draw [redregion] (0,0) to (-0.5,0) to (-0.5,1) to (1.5,1) to (1.5,0) to (1,0) to [out=up, in=up, looseness=2] (0,0);
\end{tz}
\end{calign}
In total, every vertex appears in four variants:
\begin{calign}
\begin{tz}[xscale=0.8]
\draw [blueregion] (0,0) rectangle +(1,2);
\draw [redregion] (1,0) rectangle +(1,2);
\draw [string] (1,0) to +(0,2);
\node [blob] at (1,1) {$F$};
\end{tz}
&
\begin{tz}[xscale=0.8]
\draw [redregion] (0,0) rectangle +(1,2);
\draw [blueregion] (1,0) rectangle +(1,2);
\draw [string] (1,0) to +(0,2);
\node [blob] at (1,1) {$F^*$};
\end{tz}
:=\begin{tz}[scale=1.33, xscale=1]
\node (f) [blob] at (-0.5,1) {$F$};
\draw [blueregion] (0,0.25) to (0,1 |- f.north) to [out=up, in=up, looseness=2] (f.north) to (f.south) to [out=down, in=down, looseness=2] (-1,1 |- f.south) to (-1,1.75) to (0.5,1.75) to (0.5,0.25);
\draw [redregion] (0,0.25) to (0,1 |- f.north) to [out=up, in=up, looseness=2] (f.north) to (f.south) to [out=down, in=down, looseness=2] (-1,1 |- f.south) to (-1,1.75) to (-1.5,1.75) to (-1.5,0.25);
\draw [string] (0,0.25) to (0,1 |- f.north) to [out=up, in=up, looseness=2] (f.north) to (f.south) to [out=down, in=down, looseness=2] (-1,1 |- f.south) to (-1,1.75);
\end{tz}
=\begin{tz}[scale=1.33, xscale=-1]
\node (f) [blob] at (-0.5,1) {$F$};
\draw [redregion] (0,0.25) to (0,1 |- f.north) to [out=up, in=up, looseness=2] (f.north) to (f.south) to [out=down, in=down, looseness=2] (-1,1 |- f.south) to (-1,1.75) to (0.5,1.75) to (0.5,0.25);
\draw [blueregion] (0,0.25) to (0,1 |- f.north) to [out=up, in=up, looseness=2] (f.north) to (f.south) to [out=down, in=down, looseness=2] (-1,1 |- f.south) to (-1,1.75) to (-1.5,1.75) to (-1.5,0.25);
\draw [string] (0,0.25) to (0,1 |- f.north) to [out=up, in=up, looseness=2] (f.north) to (f.south) to [out=down, in=down, looseness=2] (-1,1 |- f.south) to (-1,1.75);
\end{tz}
\\
\label{eq:reflectedrotated}
\begin{tz}[xscale=0.8]
\draw [blueregion] (0,0) rectangle +(1,2);
\draw [redregion] (1,0) rectangle +(1,2);
\draw [string] (1,0) to +(0,2);
\node [blob] at (1,1) {$F^\dagger$};
\end{tz}
&
\begin{tz}[xscale=0.8]
\draw [redregion] (0,0) rectangle +(1,2);
\draw [blueregion] (1,0) rectangle +(1,2);
\draw [string] (1,0) to +(0,2);
\node [blob] at (1,1) {$F_*$};
\end{tz}
:=\begin{tz}[scale=1.33, xscale=1]
\node (f) [blob, yscale=1] at (-0.5,1) {$F$};
\draw [blueregion] (0,0.25) to (0,1 |- f.north) to [out=up, in=up, looseness=2] (f.north) to (f.south) to [out=down, in=down, looseness=2] (-1,1 |- f.south) to (-1,1.75) to (0.5,1.75) to (0.5,0.25);
\draw [redregion] (0,0.25) to (0,1 |- f.north) to [out=up, in=up, looseness=2] (f.north) to (f.south) to [out=down, in=down, looseness=2] (-1,1 |- f.south) to (-1,1.75) to (-1.5,1.75) to (-1.5,0.25);
\draw [string] (0,0.25) to (0,1 |- f.north) to [out=up, in=up, looseness=2] (f.north) to (f.south) to [out=down, in=down, looseness=2] (-1,1 |- f.south) to (-1,1.75);
\node (f) [blob] at (-0.5,1) {$F^\dagger$};
\end{tz}
=\begin{tz}[scale=1.33, xscale=-1]
\node (f) [blob] at (-0.5,1) {$F$};
\draw [redregion] (0,0.25) to (0,1 |- f.north) to [out=up, in=up, looseness=2] (f.north) to (f.south) to [out=down, in=down, looseness=2] (-1,1 |- f.south) to (-1,1.75) to (0.5,1.75) to (0.5,0.25);
\draw [blueregion] (0,0.25) to (0,1 |- f.north) to [out=up, in=up, looseness=2] (f.north) to (f.south) to [out=down, in=down, looseness=2] (-1,1 |- f.south) to (-1,1.75) to (-1.5,1.75) to (-1.5,0.25);
\draw [string] (0,0.25) to (0,1 |- f.north) to [out=up, in=up, looseness=2] (f.north) to (f.south) to [out=down, in=down, looseness=2] (-1,1 |- f.south) to (-1,1.75);
\node (f) [blob] at (-0.5,1) {$F^\dagger$};
\end{tz}
\end{calign}
The equations on the right-hand sides can be shown to follow from the definitions~\eqref{eq:vanillacup}.

A dagger structure gives rise to a general notion of unitarity.
\begin{definition}
\label{def:unitary2morphism}
A vertex $U$ is \textit{unitary} when it satisfies the following equations:
\begin{align}
\label{eq:unitarity}
\begin{tz}[xscale=1.5,string]
\path[blueregion] (0,0) to (0.5,0) to (0.5,2) to (0,2);
\path[redregion] (0.5,0) to (1,0) to (1,2) to (0.5,2);
\draw (0.5,0) to (0.5,2);
\node[blob] at (0.5,0.6) {$U$};
\node[blob] at (0.5,1.4) {$U^\dagger$};
\end{tz}
\eqgap&=\eqgap
\begin{tz} [xscale=1.5,string]
\path[blueregion] (0,0) to (0.5,0) to (0.5,2) to (0,2) ;
\path[redregion] (0.5,0) to (1,0) to (1,2) to (0.5,2) ;
\draw (0.5,0) to  (0.5,2) ;
\end{tz}
&
\begin{tz}[xscale=1.5,string]
\path[blueregion] (0,0) to (0.5,0) to (0.5,2) to (0,2);
\path[redregion] (0.5,0) to (1,0) to (1,2) to (0.5,2);
\draw (0.5,0) to (0.5,2);
\node[blob] at (0.5,0.6) {$U^\dagger$};
\node[blob] at (0.5,1.4) {$U$};
\end{tz}
\eqgap&=\eqgap
\begin{tz} [xscale=1.5,string]
\path[blueregion] (0,0) to (0.5,0) to (0.5,2) to (0,2) ;
\path[redregion] (0.5,0) to (1,0) to (1,2) to (0.5,2) ;
\draw (0.5,0) to  (0.5,2) ;
\end{tz}
\end{align}
\end{definition}

\paragraph{Standard boundaries.} This paper only makes use of a restricted portion of this calculus: wires which bound only one shaded region always correspond  to the 1\-dimensional Hilbert space $\C$ for any value of the controlling parameter. (Wires that do not bound regions may correspond to Hilbert spaces of any finite dimension.)
In particular, since they are 1\-dimensional, the Hilbert spaces arising from standard boundaries are not depicted in the corresponding family of tensor diagrams:
\begin{calign}
\label{eq:boundary}
\begin{tz}[string]
\path[blueregion] (1.75,0) rectangle (1,2);
\draw[bnd] (1,0) to (1,2);
\node[dimension] at (1.4,1.0) {$i{:}n$};
\end{tz}
\,\, \leftrightsquigarrow\hspace{-10pt}
\begin{tz}[string]
\clip (0.25, -0.5) rectangle (1.75, 2.5);
\draw[dotted]  (1,0) to (1,2);
\node[scale=0.7] at (1,-0.3) {$H_i = \mathbb{C}$};
\end{tz}
&
\begin{tz}[string]
\path[blueregion] (0.25,0) rectangle (1,2);
\draw[bnd] (1,0) to (1,2);
\node[dimension] at (0.6,1) {$i{:}n$};
\end{tz}
\,\, \leftrightsquigarrow\hspace{-10pt}
\begin{tz}[string]
\clip (0.25, -0.5) rectangle (1.75, 2.5);
\draw[dotted]  (1,0) to (1,2);
\node[scale=0.7] at (1,-0.3) {$H_i = \mathbb{C}$};
\end{tz}
\end{calign}
This means that once the parameter $i:n$ for the region is given, no further labelling is needed for the wire itself.

The following properties may be verified for these standard boundaries:
\begin{calign}
\label{eq:circlen}
\begin{tz}[string]
\path[blueregion] (-1,-1) rectangle (1,1);
\path[ fill=white,bnd] (0,0) circle (0.5cm);
\end{tz}
\eqgap=\eqgap
\begin{tz}[string]
\path[blueregion] (-1,-1) rectangle (1,1);
\end{tz}
&
\begin{tz}
\draw [blueregion, bnd] (0,0) circle (0.5cm);
\end{tz}
\eqgap=\eqgap
n
\end{calign}

\noindent
Many definitions and results of this paper hold more generally, but the main application to constructions of Hadamard matrices, unitary error bases and quantum Latin squares in Sections~\ref{sec:binaryquantum}--\ref{sec:infinity} only involve this restricted calculus.

\paragraph{Labels for wires.}
As mentioned above, wires that do not bound regions may correspond to a Hilbert space $\C^n$ of any finite dimension. We will label such a wire simply by $n$, the dimension of its Hilbert space, or sometimes as $i:n$, indicating that $i$ is a parameter that will range over the computational basis elements $\ket i \in \C^n$.

\paragraph{Biunitaries.} Having defined our graphical calculus, we now  define biunitarity.

\def\sidew{0.5}
\begin{definition}
A \emph{biunitary} is a vertex
\begin{equation}
\label{eq:2cellbiunitary}
\begin{tz}[string]
\node [blob] at (1,1) {$U$};
\path[redregion] (0.25,0) to [out=90, in=-135] (1,1) to [out= 135, in=-90] (0.25,2) to (0.25-\sidew,2) to (0.25-\sidew,0);
\path[blueregion] (1.75,0) to [out=90, in=-45] (1,1) to [out= 45, in=-90] (1.75,2) to (1.75+\sidew,2) to (1.75+\sidew,0);
\path[greenregion,draw, string] (0.25,0) to [out=90, in=-135] (1,1) to [out=-45, in=90] (1.75,0);
\path[yellowregion, draw, string] (0.25,2) to [out=-90, in=135] (1,1) to [out=45, in=-90] (1.75,2);
\end{tz}
\end{equation}
which is unitary~\eqref{eq:biunitaryverticallyunitary2}, and which also satisfies the following \textit{horizontal unitarity} equations~\eqref{eq:biunitaryhorizontallyunitary2} for some scalar $\lambda$:
\def\bigangle{150}
\def\smallangle{30}
\def \sidew {0.5}
\def \scl{0.69}
\begin{align}
\label{eq:biunitaryverticallyunitary2}
\begin{tz}[string,scale=\scl]
\path[redregion] (0.25-\sidew,0) rectangle (1,3.5);
\path[blueregion] (1,0) rectangle (1.75+\sidew,3.5);
\path[fill=white] (0.25,0) to [out=90, in=-135] (1,1) to [out=-45, in=90] (1.75,0);
\path[greenregion,draw] (0.25,0) to [out=90, in=-135] (1,1) to [out=-45, in=90] (1.75,0);
\path[fill=white] (0.25,3.5) to [out=-90, in=135] (1,2.5) to [out=45, in=-90] (1.75,3.5);
\path[greenregion,draw] (0.25,3.5) to [out=-90, in=135] (1,2.5) to [out=45, in=-90] (1.75,3.5);
\path[fill=white] (1,2.5) to [out=-135, in=90] (0.35, 1.75) to [out=-90, in=135] (1,1) to [out=45, in=-90] (1.65,1.75) to [out=90, in=-45] (1,2.5);
\path[yellowregion,draw] (1,2.5) to [out=-135, in=90] (0.35, 1.75) to [out=-90, in=135] (1,1) to [out=45, in=-90] (1.65,1.75) to [out=90, in=-45] (1,2.5);
\node[blob] at (1,1) {$U$};
\node[blob] at (1,2.5) {$U^\dagger$};
\end{tz}
&=
\begin{tz}[string,scale=\scl]
\path[redregion] (0.25-\sidew,0) rectangle (0.25,3.5);
\path[blueregion] (1.75,0) rectangle (1.75+\sidew,3.5);
\path[greenregion] (0.25,0) rectangle (1.75,3.5);
\draw (0.25,0) to (0.25,3.5);
\draw (1.75,0) to (1.75,3.5);
\end{tz}
&
\begin{tz}[string,scale=\scl]
\path[redregion] (0.25-\sidew,0) rectangle (1,3.5);
\path[blueregion] (1,0) rectangle (1.75+\sidew,3.5);
\path[fill=white] (0.25,0) to [out=90, in=-135] (1,1) to [out=-45, in=90] (1.75,0);
\path[yellowregion,draw] (0.25,0) to [out=90, in=-135] (1,1) to [out=-45, in=90] (1.75,0);
\path[fill=white] (0.25,3.5) to [out=-90, in=135] (1,2.5) to [out=45, in=-90] (1.75,3.5);
\path[yellowregion,draw] (0.25,3.5) to [out=-90, in=135] (1,2.5) to [out=45, in=-90] (1.75,3.5);
\path[fill=white] (1,2.5) to [out=-135, in=90] (0.35, 1.75) to [out=-90, in=135] (1,1) to [out=45, in=-90] (1.65,1.75) to [out=90, in=-45] (1,2.5);
\path[greenregion,draw] (1,2.5) to [out=-135, in=90] (0.35, 1.75) to [out=-90, in=135] (1,1) to [out=45, in=-90] (1.65,1.75) to [out=90, in=-45] (1,2.5);
\node[blob] at (1,1) {$U^\dagger$};
\node[blob] at (1,2.5) {$U$};
\end{tz}
&=
\begin{tz}[string,scale=\scl]
\path[redregion] (0.25-\sidew,0) rectangle (0.25,3.5);
\path[blueregion] (1.75,0) rectangle (1.75+\sidew,3.5);
\path[yellowregion] (0.25,0) rectangle (1.75,3.5);
\draw (0.25,0) to (0.25,3.5);
\draw (1.75,0) to (1.75,3.5);
\end{tz}
\\[5pt]
\label{eq:biunitaryhorizontallyunitary2}
\begin{tz}[string,scale=\scl]
\path[redregion] (3.25,2) to [out=-90, in=45] (2.5,1) to [out=-45, in=90] (3.25,0) to(3.25+\sidew,0) to (3.25+\sidew,2);
\path[redregion] (0.25,0) to [out=90, in=-135] (1,1) to [out=135, in=-90] (0.25,2) to (0.25-\sidew,2) to (0.25-\sidew,0);
\path[blueregion] (2.5,1) to [out=-135, in=0] (1.75,0.3)  to [out=180, in=-45] (1,1) to [out=45, in=180] (1.75,1.7) to [out=0, in=135] (2.5,1);
\path[greenregion,draw] (3.25,0) to [out=90, in=-45] (2.5,1)to [out=-135, in=0] (1.75,0.3)  to [out=180, in=-45] (1,1) to [out=-135, in=90] (0.25,0);
\path[yellowregion,draw] (0.25,2) to [out=-90, in=135] (1,1) to [out=45, in=180] (1.75,1.7) to [out=0, in=135] (2.5,1) to [out=45, in=-90] (3.25,2);
\node[blob] at (1,1) {$U$};
\node[blob]at (2.5,1) {$U_*$};
\end{tz}
&=
\lambda\,
\begin{tz}[string,scale=\scl] 
\path[redregion] (3.25,0) to [out=90, in=0] (1.75,0.7)  to [out=180, in=90] (0.25,0) to (0.25-\sidew,0) to (0.25-\sidew,2) to (0.25,2) to [out=-90, in=180] (1.75,1.3) to [out=0 , in=-90] (3.25,2) to (3.25+\sidew,2) to (3.25+\sidew,0);
\path[yellowregion,draw] (3.25,2) to [out=-90, in=0] (1.75,1.3)  to [out=180, in=-90] (0.25,2);
\path[greenregion,draw] (3.25,0) to [out=90, in=0] (1.75,0.7)  to [out=180, in=90] (0.25,0);
\end{tz}
&
\begin{tz}[string,scale=\scl]
\path[blueregion] (0.25,0) to [out=90, in=-135] (1,1) to [out=135, in=-90] (0.25,2) to (0.25-\sidew,2) to (0.25-\sidew,0);
\path[redregion] (2.5,1) to [out=-135, in=0] (1.75,0.3)  to [out=180, in=-45] (1,1) to [out=45, in=180] (1.75,1.7) to [out=0, in=135] (2.5,1);
\path[blueregion] (3.25,2) to [out=-90, in=45] (2.5,1) to [out=-45, in=90] (3.25,0) to(3.25+\sidew,0) to (3.25+\sidew,2);
\path[greenregion,draw] (3.25,0) to [out=90, in=-45] (2.5,1)to [out=-135, in=0] (1.75,0.3)  to [out=180, in=-45] (1,1) to [out=-135, in=90] (0.25,0);
\path[yellowregion,draw] (0.25,2) to [out=-90, in=135] (1,1) to [out=45, in=180] (1.75,1.7) to [out=0, in=135] (2.5,1) to [out=45, in=-90] (3.25,2);
\node[blob] at (1,1) {$U_*$};
\node[blob] at (2.5,1) {$U$};
\end{tz}
&=
 \lambda\,
\begin{tz} [string,scale=\scl]
\path[blueregion] (3.25,0) to [out=90, in=0] (1.75,0.7)  to [out=180, in=90] (0.25,0) to (0.25-\sidew,0) to (0.25-\sidew,2) to (0.25,2) to [out=-90, in=180] (1.75,1.3) to [out=0 , in=-90] (3.25,2) to (3.25+\sidew,2) to (3.25+\sidew,0);
\path[yellowregion,draw] (3.25,2) to [out=-90, in=0] (1.75,1.3)  to [out=180, in=-90] (0.25,2);
\path[greenregion,draw] (3.25,0) to [out=90, in=0] (1.75,0.7)  to [out=180, in=90] (0.25,0);
\end{tz} 
\end{align}
\end{definition}

\noindent
Note that biunitarity depends on a chosen partition of the input and output wires into two parts. Such a partition may not be unique; in particular, every unitary vertex $U$ is biunitary with respect to the following partitions: 
\begin{calign}
\label{eq:unitarybiunitary}
\begin{tz}[scale=0.8,string] 
\path[blueregion] (0.25,0) to [out= up, in=-135] (1,1) to [out= 45, in=down] (1.75,2) to (0.25-\side, 2) to (0.25-\side,0);
\path[redregion] (0.25,0) to [out= up, in=-135] (1,1) to [out= 45, in=down] (1.75,2) to (1.75+\side,2) to (1.75+\side,0);
\draw  (0.25,0) to [out= up, in=-135] (1,1) to [out= 45, in=down] (1.75,2);
\node[blob] at (1,1) {$U$};
\end{tz}
&
\begin{tz}[scale=0.8,string] 
\path[blueregion] (1.75,0) to [out= up, in=-45] (1,1) to [out= 135, in=down] (0.25,2) to (0.25-\side, 2) to (0.25-\side,0);
\path[redregion]  (1.75,0) to [out= up, in=-45] (1,1) to [out= 135, in=down] (0.25,2) to (1.75+\side,2) to (1.75+\side,0);
\draw  (1.75,0) to [out= up, in=-45] (1,1) to [out= 135, in=down] (0.25,2) ;
\node[blob] at (1,1) {$U$};
\end{tz}
\end{calign}
\noindent
 The scalar $\lambda$ is uniquely determined and can be recovered as a consequence of equations~\eqref{eq:biunitaryverticallyunitary2} and~\eqref{eq:biunitaryhorizontallyunitary2}:
\begin{equation}\label{eq:recoverscalar}\lambda\,
\begin{tz}[string,scale=0.8] 
\path[redregion] (0,0) rectangle (2,3);
\path[greenregion] (2,0) rectangle (2+\side,3);
\draw (2,0) to (2,3);

\path[fill = white] (1, 1.5) circle (0.5);
\path[yellowregion,draw] (1, 1.5) circle (0.5);

\end{tz}
\, \superequals{eq:biunitaryhorizontallyunitary2}\,
\begin{tz}[string,scale=0.8]
\path[redregion] (0.25-\side,0) rectangle (4.1,3);

\path[fill=white] (2.5,1) to [out=-135, in=0] (1.75,0.3)  to [out=180, in=-45] (1,1) to [out=45, in=180] (1.75,1.7) to [out=0, in=135] (2.5,1);
\path[blueregion] (2.5,1) to [out=-135, in=0] (1.75,0.3)  to [out=180, in=-45] (1,1) to [out=45, in=180] (1.75,1.7) to [out=0, in=135] (2.5,1);

\path[fill=white] (0.25,0) to [out=90, in=-135] (1,1) to [out=-45, in=180] (1.75,0.3) to [out=0, in=-135] (2.5,1) to (2.6,1.2) to [out=-90, in=180] (3.1,0.3) to [out=0, in=-90] (3.6, 1.2) to (3.6, 3) to  (3.6+\side,3) to (3.6+\side,0);
\path[greenregion] (0.25,0) to [out=90, in=-135] (1,1) to [out=-45, in=180] (1.75,0.3) to [out=0, in=-135] (2.5,1) to (2.6,1.2) to [out=-90, in=180] (3.1,0.3) to [out=0, in=-90] (3.6, 1.2) to (3.6, 3) to  (3.6+\side,3) to (3.6+\side,0);
\draw (0.25,0) to [out=90, in=-135] (1,1) to [out=-45, in=180] (1.75,0.3) to [out=0, in=-135] (2.5,1) to (2.6,1.2) to [out=-90, in=180] (3.1,0.3) to [out=0, in=-90] (3.6, 1.2) to (3.6, 3);

\path[fill=white] (0.35,2) to [out=-90, in=135] (1,1) to [out=45, in=180] (1.75,1.7) to [out=0, in=135] (2.5,1) to [out=45, in=-90] (3.15,2) to [out=90, in=0] (1.75,2.7) to [out=180, in=90] (0.35,2);
\path[yellowregion,draw] (0.35,2) to [out=-90, in=135] (1,1) to [out=45, in=180] (1.75,1.7) to [out=0, in=135] (2.5,1) to [out=45, in=-90] (3.15,2) to [out=90, in=0] (1.75,2.7) to [out=180, in=90] (0.35,2);

\node[blob] at (1,1) {$U$};
\node[blob]at (2.5,1) {$U_*$};
\end{tz} 
\,\superequals{eq:reflectedrotated}\,
\begin{tz}[string,scale=0.8]
\path[redregion] (0.25-\side,0) rectangle (2.6,3);

\path[fill=white] (0.25,0) to [out=90, in=-135] (1,1) to (1,2) to [out=135, in=-90] (0.25,3) to (2.6, 3) to (2.6,0);
\path[greenregion] (0.25,0) to [out=90, in=-135] (1,1) to (1,2) to [out=135, in=-90] (0.25,3) to (2.6, 3) to (2.6,0);
\draw (0.25,0) to [out=90, in=-135] (1,1) to (1,2) to [out=135, in=-90] (0.25,3);

\path[fill=white] (1.1,1.3) to [out=-90, in=180] (1.6,0.5) to [out=0, in=-90] (2.1, 1.3) to (2.1,1.7) to [out= 90, in=0] (1.6, 2.5) to [out=180, in=90] (1.1, 1.7) to (1,2) to [out=-45, in=90] (1.5,1.5) to [out=-90, in=45] (1,1);
\path[blueregion,draw] (1.1,1.3) to [out=-90, in=180] (1.6,0.5) to [out=0, in=-90] (2.1, 1.3) to (2.1,1.7) to [out= 90, in=0] (1.6, 2.5) to [out=180, in=90] (1.1, 1.7) to (1,2) to [out=-45, in=90] (1.5,1.5) to [out=-90, in=45] (1,1);

\path[fill=white] (1,1) to [out=135, in=-90] (0.5, 1.5) to [out=90, in=-135] (1,2) to [out=-45, in=90] (1.5,1.5) to [out=-90, in=45] (1,1);
\path[yellowregion,draw] (1,1) to [out=135, in=-90] (0.5, 1.5) to [out=90, in=-135] (1,2) to [out=-45, in=90] (1.5,1.5) to [out=-90, in=45] (1,1);

\node[blob] at (1,1) {$U$};
\node[blob] at (1,2) {$U^\dagger$};
\end{tz}
\,\superequals{eq:biunitaryverticallyunitary2}\,
\begin{tz}[scale=0.8]
\path[redregion] (0.4-\side,0) rectangle (0.4,3);

\path[greenregion] (0.4,0) rectangle (2.4,3);
\draw [string] (0.4,0) to (0.4,3);

\path[fill=white] (1.4, 1.5) circle (0.5);
\path[blueregion,draw, string] (1.4, 1.5) circle (0.5);

\end{tz}
\end{equation}
In particular, $\lambda$ is real and positive.
We will usually use the following equivalent formulation of biunitarity.
\begin{definition}
The \emph{clockwise} and \emph{anticlockwise quarter rotation} of a vertex $U$ of type~\eqref{eq:2cellbiunitary} is given by the following composites, respectively:
\begin{calign}
\begin{tz}[string, xscale=0.5,yscale=-1]
\path [yellowregion] (0.25,0) to [out=up, in=-135] (1,1) to [out=-45, in=left] (2.3,0.3) to [out=right, in=down] (3.25,1) to (3.25,2) to (3.25+2*\side,2) to (3.25+2*\side,0);
\path[greenregion] (1.75,2) to [out=down, in=45] (1,1) to [out=135, in=right] (-0.3, 1.7) to [out=left, in=up] (-1.25,1) to (-1.25,0) to(-1.25-2*\side,0) to (-1.25-2*\side,2);
\path[redregion,draw, string] (0.25,0) to [out=up, in=-135] (1,1) to [out=135, in=right] (-0.3, 1.7) to [out=left, in=up] (-1.25,1) to (-1.25,0);
\path[blueregion,draw, string] (1.75,2) to [out=down, in=45] (1,1) to [out= -45, in= left] (2.3,0.3) to [out=right, in=down] (3.25,1) to (3.25,2);
\node [blob] at (1,1) {$U$};
\end{tz} 
\hspace{0.15\textwidth}
\begin{tz}[string, xscale=0.5]
\path [greenregion] (0.25,0) to [out=up, in=-135] (1,1) to [out=-45, in=left] (2.3,0.3) to [out=right, in=down] (3.25,1) to (3.25,2) to (3.25+2*\side,2) to (3.25+2*\side,0);
\path[yellowregion] (1.75,2) to [out=down, in=45] (1,1) to [out=135, in=right] (-0.3, 1.7) to [out=left, in=up] (-1.25,1) to (-1.25,0) to(-1.25-2*\side,0) to (-1.25-2*\side,2);
\path[redregion,draw, string] (0.25,0) to [out=up, in=-135] (1,1) to [out=135, in=right] (-0.3, 1.7) to [out=left, in=up] (-1.25,1) to (-1.25,0);
\path[blueregion,draw, string] (1.75,2) to [out=down, in=45] (1,1) to [out= -45, in= left] (2.3,0.3) to [out=right, in=down] (3.25,1) to (3.25,2);
\node [blob] at (1,1) {$U$};
\end{tz}
\end{calign}
\end{definition}
\begin{proposition}
\label{thm:rotationfactor}
Given a vertex $U$ of type~\eqref{eq:2cellbiunitary}, the following are equivalent:
\begin{enumerate}
\item $U$ is biunitary;
\item $U$ is unitary, and its clockwise quarter rotation is proportional to a unitary;
\item $U$ is unitary, and its anticlockwise quarter rotation is proportional to a unitary.
\end{enumerate}
Furthermore, in cases 2 and 3, the proportionality factor is unique up to a phase and given by a square root of $\lambda$.
\end{proposition}
\begin{proof} The proposition follows straightforwardly from deformations of \eqref{eq:biunitaryhorizontallyunitary2}.
\end{proof}
\begin{cor}
\label{lemma:rotate}
Given a biunitary, arbitrary quarter rotations, or reflections about horizontal or vertical axes, are again proportional to biunitaries.
\end{cor}

\noindent
In particular, as soon as we have characterized specific quantum structures in terms of biunitaries of certain types, we know that rotated and reflected versions of this type also correspond to this quantum structure, possibly after multiplication by a scalar.

\subsection{Characterizing quantum structures}
\label{sec:characterizing}

In this section we recall the biunitary characterizations of Hadamard matrices and unitary error bases, and give new characterizations of quantum Latin squares and controlled families. These results are summarized in \autoref{fig:biunitarytypes}.

\begin{figure}[t!]
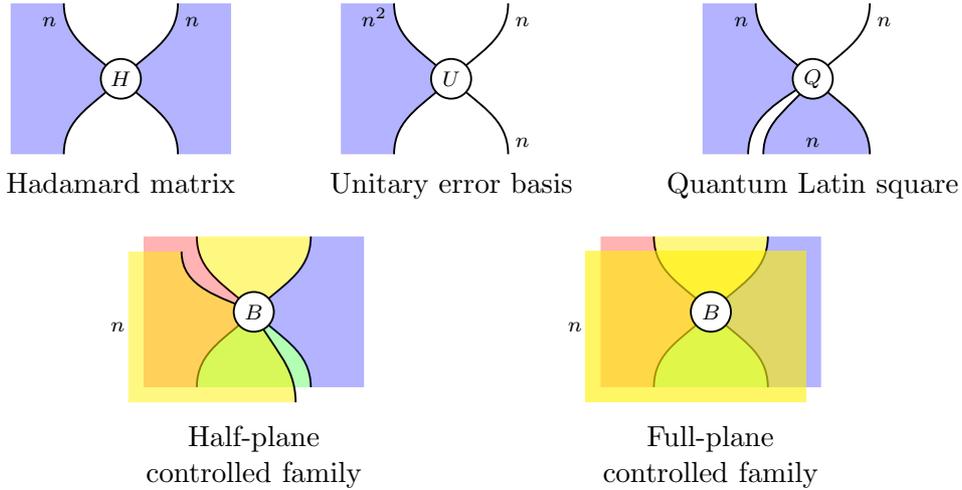

\small
\begin{calign} 
\nonumber
\begin{tz}[string]
\path [use as bounding box] (-0.25,0) rectangle +(2.5,2);
\path[blueregion] (0.25,0) to [out=90, in=-135] (1,1) to [out= 135, in=-90] (0.25,2) to (0.25-\side,2) to (0.25-\side,0);
\draw[bnd] (0.25,0) to [out=90, in=-135] (1,1) to [out=135, in=-90] (0.25,2);
\path[blueregion] (1.75,0) to [out=90, in=-45] (1,1) to [out=45, in=-90] (1.75,2) to (1.75+\side,2) to (1.75+\side,0);
\draw[bnd] (1.75,0) to [out=90, in=-45]  (1,1) to [out= 45, in=-90] (1.75,2);
\node [dimension, below right] at (1.8,1.75) {$\smash{n}$};
\node [dimension,below left] at (0.2,1.75) {$\smash{n}$};
\node [blob] at (1,1) {$H$};
\end{tz}
& 
\begin{tz}[string]
\path [use as bounding box] (-0.25,0) rectangle +(2.5,2);
\path[blueregion] (0.25,0) to [out=90, in=-135] (1,1) to [out= 135, in=-90] (0.25,2) to (0.25-\side,2) to (0.25-\side,0);
\draw[bnd] (0.25,0) to [out=90, in=-135] (1,1) to [out=135, in=-90] (0.25,2);
\path[draw, string] (1.75,0) to [out=90, in=-45] (1,1) to [out=45, in=-90] (1.75,2);
\node [dimension,below left] at (0.2,1.75) {$\smash{n^2}$};
\node [dimension, above right] at (1.8,0.05) {$n$};
\node [dimension, below right] at (1.8,1.75) {$\smash{n}$};
\node [blob] at (1,1) {$U$};
\end{tz}
&
\begin{tz}[xscale=-1, string, scale=1]
\path [use as bounding box] (0.25-\side,0) rectangle (1.75+\side,2);
\node [blob] at (1,1) {$Q$};
\path [blueregion,bnd] (0.25,0) to [out=up, in=-135] (1,1) to [out=-55, in=up] (1.65,0);
\path [blueregion] (1.75,2) to [out=down, in=45] (1,1) to [out=-35, in=up] (1.85,0) to (1.75+\side,0) to (1.75+\side,2);
\draw [bnd] (1.75,2) to [out=-90, in=45] (1,1) to [out=-35, in=up] (1.85,0);
\draw [string] (0.25,2) to [out=-90, in=135] (1,1);
\node [dimension,below left] at (1.8,1.75) {$\smash{n}$};
\node [dimension,above] at (1,0.05) {$n$};
\node [dimension,below right] at (0.2,1.75) {$\smash{n}$};
\end{tz}
\\*[0pt]
\nonumber
\text{Hadamard matrix}& \text{Unitary error basis} & \text{Quantum Latin square} 
\end{calign}

\vspace{-20pt}
\begin{calign}
\nonumber
\begin{tz}[string]
\path [use as bounding box] (0.25-\side,-0.1) rectangle (1.75+\side,2);
\path [greenregion, draw] (0.25,0) to [out=90, in=-135] (1,1) to [out=-45, in=90] (1.75,0);
\path [yellowregion, draw] (0.25,2) to [out=-90, in=135] (1,1) to [out=45, in=-90] (1.75,2);
\path [redregion] (0.25,0) to [out=90, in=-135] (1,1) to [out=135, in=-90] (0.25,2) to (0.25-\side,2) to (0.25-\side,0);
\path [blueregion] (1.75,0) to [out=90, in=-45] (1,1) to [out=45, in=-90] (1.75,2) to (1.75+\side,2) to (1.75+\side,0);
\path [yellowregion] (0.25-\xdelta-\side,2-\ydelta) to (0.25-\xdelta,2-\ydelta) to [out=-90, in=155] (1,1) to [out=-65, in=90]  (1.75-\xdelta,-\ydelta) to (0.25-\xdelta-\side,-\ydelta);
\draw[bnd] (0.25-\xdelta,2-\ydelta) to [out=-90, in=155] (1,1) to [out=-65, in=90]  (1.75-\xdelta,-\ydelta);
\node [blob] at (1,1) {$B$};
\node [dimension,left] at (0.25-\xdelta-\side,1-\ydelta) {$n$};
\end{tz} 
&
\begin{tz}[string]
\path [use as bounding box] (0.25-\side,-0.1) rectangle (1.75+\side,2);
\path [greenregion, draw] (0.25,0) to [out=90, in=-135] (1,1) to [out=-45, in=90] (1.75,0);
\path [yellowregion, draw] (0.25,2) to [out=-90, in=135] (1,1) to [out=45, in=-90] (1.75,2);
\path [redregion] (0.25,0) to [out=90, in=-135] (1,1) to [out=135, in=-90] (0.25,2) to (0.25 -\side,2) to (0.25-\side,0);
\path[blueregion] (1.75,0) to [out=90, in=-45] (1,1) to [out=45, in=-90] (1.75,2) to (1.75+\side,2) to (1.75+\side,0);
\path[yellowregion,fill opacity= 0.6] (0.25-\xdelta-\side,-\ydelta) rectangle (1.75-\xdelta+\side,2-\ydelta);
\node [blob] at (1,1) {$B$};
\node [dimension, left] at (0.25-\xdelta-\side,1-\ydelta) {$n$};
\end{tz} 
\\[5pt]\nonumber
\text{Half-plane} & \text{Full-plane}
\\[-3pt]\nonumber
\text{controlled family} & \text{controlled family}
\end{calign}

\vspace{-10pt}
\caption{Quantum structures and their associated biunitary types.\label{fig:biunitarytypes}}
\end{figure}

Except for Sections~\ref{sec:diagonalcomposition},~\ref{sec:equivalence} and the discussion of controlled families and interchangers in Section~\ref{sec:characterizing}, all wires in the following diagrams are either standard boundaries \eqref{eq:boundary} or Hilbert spaces that do not bound any region. 

\paragraph{Dimensional constraints.} For a linear map $f:H \to J$ to be unitary imposes a certain algebraic constraint on the dimensions of $H$ and $J$; namely, $\dim(H) = \dim(J)$. For a vertex of type \eqref{eq:2cellbiunitary} to be biunitary similarly induces certain constraints on the allowed labels for the surrounding regions and wires.

In all cases, these constraints are easily identified and solved for. For example, consider the following vertex $U$ and its clockwise quarter rotation:
\begin{calign}
\label{eq:dimensioncounting}
\begin{tz}[string]
\path [use as bounding box] (-0.25,0) rectangle +(2.5,2);
\path[blueregion] (0.25,0) to [out=90, in=-135] (1,1) to [out= 135, in=-90] (0.25,2) to (0.25-\side,2) to (0.25-\side,0);
\draw[bnd] (0.25,0) to [out=90, in=-135] (1,1) to [out=135, in=-90] (0.25,2);
\path[draw, string] (1.75,0) to [out=90, in=-45] (1,1) to [out=45, in=-90] (1.75,2);
\node [dimension,below left] at (0.2,1.75) {$\smash{p}$};
\node [dimension, above right] at (1.8,0.05) {$m$};
\node [dimension, below right] at (1.8,1.75) {$\smash{n}$};
\node [blob] at (1,1) {$U$};
\end{tz}
&
\begin{tz}[string, xscale=0.5,yscale=-1]
\path[blueregion,bnd] (0.25,0) to [out=up, in=-135] (1,1) to [out=135, in=right] (-0.35, 1.7) to [out=left, in=up] (-1.25,1) to (-1.25,0);
\path[draw, string] (1.75,2) to [out=down, in=45] (1,1) to [out= -45, in= left] (2.3,0.3) to [out=right, in=down] (3.25,1) to (3.25,2);
\node [blob] at (1,1) {$U$};
\node[dimension] at (-0.5, 0.3) {$\smash{p}$};
\node[above right, dimension] at (1.8,1.95) {$m$};
\node[above right, dimension] at (3.3,1.95) {$n$};
\end{tz} \end{calign}
Here, $n,m$ and $p$ denote the dimensions of the corresponding region or wire, respectively. For the first of these to be unitary requires that $n=m$, while for the second to be proportional to a unitary requires  $p = nm$. By \autoref{thm:rotationfactor}, for $U$ to be biunitary, we therefore require $(n,m,p)=(n,n,n^2)$, and the space of allowed types is parameterized by a single natural number. In a similar way, for the rest of this section, we will always label biunitaries by their allowed dimensions.

\paragraph{Hadamard matrices.} Hadamard matrices were identified by Jones to be characterized in terms of biunitarity~\cite{Jones:1999}. Complex Hadamard matrices play an important role in mathematical physics and quantum information theory~\cite{Durt:2010}; in particular, they encode the data of a basis which is unbiased with respect to the computational basis.
\begin{definition} \label{def:Hadamard}A \textit{Hadamard matrix} is a matrix $H\in \mathrm{Mat}_n(\mathbb{C})$ with the following properties, for $i,j\in [n]$:
\begin{align}
\label{eq:hadamard1}
H_{i,j} \overline H _{i,j}&= 1
\\
\label{eq:hadamard2}
\textstyle\sum_k H_{i,k}^{} \overline H_{j,k} &= \delta_{i,j} n
\\
\label{eq:hadamard3}
\textstyle\sum_k \overline H_{k,i}H_{k,j}^{}  &= \delta_{i,j} n
\end{align}
Properties~\eqref{eq:hadamard2} and\eqref{eq:hadamard3} are equivalent, but we include them both for completeness.
\end{definition}

The biunitary characterization of Hadamard matrices is due to Jones in the setting of the spin model planar algebra, which our mathematical setup generalizes. It was shown in \cite[Theorem 4.5]{Vicary:2012hq} that this characterization is equivalent to that of Coecke and Duncan in terms of interacting Frobenius algebras~\cite{Coecke:2008}.
\begin{proposition}[Jones {\cite[Section 2.11]{Jones:1999}}]
\label{prop:hadamardbiunitary}
Hadamard matrices of dimension $n$ correspond to biunitaries of the following type:
\begin{equation}
\label{eq:typeHadamard}
\begin{tz}[string]
\path [use as bounding box] (-0.25,0) rectangle +(2.5,2.2);
\path[blueregion] (0.25,0) to [out=90, in=-135] (1,1) to [out= 135, in=-90] (0.25,2) to (0.25-\side,2) to (0.25-\side,0);
\draw[bnd] (0.25,0) to [out=90, in=-135] (1,1) to [out=135, in=-90] (0.25,2);
\path[blueregion] (1.75,0) to [out=90, in=-45] (1,1) to [out=45, in=-90] (1.75,2) to (1.75+\side,2) to (1.75+\side,0);
\draw[bnd] (1.75,0) to [out=90, in=-45]  (1,1) to [out= 45, in=-90] (1.75,2);
\node [dimension, below right] at (1.8,1.75) {$\smash{n}$};
\node [dimension,below left] at (0.2,1.75) {$\smash{n}$};
\node [blob] at (1,1) {$H$};
\end{tz}
\end{equation}
\end{proposition}
\begin{proof} 
A vertex of type \eqref{eq:typeHadamard} represents a family of scalars $H_{i,j}$ controlled by $i,j\in [n]$:
\begin{equation}
\nonumber
\begin{tz}[string]
\path [use as bounding box] (-0.25,0) rectangle +(2.5,2.);
\path[blueregion] (0.25,0) to [out=90, in=-135] (1,1) to [out= 135, in=-90] (0.25,2) to (0.25-\side,2) to (0.25-\side,0);
\draw[bnd] (0.25,0) to [out=90, in=-135] (1,1) to [out=135, in=-90] (0.25,2);
\path[blueregion] (1.75,0) to [out=90, in=-45] (1,1) to [out=45, in=-90] (1.75,2) to (1.75+\side,2) to (1.75+\side,0);
\draw[bnd] (1.75,0) to [out=90, in=-45]  (1,1) to [out= 45, in=-90] (1.75,2);
\node [dimension, below left] at (1.7+\side,1.75) {$\smash{j{:}n}$};
\node [dimension,below right] at (0.3-\side,1.75) {$\smash{i{:}n}$};
\node [blob] at (1,1) {$H$};
\end{tz}
 \hspace{35pt} \leftrightsquigarrow \hspace{30pt} 
\begin{tz}[string]
\node [blob] at (0,0) {$H_{i,j}$};
\end{tz}
\end{equation}
The first vertical unitarity equation corresponds to the following equality of controlled families:
\def\scl{0.75}
\[
\label{eq:Hvertunitary}
\begin{tz}[string,scale=\scl]
\path[blueregion] (-0.75, 0) rectangle (0, 4);
\path[blueregion] (1.5,0) rectangle (2.25,4);
\draw[bnd] (0,0) to +(0,4);
\draw[bnd] (1.5,0) to +(0,4);
\node[dimension, below left] at (2.2,3.65) {$\smash{j{:}n}$};
\node[dimension, below right] at (-0.7,3.65) {$\smash{i{:}n}$};
\end{tz}
\eqgap=\eqgap
\begin{tz}[scale=\scl]
\draw [blueregion] (0,0) to [out=up, in=-135] (0.75,1) to [out=135, in=down] (0,2) to [out=up, in=-135] (0.75,3) to [out=135, in=down] (0,4) to (-0.75,4) to (-0.75,0);
\draw [bnd] (0,0) to [out=up, in=-135] (0.75,1) to [out=135, in=down] (0,2) to [out=up, in=-135] (0.75,3) to [out=135, in=down] (0,4);
\draw [blueregion] (1.5,0) to [out=up, in=-45] (0.75,1) to [out=45, in=down] (1.5,2) to [out=up, in=-45] (0.75,3) to [out=45, in=down] (1.5,4) to (2.25,4) to (2.25,0);
\draw [bnd] (1.5,0) to [out=up, in=-45] (0.75,1) to [out=45, in=down] (1.5,2) to [out=up, in=-45] (0.75,3) to [out=45, in=down] (1.5,4);
\node [blob] at (0.75,1) {$H$};
\node [blob] at (0.75,3) {$H^\dag$};
\node[dimension, below left] at (2.2,3.65) {$\smash{j{:}n}$};
\node[dimension, below right] at (-0.7,3.65) {$\smash{i{:}n}$};
\end{tz}
\hspace{35pt} \leftrightsquigarrow \hspace{35pt}
1\, \,=\begin{tz}[scale=\scl]
\node [blob,minimum width=21pt] at (0.75,1) {$H_{i,j}$};
\node [blob,minimum width=21pt] at (0.75,3) {$\overline{H}_{i,j}$};
\end{tz}
\hspace{20pt} \forall \, i,j \in [n]
\]
This means that $H_{i,j}\overline{H}_{i,j} = 1$ for all $i,j \in [n]$ which recovers condition~\eqref{eq:hadamard1}. The other vertical composite gives the same condition. For horizontal unitarity, we consider the following equation:
\[
\lambda \, 
\begin{tz}[string,xscale=0.5,scale=\scl]
\path[blueregion] (-1.25,0) rectangle (0.25,4);
\draw[bnd] (-1.25,0) to +(0,4);
\draw[bnd] (0.25,0) to + (0,4); 
\node[dimension, below] at (-0.5, 3.65) {$\smash{j{:}n}$};
\node[dimension, above] at (-0.5, 0.35) {$\smash{i{:}n}$};
\end{tz}
\eqgap = \eqgap
\begin{tz}[string, xscale=0.5,scale=\scl]
\path[blueregion,bnd] (0.25,0) to [out=up, in=-135] (1,1) to [out=135, in=right] (-0.35, 1.7) to [out=left, in=up] (-1.25,1) to (-1.25,0);
\path[blueregion,bnd] (0.25,4) to [out=down, in=135] (1,3) to [out=-135, in=right] (-0.35, 2.3) to [out=left, in=down] (-1.25,3) to (-1.25,4);
\path[blueregion,bnd] (1.75,2) to [out=down, in=45] (1,1) to [out= -45, in= left] (2.3,0.3) to [out=right, in=down] (3.6,2)to [out=up, in=right] (2.3,3.7) to [out=left, in=45] (1,3) to [out= -45, in=up] (1.75,2);
\node [blob] at (1,1) {$H$};
\node[blob] at (1,3) {$H^\dagger$};
\node[dimension, below] at (-0.5, 3.65) {$\smash{j{:}n}$};
\node[dimension, above] at (-0.5, 0.35) {$\smash{i{:}n}$};
\node[dimension] at (2.67, 2) {$k{:}n$};
\end{tz}
\hspace{35pt} \leftrightsquigarrow \hspace{35pt}
\lambda\, \delta_{i,j}\,=\, \sum_{k\in [n]}\,\,\begin{tz}[scale=\scl]
\node [blob,minimum width=21pt] at (1,1) {$H_{i,k}$};
\node [blob,minimum width=21pt] at (1,3) {$\overline{H}_{j,k}$};
\end{tz}
\hspace{20pt} \forall \, i,j \in [n]
\]
In other words, $\sum_k H_{i,k} \overline{H}_{j,k} = \lambda\,\delta_{i,j}$ for all $i,j \in [n]$. Together with~\eqref{eq:hadamard1} this implies that $\lambda = n$ and recovers condition~\eqref{eq:hadamard2}. Similarly, condition~\eqref{eq:hadamard3} is satisfied just when the other horizontal unitarity composite is satisfied.
\end{proof}

Following the argument~\eqref{eq:recoverscalar}, the scalar $\lambda=n$ could have been recovered as follows:
\begin{equation}
\label{eq:recoverscalarH}
\lambda \,\,
\begin{tz}[scale=0.8]
\path[blueregion] (-0.4,0) rectangle (0.4,3);
\draw[bnd] (0.4,0) to (0.4,3);
\end{tz}
\,\superequals{eq:circlen}\,
\lambda\,\,
\begin{tz}[string,scale=0.8] 
\path[blueregion] (0,0) rectangle (2,3);
\draw[bnd] (2,0) to (2,3);
\path[fill = white] (1, 1.5) circle (0.5);
\path[bnd] (1, 1.5) circle (0.5);
\end{tz}
\,\superequals{eq:recoverscalar}\,
\begin{tz}[scale=0.8]
\path[blueregion] (-0.4,0) rectangle (0.4,3);
\draw[bnd] (0.4,0) to (0.4,3);
\path[fill=white] (1.4, 1.5) circle (0.5);
\path[blueregion,bnd] (1.4, 1.5) circle (0.5);
\end{tz}
\,\superequals{eq:circlen}\,
n \,\,
\begin{tz}[scale=0.8]
\path[blueregion] (-0.4,0) rectangle (0.4,3);
\draw[bnd] (0.4,0) to (0.4,3);
\end{tz}
\end{equation}
The same holds for unitary error bases and quantum Latin squares below.

\paragraph{Unitary error bases.} Originally introduced by Knill~\cite{Knill:1996_2}, unitary error bases are ubiquitous in modern quantum information theory. They lie at the heart of {quantum error correcting codes}~\cite{Shor:1996} and procedures such as {superdense coding} and {quantum teleportation}~\cite{Werner:2001}.

\begin{definition}[Knill~\cite{Knill:1996_2}] \label{def:UEB}A \textit{unitary error basis} (UEB) on an $n$\-dimensional Hilbert space $H$ is a collection of unitary matrices $\{ U_a \in \mathrm{U}(H)\,|\, a\in[n^2]\}$, satisfying the following orthogonality property, for $a,b\in [n^2]$:
\begin{equation}\Tr(U_a^\dagger U_b^{} )=n\, \delta_{a,b}\end{equation}
That is, a UEB is an orthogonal basis of the space $\mathrm{End}(H)$ consisting entirely of unitary matrices.
\end{definition}

\noindent 
We denote the $(i,j)$th matrix element of the matrix $U_a$ by $U_{a,i,j} = \left(U_a\right)_{i,j} = \bra{i} U_a \ket{j}$.

\begin{proposition}[V.~{\cite[Theorem 4.2]{Vicary:2012hq}}]
Unitary error bases on an $n$\-dimensional Hilbert space correspond to biunitaries of the following type: 
\begin{equation}\label{eq:UEBbiunitary}
\begin{tz}[string]
\path [use as bounding box] (-0.25,0) rectangle +(2.5,2.2);
\path[blueregion] (0.25,0) to [out=90, in=-135] (1,1) to [out= 135, in=-90] (0.25,2) to (0.25-\side,2) to (0.25-\side,0);
\draw[bnd] (0.25,0) to [out=90, in=-135] (1,1) to [out=135, in=-90] (0.25,2);
\path[draw, string] (1.75,0) to [out=90, in=-45] (1,1) to [out=45, in=-90] (1.75,2);
\draw (1.75,0) to [out=90, in=-45]  (1,1) to [out= 45, in=-90] (1.75,2);
\node [dimension,below left] at (0.2,1.75) {$\smash{n^2}$};
\node [dimension, above right] at (1.8,0.05) {$n$};
\node [dimension, below right] at (1.8,1.75) {$\smash{n}$};
\node [blob] at (1,1) {$U$};
\end{tz}
\end{equation}
\end{proposition}
\begin{proof}

A vertex of type \eqref{eq:UEBbiunitary} represents a family of linear maps $U_a$ controlled by $a\in [n^2]$:
\begin{equation}
\nonumber
\begin{tz}[string]
\path [use as bounding box] (-0.25,0) rectangle +(2.5,2.);
\path[blueregion] (0.25,0) to [out=90, in=-135] (1,1) to [out= 135, in=-90] (0.25,2) to (0.25-\side,2) to (0.25-\side,0);
\draw[bnd] (0.25,0) to [out=90, in=-135] (1,1) to [out=135, in=-90] (0.25,2);
\draw (1.75,0) to [out=90, in=-45]  (1,1) to [out= 45, in=-90] (1.75,2);
\node [dimension,below right] at (0.3-\side,1.75) {$\smash{a{:}n^2}$};
\node [blob] at (1,1) {$U$};
\end{tz}
 \hspace{25pt} \leftrightsquigarrow \hspace{35pt} 
\begin{tz}[string]
\draw (1.75,0) to [out= up, in = down] (1,1) to [out= up, in = down] (1.75,2);
\node [blob] at (1,1) {$U_a$};
\end{tz}
\end{equation}
The first vertical unitarity equation corresponds to the following equality between controlled families:
\def\scl{0.75}
\[
\begin{tz}[string,scale=\scl]
\path[blueregion] (-0.95, 0) rectangle (0, 4);
\draw[bnd] (0,0) to +(0,4);
\draw (1.5,0) to +(0,4);
\node[dimension, below right] at (-0.9,3.65) {$\smash{a{:}n^2}$};
\end{tz}
\eqgap=\eqgap
\begin{tz}[scale=\scl]
\draw [blueregion] (0,0) to [out=up, in=-135] (0.75,1) to [out=135, in=down] (0,2) to [out=up, in=-135] (0.75,3) to [out=135, in=down] (0,4) to (-0.95,4) to (-0.95,0);
\draw [bnd] (0,0) to [out=up, in=-135] (0.75,1) to [out=135, in=down] (0,2) to [out=up, in=-135] (0.75,3) to [out=135, in=down] (0,4);
\draw [string] (1.5,0) to [out=up, in=-45] (0.75,1) to [out=45, in=down] (1.5,2) to [out=up, in=-45] (0.75,3) to [out=45, in=down] (1.5,4);
\node [blob] at (0.75,1) {$U$};
\node [blob] at (0.75,3) {$U^\dag$};
\node[dimension, below right] at (-0.9,3.65) {$\smash{a{:}n^2}$};
\end{tz}
\hspace{35pt} \leftrightsquigarrow \hspace{35pt}
\begin{tz}[string,scale=\scl]
\draw (1.5,0) to +(0,4);
\end{tz} 
\eqgap= \eqgap
\begin{tz}[scale=\scl]
\node [blob, minimum width=16pt] at (0.75,1) {$U_{a}$};
\node [blob, minimum width=16pt] at (0.75,3) {$U_a ^\dagger$};
\draw [string] (1.5,0) to [out=up, in=down] (0.75,1) to [out=up, in=down] (1.5,2) to [out=up, in=down] (0.75,3) to [out=up, in=down] (1.5,4);
\end{tz}\hspace{20pt} \forall a \in [n^2]
\]
Together with the other vertical composite\footnote{Once we have fixed the dimensional constraints as described at the beginning of~\autoref{sec:characterizing}, the two conditions on vertical composition (or horizontal composition, respectively) become equivalent. Strictly speaking, we therefore do not need to verify the `other vertical composite'.}, this implies that the linear maps $U_a$ are unitary for all $a\in [n^2]$. For horizontal unitarity, we consider the following equation:
\[
\lambda \, 
\begin{tz}[string,xscale=0.6,scale=\scl]
\path[blueregion] (-1.25,0) rectangle (0.25,4);
\draw[bnd] (-1.25,0) to +(0,4);
\draw[bnd] (0.25,0) to + (0,4); 
\node[dimension, below] at (-0.42, 3.65) {$\smash{a{:}n^2}$};
\node[dimension, above] at (-0.42, 0.35) {$\smash{b{:}n^2}$};
\end{tz}
\eqgap = \eqgap
\begin{tz}[string, xscale=0.6,scale=\scl]
\path[blueregion,bnd] (0.25,0) to [out=up, in=-135] (1,1) to [out=135, in=right] (-0.35, 1.7) to [out=left, in=up] (-1.25,1) to (-1.25,0);
\path[blueregion,bnd] (0.25,4) to [out=down, in=135] (1,3) to [out=-135, in=right] (-0.35, 2.3) to [out=left, in=down] (-1.25,3) to (-1.25,4);
\path[draw, string] (1.75,2) to [out=down, in=45] (1,1) to [out= -45, in= left] (2.3,0.3) to [out=right, in=down] (3.6,2)to [out=up, in=right] (2.3,3.7) to [out=left, in=45] (1,3) to [out= -45, in=up] (1.75,2);
\node [blob] at (1,1) {$U$};
\node[blob] at (1,3) {$U^\dagger$};
\node[dimension, below] at (-0.42, 3.65) {$\smash{a{:}n^2}$};
\node[dimension, above] at (-0.42, 0.35) {$\smash{b{:}n^2}$};
\end{tz}
\hspace{35pt} \leftrightsquigarrow \hspace{35pt}
\lambda\, \delta_{a,b}\eqgap = \eqgap\begin{tz}[xscale=0.6, scale=\scl]
\path[draw, string] (1.75,2) to [out=down, in=up] (1,1) to [out= down, in= left] (2.3,0.3) to [out=right, in=down] (3.6,2)to [out=up, in=right] (2.3,3.7) to [out=left, in=up] (1,3) to [out= down, in=up] (1.75,2);
\node [blob,minimum width=17pt] at (1,1) {$U_{b}$};
\node [blob,minimum width=17pt] at (1,3) {$U^\dagger_a$};
\end{tz}\hspace{20pt}\forall a,b\in [n^2]
\]
By equation \eqref{eq:trace}, this means that $\Tr(U_a^\dagger U_b) = \lambda\, \delta_{a,b}$ for all $a,b \in [n^2]$. Since all matrices $U_a$ are unitary, it follows that $\lambda = n$. Together with the other horizontal unitarity condition, this implies that the matrices $\frac{1}{\sqrt{n}}U_a$ form an orthonormal basis of $\End(H)$.\end{proof}


\paragraph{Quantum Latin squares.} Quantum Latin squares were introduced by Musto and the second author~\cite{Musto:2015} as generalizations of classical Latin squares, with applications to the construction of unitary error bases. Related constructions were also introduced independently by Banica and Nicoar\u{a}~\cite{Banica:2007}.
\begin{definition}[Musto \& V.~{\cite[Definition 1]{Musto:2015}}] \label{def:QLS}A \textit{quantum Latin square} (QLS) on an $n$\-dimensional Hilbert space $H$ is a square grid of vectors $
\left\{\ket{Q_{a,b}}\in H\,|\,a,b \in [n]\right\}
$
such that each row $\{\ket{Q_{a,b}}|\,b \in [n]\}$ and each column $\{ \ket{Q_{a,b}}|\,a \in [n]\}$ form an orthonormal basis of $H$; for $a,b,c \in [n]$:
\begin{align}
\braket{ Q_{a,b} }{ Q_{a,c}} &= \delta_{b,c}
&
\braket{Q_{a,c}}{Q_{b,c}} &= \delta_{a,b}
\end{align}
\end{definition}

\noindent
We denote the $i$th entry of the vector $\ket{Q_{a,b}}$ by \mbox{$Q_{a,b,i} = \braket{i}{Q_{a,b}} $}.
\begin{proposition}
Quantum Latin squares on an $n$\-dimensional Hilbert space correspond to biunitaries of the following type:
\begin{equation}
\label{eq:QLSbiunitary}
\begin{tz}[xscale=-1, string]
\path [use as bounding box] (0.25-\side,0) rectangle (1.75+\side,2);
\node [blob] at (1,1) {$Q$};
\path [blueregion,bnd] (0.25,0) to [out=up, in=-135] (1,1) to [out=-55, in=up] (1.65,0);
\path [blueregion] (1.75,2) to [out=down, in=45] (1,1) to [out=-35, in=up] (1.85,0) to (1.75+\side,0) to (1.75+\side,2);
\draw [bnd] (1.75,2) to [out=-90, in=45] (1,1) to [out=-35, in=up] (1.85,0);
\draw [string] (0.25,2) to [out=-90, in=135] (1,1);
\node [dimension,below left] at (1.8,1.75) {$\smash{n}$};
\node [dimension,above] at (1,0.05) {$n$};
\node [dimension,below right] at (0.2,1.75) {$\smash{n}$};
\end{tz}
\end{equation}
\end{proposition}

\begin{proof}
A vertex of type~\eqref{eq:QLSbiunitary} represents a family of vectors $\ket{Q_{a,b}}$ controlled by $a,b\in [n]$:
\begin{equation}
\begin{tz}[xscale=-1, string]
\clip (0.25-\side,0) rectangle (1.75+\side,2);
\node [blob] at (1,1) {$Q$};
\path [blueregion,bnd] (0.25,0) to [out=up, in=-135] (1,1) to [out=-55, in=up] (1.65,0);
\path [blueregion] (1.75,2) to [out=down, in=45] (1,1) to [out=-35, in=up] (1.85,0) to (1.75+\side,0) to (1.75+\side,2);
\draw [bnd] (1.75,2) to [out=-90, in=45] (1,1) to [out=-35, in=up] (1.85,0);
\draw [string] (0.25,2) to [out=-90, in=135] (1,1);
\node [dimension,above] at (1,0.05) {$b{:}n$};
\node [dimension,below right] at (1.75+\side,1.75) {$\smash{a{:}n}$};
\end{tz}
 \hspace{35pt} \leftrightsquigarrow \hspace{30pt} 
 \begin{tz}[xscale=-1, string]
 \clip (0.25-\side,0) rectangle (1.75+\side,2);
\draw [string] (0.25,2) to [out=-90, in=135] (1,1);
\node [blob] at (1,1) {$Q_{a,b}$};
\end{tz}
\end{equation}
The first vertical unitarity equation corresponds to the following equality between controlled families:
\def\scl{0.75}
\[
\begin{tz}[string,scale=\scl,xscale=-1]
\path[blueregion] (0.25,0) rectangle (1.65,4);
\path[blueregion] (1.85,0) rectangle (2.75,4);
\draw[bnd] (1.85,0) to +(0,4);
\draw[bnd] (1.65,0) to +(0,4);
\draw[bnd] (0.25,0) to +(0,4);
\node[dimension, above] at (1,0.05) {$b{:}n$};
\node[dimension, below] at (1,3.75) {$\smash{c{:}n}$};
\node[dimension, below right] at (2.7,3.75) {$\smash{a{:}n}$};
\end{tz}
\eqgap = \eqgap 
\begin{tz}[string,scale=\scl,xscale=-1]
\node [blob] at (1,1) {$Q$};
\node [blob] at (1,3) {$Q^\dagger$};
\path [blueregion,bnd] (0.25,0) to [out=up, in=-135] (1,1) to [out=-55, in=up] (1.65,0);
\path [blueregion] (1.85,4) to [out=down, in=35] (1,3) to [out=-45, in=up] (1.75,2) to [out=down, in=45] (1,1) to [out=-35, in=up] (1.85,0) to (1.85,0) to (2.75,0) to (2.75,4);
\draw [bnd] (1.85,4) to [out=down, in=35] (1,3) to [out=-45, in=up] (1.75,2) to [out=down, in=45] (1,1) to [out=-35, in=up] (1.85,0);
\path[blueregion,bnd] (0.25,4) to [out=down, in=135] (1,3) to [out=55, in=down] (1.65,4);
\draw [string] (1,3) to [out=-135, in=up] (0.25,2) to [out=-90, in=135] (1,1);
\node[dimension, above] at (1,0.05) {$b{:}n$};
\node[dimension, below] at (1,3.75) {$\smash{c{:}n}$};
\node[dimension, below right] at (2.7,3.75) {$\smash{a{:}n}$};
\end{tz}
\hspace{35pt} \leftrightsquigarrow \hspace{35pt}
\delta_{b,c} \eqgap = \eqgap
\begin{tz}[string,scale=\scl,xscale=-1]
\draw [string] (1,3) to [out=-135, in=up] (0.25,2) to [out=-90, in=135] (1,1);
\node [blob] at (1,1) {$Q_{a,b}$};
\node [blob] at (1,3) {$Q^\dagger_{a,c}$};
\end{tz}\hspace{20pt} \forall a,b,c\in [n]
\]
This means that $\braket{Q_{a,c}}{Q_{a,b}} =\delta_{b,c}$ for all $a,b,c\in [n]$. Together with the other vertical composite this is equivalent to the fact that the rows $\left\{ \ket{Q_{a,b}}~|~b\in [n]\right\}$ form orthonormal bases. For horizontal unitarity, we consider the following equation:

\[\hspace{-10pt}
\lambda \, \begin{tz}[string,xscale=0.7,scale=\scl]
\path[blueregion] (-1.25,0) rectangle (0.05,4);
\path[blueregion]( 0.45,0) rectangle (3,4);
\draw[bnd] (-1.25,0) to +(0,4);
\draw[bnd] (0.05,0) to +(0,4);
\draw[bnd] (0.45,0) to +(0,4);
\node[dimension, above] at (-0.65,0.05) {$a{:}n$};
\node[dimension, below] at (-0.65,3.75) {$\smash{b{:}n}$};
\node[dimension, below left] at (3,3.75) {$\smash{c{:}n}$};
\end{tz}
\eqgap = \eqgap
\begin{tz}[string, xscale=0.7,scale=\scl]
\path[blueregion,bnd] (0.05,0) to [out=up, in=-145] (1,1) to [out=135, in=right] (-0.35, 1.7) to [out=left, in=up] (-1.25,1) to (-1.25,0);
\path[blueregion,bnd] (0.05,4) to [out=down, in=145] (1,3) to [out=-135, in=right] (-0.35, 2.3) to [out=left, in=down] (-1.25,3) to (-1.25,4);
\path[blueregion] (0.45, 0) to [out=up, in=-125] (1,1)  to [out= -45, in= left] (2.3,0.3) to [out=right, in=down] (3.6,2)to [out=up, in=right] (2.3,3.7) to [out=left, in=45] (1,3) to [out=125, in=down] (0.45, 4) to (5,4) to (5,0);
\path[bnd]  (0.45, 0) to [out=up, in=-125] (1,1)  to [out= -45, in= left] (2.3,0.3) to [out=right, in=down] (3.6,2)to [out=up, in=right] (2.3,3.7) to [out=left, in=45] (1,3) to [out=125, in=down] (0.45, 4);
\draw[string] (1,1) to [out=45, in=down] (1.75,2) to [out=up, in=-45] (1,3);
\node [blob,minimum width = 16pt] at (1,1) {$Q$};
\node[blob,minimum width=16pt] at (1,3) {$Q^\dagger$};
\node[dimension, above] at (-0.65,0.05) {$a{:}n$};
\node[dimension, below] at (-0.65,3.75) {$\smash{b{:}n}$};
\node[dimension, below left] at (5,3.75) {$\smash{c{:}n}$};
\end{tz}
\hspace{25pt} \leftrightsquigarrow \hspace{25pt}
\lambda \, \delta_{a,b} \eqgap = \eqgap
\begin{tz}[string, xscale=0.7,scale=\scl]
\draw[string] (1,1) to [out=45, in=down] (1.75,2) to [out=up, in=-45] (1,3);
\node [blob,minimum width = 16pt] at (1,1) {$Q_{a,c}$};
\node[blob,minimum width=16pt] at (1,3) {$Q^\dagger_{b,c}$};
\end{tz}\hspace{15pt} \forall a,b,c\in [n]
\]
This means that $\braket{Q_{b,c}}{Q_{a,c}} = \lambda\, \delta_{a,b}$ for all $a,b,c\in [n]$. Since all vectors $\ket{Q_{a,b}}$ are normalized, it follows that $\lambda = 1$. Together with the other horizontal unitarity condition this is equivalent to the fact that the columns $\left\{ \ket{Q_{a,b}}~|~a\in [n] \right\}$ are orthonormal bases. \end{proof}

\paragraph{Controlled families.} In quantum information, we often want to describe lists of structures, parameterized by a given index. A standard name for such a list is a controlled family. 
\begin{definition}
For a given quantum structure $X$, an \textit{$n$\-controlled family} is an ordered list of $n$ instances of $X$.
\end{definition}

In index notation, we reserve superscript for controlling indices.
For example, a controlled family of Hadamard matrices would be written as $H^c_{a,b}$, where $c$ iterates through the controlled family and $a$ and $b$ are the actual indices of the Hadamard matrix $H^c$.

\begin{proposition} An $n$\-controlled family of biunitaries of type
\begin{equation}
\label{eq:halfplaneplain}
\begin{tz}[string]
\path [greenregion, draw] (0.25,0) to [out=90, in=-135] (1,1) to [out=-45, in=90] (1.75,0);
\path [yellowregion, draw] (0.25,2) to [out=-90, in=135] (1,1) to [out=45, in=-90] (1.75,2);
\path [redregion] (0.25,0) to [out=90, in=-135] (1,1) to [out=135, in=-90] (0.25,2) to (0.25-\side,2) to (0.25-\side,0);
\path [blueregion] (1.75,0) to [out=90, in=-45] (1,1) to [out=45, in=-90] (1.75,2) to (1.75+\side,2) to (1.75+\side,0);
\node [blob] at (1,1) {$A$};
\end{tz} 
\end{equation}
corresponds to a biunitary of the same type with an additional half-plane or full-plane sheet of dimension $n$ attached:
\begin{calign}
\label{eq:controlledfull}
\begin{tz}[string, scale=1]
\path [use as bounding box] (0.25-\side,-0.1) rectangle (1.75+\side,2);
\path [greenregion, draw] (0.25,0) to [out=90, in=-135] (1,1) to [out=-45, in=90] (1.75,0);
\path [yellowregion, draw] (0.25,2) to [out=-90, in=135] (1,1) to [out=45, in=-90] (1.75,2);
\path [redregion] (0.25,0) to [out=90, in=-135] (1,1) to [out=135, in=-90] (0.25,2) to (0.25-\side,2) to (0.25-\side,0);
\path [blueregion] (1.75,0) to [out=90, in=-45] (1,1) to [out=45, in=-90] (1.75,2) to (1.75+\side,2) to (1.75+\side,0);
\path [yellowregion] (0.25-\xdelta-\side,2-\ydelta) to (0.25-\xdelta,2-\ydelta) to [out=-90, in=155] (1,1) to [out=-65, in=90]  (1.75-\xdelta,-\ydelta) to (0.25-\xdelta-\side,-\ydelta);
\draw[bnd] (0.25-\xdelta,2-\ydelta) to [out=-90, in=155] (1,1) to [out=-65, in=90]  (1.75-\xdelta,-\ydelta);
\node [blob] at (1,1) {$B$};
\node [dimension,left] at (0.25-\xdelta-\side,1-\ydelta) {$n$};
\end{tz} 
&
\begin{tz}[string]
\path [use as bounding box] (0.25-\side,-0.1) rectangle (1.75+\side,2);
\path [greenregion, draw] (0.25,0) to [out=90, in=-135] (1,1) to [out=-45, in=90] (1.75,0);
\path [yellowregion, draw] (0.25,2) to [out=-90, in=135] (1,1) to [out=45, in=-90] (1.75,2);
\path [redregion] (0.25,0) to [out=90, in=-135] (1,1) to [out=135, in=-90] (0.25,2) to (0.25 -\side,2) to (0.25-\side,0);
\path[blueregion] (1.75,0) to [out=90, in=-45] (1,1) to [out=45, in=-90] (1.75,2) to (1.75+\side,2) to (1.75+\side,0);
\path[yellowregion,fill opacity= 0.6] (0.25-\xdelta-\side,-\ydelta) rectangle (1.75-\xdelta+\side,2-\ydelta);
\node [blob] at (1,1) {$B$};
\node [dimension, left] at (0.25-\xdelta-\side,1-\ydelta) {$n$};
\end{tz} 
\end{calign}
\end{proposition}

\begin{proof}
A half-plane biunitary $B$ corresponds to a family of vertices of type \eqref{eq:halfplaneplain} controlled by an index $i\in [n]$:
\[
\begin{tz}[string, scale=0.9]
\path [use as bounding box] (0.25-\side,-0.1) rectangle (1.75+\side,2);
\path [greenregion, draw] (0.25,0) to [out=90, in=-135] (1,1) to [out=-45, in=90] (1.75,0);
\path [yellowregion, draw] (0.25,2) to [out=-90, in=135] (1,1) to [out=45, in=-90] (1.75,2);
\path [redregion] (0.25,0) to [out=90, in=-135] (1,1) to [out=135, in=-90] (0.25,2) to (0.25-\side,2) to (0.25-\side,0);
\path [blueregion] (1.75,0) to [out=90, in=-45] (1,1) to [out=45, in=-90] (1.75,2) to (1.75+\side,2) to (1.75+\side,0);
\path [yellowregion] (0.25-\xdelta-\side,2-\ydelta) to (0.25-\xdelta,2-\ydelta) to [out=-90, in=155] (1,1) to [out=-65, in=90]  (1.75-\xdelta,-\ydelta) to (0.25-\xdelta-\side,-\ydelta);
\draw[bnd] (0.25-\xdelta,2-\ydelta) to [out=-90, in=155] (1,1) to [out=-65, in=90]  (1.75-\xdelta,-\ydelta);
\node [blob] at (1,1) {$B$};
\node [dimension,left] at (0.25-\xdelta-\side,1-\ydelta) {$i{:}n$};
\end{tz} \hspace{35pt} \leftrightsquigarrow \hspace{35pt}
\begin{tz}[string,scale=0.9]
\path [greenregion, draw] (0.25,0) to [out=90, in=-135] (1,1) to [out=-45, in=90] (1.75,0);
\path [yellowregion, draw] (0.25,2) to [out=-90, in=135] (1,1) to [out=45, in=-90] (1.75,2);
\path [redregion] (0.25,0) to [out=90, in=-135] (1,1) to [out=135, in=-90] (0.25,2) to (0.25-\side,2) to (0.25-\side,0);
\path [blueregion] (1.75,0) to [out=90, in=-45] (1,1) to [out=45, in=-90] (1.75,2) to (1.75+\side,2) to (1.75+\side,0);
\node [blob] at (1,1) {$B_i$};
\end{tz}
\]
The first unitarity equation amounts to the following equation of controlled families:
\[\hspace{-5pt}
\begin{tz}[string,scale=0.8]
\path[redregion] (0.25-\side, 0) rectangle (0.25,4) ;
\path[greenregion] (0.25,0) rectangle (1.75, 4);
\path[blueregion] (1.75,0) rectangle (1.75+\side, 4);
\draw (0.25,0) to +(0,4);
\draw (1.75,0) to +(0,4);
\path[yellowregion] (0.25-\xdelta-\side, -\ydelta) rectangle (1.75-\xdelta,4-\ydelta);
\draw[bnd] (1.75-\xdelta, -\ydelta) to +(0,4);
\node [dimension,below left] at (0.25-\xdelta-\side,3.9-\ydelta) {$i{:}n$};
\end{tz}
\, =\hspace{-5pt}
\begin{tz}[string,scale=0.8]
\path [redregion] (0.25,0) to [out=90, in=-135] (1,1) to [out=135, in=-90] (0.25,2) to [out=up, in=-135] (1,3) to [out=135, in=down] (0.25,4) to (0.25-\side,4) to (0.25-\side,0);
\path [blueregion] (1.75,0) to [out=90, in=-45] (1,1) to [out=45, in=-90] (1.75,2) to [out=up, in=-45] (1,3) to [out=45, in=down] (1.75,4) to (1.75+\side,4) to (1.75+\side,0);
\path [greenregion,draw] (0.25,0) to [out=90, in=-135] (1,1) to [out=-45, in=90] (1.75,0);
\draw [greenregion,draw]  (1.75,4) to [out=down, in=45] (1,3) to [out=135, in=down] (0.25,4);
\path [yellowregion,draw] (1,3) to [out=-135, in=up] (0.25,2) to [out=-90, in=135] (1,1) to [out=45, in=-90] (1.75,2) to [out=up, in=-45] (1,3);
\path [yellowregion] (0.25-\xdelta-\side,4-\ydelta) to (1.55,4-\ydelta) to [out=down, in=65] (1,3) to [out=-145, in=up] (0.05,2) to [out=-90, in=145] (1,1) to [out=-55, in=90]  (1.55,-\ydelta) to (0.25-\xdelta-\side,-\ydelta);
\draw [bnd] (1.55,4-\ydelta) to [out=down, in=65] (1,3) to [out=-145, in=up] (0.05,2) to [out=-90, in=145] (1,1) to [out=-55, in=90]  (1.55,-\ydelta);
\node [blob] at (1,1) {$B$};
\node [blob] at (1,3) {$B^\dag$};
\node [dimension,below left] at (0.25-\xdelta-\side,3.9-\ydelta) {$i{:}n$};
\end{tz}
\hspace{10pt} \leftrightsquigarrow \hspace{10pt}
\begin{tz}[string,scale=0.8]
\path[redregion] (0.25-\side, 0) rectangle (0.25,4) ;
\path[greenregion] (0.25,0) rectangle (1.75, 4);
\path[blueregion] (1.75,0) rectangle (1.75+\side, 4);
\draw (0.25,0) to +(0,4);
\draw (1.75,0) to +(0,4);
\end{tz}
\,=\,
\begin{tz}[string,scale=0.8]
\path [redregion] (0.25,0) to [out=90, in=-135] (1,1) to [out=135, in=-90] (0.25,2) to [out=up, in=-135] (1,3) to [out=135, in=down] (0.25,4) to (0.25-\side,4) to (0.25-\side,0);
\path [blueregion] (1.75,0) to [out=90, in=-45] (1,1) to [out=45, in=-90] (1.75,2) to [out=up, in=-45] (1,3) to [out=45, in=down] (1.75,4) to (1.75+\side,4) to (1.75+\side,0);
\path [greenregion,draw] (0.25,0) to [out=90, in=-135] (1,1) to [out=-45, in=90] (1.75,0);
\draw [greenregion,draw]  (1.75,4) to [out=down, in=45] (1,3) to [out=135, in=down] (0.25,4);
\path [yellowregion,draw] (1,3) to [out=-135, in=up] (0.25,2) to [out=-90, in=135] (1,1) to [out=45, in=-90] (1.75,2) to [out=up, in=-45] (1,3);
\node [blob,minimum width=17pt] at (1,1) {$B_i$};
\node [blob,minimum width=17pt] at (1,3) {$B_i^\dag$};
\end{tz}\hspace{15pt} \forall i \in [n]
\]
Together with the second vertical unitarity equation, this implies that the vertices $B_i$ are unitary for each $i\in [n]$. For horizontal unitarity, we consider the following equation: 

\def\sidew{0.75}
\[\hspace{-12pt}
\lambda \hspace{-10pt}
\begin{tz}[string,xscale=0.5,scale=0.8]
\path[greenregion] (1.75-2.5*\sidew,-2) rectangle (1.75,2);
\path[blueregion] (1.75,-2) rectangle (3.25,2);
\path[yellowregion] (3.25,-2) rectangle (3.25+2*\sidew, 2);
\draw (1.75,-2) to +(0,4);
\draw (3.25,-2) to +(0,4);
\path[yellowregion] (1.75-2*\xdelta-2.5*\sidew,-2-\ydelta) rectangle (1.75-2*\xdelta, 2-\ydelta);
\draw[bnd] (1.75-2*\xdelta,-2-\ydelta) to +(0,4);
\node[dimension, below left] at (1.75-2*\xdelta-2.5*\sidew, 1.9-\ydelta) {$i{:}n$};
\end{tz}
\,=\hspace{-10pt} 
\begin{tz}[string, xscale=0.5,,scale=0.8]
\path [yellowregion] (3.25,-2) to (3.25,-1) to [out=up, in=right] (2.3,-0.3) to [out=left, in=45] (1,-1) to [out=135, in=down] (0.25,0) to [out=up, in=-135] (1,1) to [out=-45, in=left] (2.3,0.3) to [out=right, in=down] (3.25,1) to (3.25,2) to (3.25+2*\sidew, 2) to (3.25+2*\sidew, -2);
\path[greenregion] (1.75,2) to [out=down, in=45] (1,1) to [out=135, in=right] (-0.3, 1.7) to [out=left, in=up] (-1.6,0) to [out=down, in=left] (-0.3,-1.7) to [out=right, in=-135] (1,-1) to [out= -45, in=up] (1.75,-2) to (-1.6-1.5*\sidew,-2) to (-1.6-1.5*\sidew,2);
\path[redregion,draw, string] (0.25,0) to [out=up, in=-135] (1,1) to [out=135, in=right] (-0.3, 1.7) to [out=left, in=up] (-1.6,0)  to [out=down, in=left] (-0.3, -1.7) to [out=right, in=-135] (1,-1) to [out=135, in=down] (0.25,0);
\path[blueregion,draw,string] (1.75,-2) to [out=up, in=-45] (1,-1) to [out=45, in=left] (2.3,-0.3) to [out= right, in=up] (3.25,-1) to (3.25,-2);
\path[blueregion,draw, string] (1.75,2) to [out=down, in=45] (1,1) to [out= -45, in= left] (2.3,0.3) to [out=right, in=down] (3.25,1) to (3.25,2);
\path[yellowregion] (1.35, -2-\ydelta) to [out=up, in=-65] (1,-1)  to [out=155, in=down] (-0.25, 0) to [out=up, in=-155] (1,1) to [out=65, in=down] (1.35,2-\ydelta) to (-1.6-2*\xdelta-1.5*\sidew, 2-\ydelta) to (-1.6-2*\xdelta-1.5*\sidew, -2-\ydelta);
\draw[bnd]  (1.35, -2-\ydelta) to [out=up, in=-65] (1,-1)  to [out=155, in=down] (-0.25, 0) to [out=up, in=-155] (1,1) to [out=65, in=down] (1.35,2-\ydelta);;
\node [blob] at (1,1) {$B^\dagger$};
\node[blob] at (1,-1) {$B$};
\node[dimension, below left] at (-1.6-2*\xdelta-1.5*\sidew, 1.9-\ydelta) {$i{:}n$};
\end{tz} 
\hspace{7.5pt} \leftrightsquigarrow \hspace{7.5pt}
\lambda \, 
\begin{tz}[string,xscale=0.5,scale=0.8]
\path[greenregion] (1.75-2*\sidew,-2) rectangle (1.75,2);
\path[blueregion] (1.75,-2) rectangle (3.25,2);
\path[yellowregion] (3.25,-2) rectangle (3.25+2*\sidew, 2);
\draw (1.75,-2) to +(0,4);
\draw (3.25,-2) to +(0,4);
\end{tz}
\,=\,
\begin{tz}[string, xscale=0.5,,scale=0.8]
\path [yellowregion] (3.25,-2) to (3.25,-1) to [out=up, in=right] (2.3,-0.3) to [out=left, in=45] (1,-1) to [out=135, in=down] (0.25,0) to [out=up, in=-135] (1,1) to [out=-45, in=left] (2.3,0.3) to [out=right, in=down] (3.25,1) to (3.25,2) to (3.25+2*\sidew, 2) to (3.25+2*\sidew, -2);
\path[greenregion] (1.75,2) to [out=down, in=45] (1,1) to [out=135, in=right] (-0.3, 1.7) to [out=left, in=up] (-1.6,0) to [out=down, in=left] (-0.3,-1.7) to [out=right, in=-135] (1,-1) to [out= -45, in=up] (1.75,-2) to (-1.6-1.5*\sidew,-2) to (-1.6-1.5*\sidew,2);
\path[redregion,draw, string] (0.25,0) to [out=up, in=-135] (1,1) to [out=135, in=right] (-0.3, 1.7) to [out=left, in=up] (-1.6,0)  to [out=down, in=left] (-0.3, -1.7) to [out=right, in=-135] (1,-1) to [out=135, in=down] (0.25,0);
\path[blueregion,draw,string] (1.75,-2) to [out=up, in=-45] (1,-1) to [out=45, in=left] (2.3,-0.3) to [out= right, in=up] (3.25,-1) to (3.25,-2);
\path[blueregion,draw, string] (1.75,2) to [out=down, in=45] (1,1) to [out= -45, in= left] (2.3,0.3) to [out=right, in=down] (3.25,1) to (3.25,2);
\node [blob,minimum width=17pt] at (1,1) {$B_i^\dagger$};
\node[blob,minimum width=17pt] at (1,-1) {$B_i$};
\end{tz} \hspace{10pt} \forall i \in [n]
\]
This means that the vertices $B_i$ satisfy the first horizontal unitarity equation for each $i\in [n]$. In a similar way, the second horizontal unitarity equation for $B$ corresponds to the second horizontal unitarity equation for the vertices $B_i$.\\
It follows that the half-plane control type corresponds to an indexed family of the underlying biunitary type. The proof for the full-plane control type is similar.\end{proof}

\noindent
By \autoref{lemma:rotate}, we could have put the half-plane controlling sheet in one of 4 different orientations. Furthermore, it makes no difference if the controlling sheet goes in front or behind. We therefore have 8 different half-plane controls and 2 different full plane controls, which we illustrate here for the case of Hadamard biunitaries:

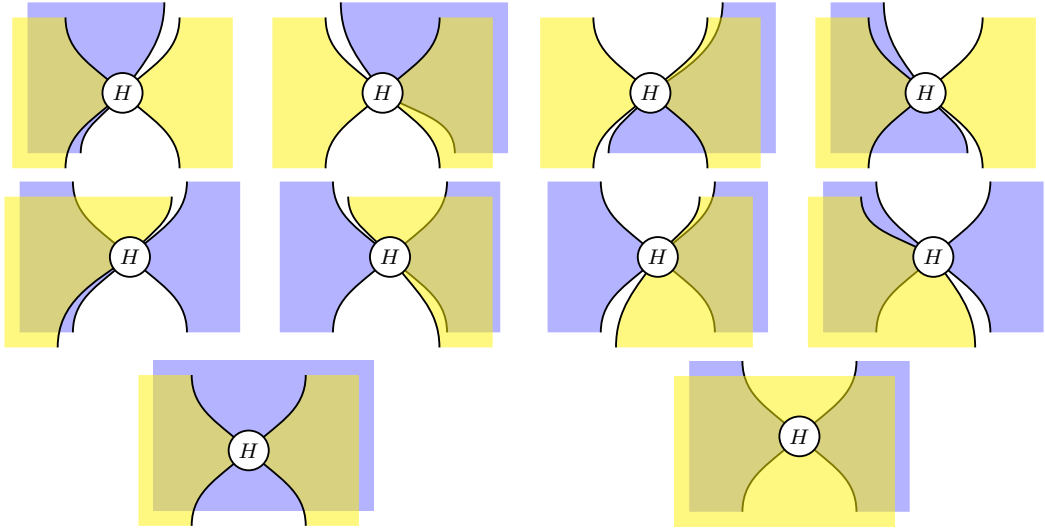
\def \yb{0.2} 
\[
\begin{array}{cccc}
\begin{tz}[string]
\path [blueregion] (0.45,\yb) to [out=90, in=-125] (1,1) to [out=55, in=-90] (1.55,2+\yb) to (0.45-\side,2+\yb) to (0.45-\side,\yb);
\draw[bnd] (0.45,\yb) to [out=90, in=-125] (1,1) to [out=55, in=-90] (1.55,2+\yb);

\path [yellowregion] (0.25,0) to [out=up, in=-135] (1,1) to [out=135, in=down] (0.25,2) to (0.25-\side,2) to (0.25-\side,0);
\draw[bnd](0.25,0) to [out=up, in=-135] (1,1) to [out=135, in=down] (0.25,2);

\path [yellowregion] (1.75,0) to [out=90, in=-45] (1,1) to [out=45, in=-90] (1.75,2) to (1.75+\side,2) to (1.75+\side,0);
\draw[bnd]  (1.75,0) to [out=90, in=-45] (1,1) to [out=45, in=-90] (1.75,2) ;
\node [blob] at (1,1) {$H$};
\end{tz}
&
\begin{tz}[string]
\path [blueregion] (0.45,2+\yb) to [out=down, in=125] (1,1) to [out=-35, in=up] (1.95,\yb) to (1.95+\side,\yb) to (1.95+\side,2+\yb);
\draw[bnd] (0.45,2+\yb) to [out=down, in=125] (1,1) to [out=-35, in=up] (1.95,\yb);

\path [yellowregion] (0.25,0) to [out=up, in=-135] (1,1) to [out=135, in=down] (0.25,2) to (0.25-\side,2) to (0.25-\side,0);
\draw[bnd](0.25,0) to [out=up, in=-135] (1,1) to [out=135, in=down] (0.25,2);

\path [yellowregion] (1.75,0) to [out=90, in=-45] (1,1) to [out=45, in=-90] (1.75,2) to (1.75+\side,2) to (1.75+\side,0);
\draw[bnd]  (1.75,0) to [out=90, in=-45] (1,1) to [out=45, in=-90] (1.75,2) ;
\node [blob] at (1,1) {$H$};
\end{tz}
&
\begin{tz}[string]
\path [blueregion] (0.45,\yb) to [out=90, in=-125] (1,1) to [out=35, in=-90] (1.95,2+\yb) to (1.95+\side,2+\yb) to (1.95+\side,\yb);
\draw[bnd] (0.45,\yb) to [out=90, in=-125] (1,1) to [out=35, in=-90] (1.95,2+\yb) ;

\path [yellowregion] (0.25,0) to [out=up, in=-135] (1,1) to [out=135, in=down] (0.25,2) to (0.25-\side,2) to (0.25-\side,0);
\draw[bnd](0.25,0) to [out=up, in=-135] (1,1) to [out=135, in=down] (0.25,2);

\path [yellowregion] (1.75,0) to [out=90, in=-45] (1,1) to [out=45, in=-90] (1.75,2) to (1.75+\side,2) to (1.75+\side,0);
\draw[bnd]  (1.75,0) to [out=90, in=-45] (1,1) to [out=45, in=-90] (1.75,2) ;
\node [blob] at (1,1) {$H$};
\end{tz}
&
\begin{tz}[string]
\path [blueregion] (0.45,2+\yb) to [out=down, in=125] (1,1) to [out=-55, in=up] (1.55,\yb) to (0.45-\side,\yb) to (0.45-\side,2+\yb);
\draw[bnd] (0.45,2+\yb) to [out=down, in=125] (1,1) to [out=-55, in=up] (1.55,\yb) ;

\path [yellowregion] (0.25,0) to [out=up, in=-135] (1,1) to [out=135, in=down] (0.25,2) to (0.25-\side,2) to (0.25-\side,0);
\draw[bnd](0.25,0) to [out=up, in=-135] (1,1) to [out=135, in=down] (0.25,2);

\path [yellowregion] (1.75,0) to [out=90, in=-45] (1,1) to [out=45, in=-90] (1.75,2) to (1.75+\side,2) to (1.75+\side,0);
\draw[bnd]  (1.75,0) to [out=90, in=-45] (1,1) to [out=45, in=-90] (1.75,2) ;
\node [blob] at (1,1) {$H$};
\end{tz}
\\
\begin{tz}[string]

\path [blueregion] (0.25,0) to [out=up, in=-135] (1,1) to [out=135, in=down] (0.25,2) to (0.25-\side,2) to (0.25-\side,0);
\draw[bnd](0.25,0) to [out=up, in=-135] (1,1) to [out=135, in=down] (0.25,2);

\path [blueregion] (1.75,0) to [out=90, in=-45] (1,1) to [out=45, in=-90] (1.75,2) to (1.75+\side,2) to (1.75+\side,0);
\draw[bnd]  (1.75,0) to [out=90, in=-45] (1,1) to [out=45, in=-90] (1.75,2) ;

\path [yellowregion] (0.05,-\yb) to [out=90, in=-145] (1,1) to [out=55, in=-90] (1.55,2-\yb) to (0.05-\side,2-\yb) to (0.05-\side,-\yb);
\draw[bnd] (0.05,-\yb) to [out=90, in=-145] (1,1) to [out=55, in=-90] (1.55,2-\yb);

\node [blob] at (1,1) {$H$};
\end{tz}
&
\begin{tz}[string]
\path [blueregion] (0.25,0) to [out=up, in=-135] (1,1) to [out=135, in=down] (0.25,2) to (0.25-\side,2) to (0.25-\side,0);
\draw[bnd](0.25,0) to [out=up, in=-135] (1,1) to [out=135, in=down] (0.25,2);

\path [blueregion] (1.75,0) to [out=90, in=-45] (1,1) to [out=45, in=-90] (1.75,2) to (1.75+\side,2) to (1.75+\side,0);
\draw[bnd]  (1.75,0) to [out=90, in=-45] (1,1) to [out=45, in=-90] (1.75,2) ;

\path [yellowregion] (0.45,2-\yb) to [out=down, in=125] (1,1) to [out=-55, in=up] (1.65,-\yb) to (1.65+\side,-\yb) to (1.65+\side,2-\yb);
\draw[bnd] (0.45,2-\yb) to [out=down, in=125] (1,1) to [out=-55, in=up] (1.65,-\yb);

\node [blob] at (1,1) {$H$};
\end{tz}
&
\begin{tz}[string]

\path [blueregion] (0.25,0) to [out=up, in=-135] (1,1) to [out=135, in=down] (0.25,2) to (0.25-\side,2) to (0.25-\side,0);
\draw[bnd](0.25,0) to [out=up, in=-135] (1,1) to [out=135, in=down] (0.25,2);

\path [blueregion] (1.75,0) to [out=90, in=-45] (1,1) to [out=45, in=-90] (1.75,2) to (1.75+\side,2) to (1.75+\side,0);
\draw[bnd]  (1.75,0) to [out=90, in=-45] (1,1) to [out=45, in=-90] (1.75,2) ;

\path [yellowregion] (0.45,-\yb) to [out=90, in=-125] (1,1) to [out=55, in=-90] (1.55,2-\yb) to (1.55+\side,2-\yb) to (1.55+\side,-\yb);
\draw[bnd](0.45,-\yb) to [out=90, in=-125] (1,1) to [out=55, in=-90] (1.55,2-\yb);

\node [blob] at (1,1) {$H$};
\end{tz}
&
\begin{tz}[string]

\path [blueregion] (0.25,0) to [out=up, in=-135] (1,1) to [out=135, in=down] (0.25,2) to (0.25-\side,2) to (0.25-\side,0);
\draw[bnd](0.25,0) to [out=up, in=-135] (1,1) to [out=135, in=down] (0.25,2);

\path [blueregion] (1.75,0) to [out=90, in=-45] (1,1) to [out=45, in=-90] (1.75,2) to (1.75+\side,2) to (1.75+\side,0);
\draw[bnd]  (1.75,0) to [out=90, in=-45] (1,1) to [out=45, in=-90] (1.75,2) ;

\path [yellowregion] (0.05,2-\yb) to [out=down, in=145] (1,1) to [out=-55, in=up] (1.55,-\yb) to (0.05-\side,-\yb) to (0.05-\side,2-\yb);
\draw[bnd](0.05,2-\yb) to [out=down, in=145] (1,1) to [out=-55, in=up] (1.55,-\yb);
\node [blob] at (1,1) {$H$};
\end{tz}
\\
\multicolumn{2}{c}{
\begin{tz}[string]
\path [blueregion] (0.45-\side,\yb) rectangle +(2.2+\side,2);
\path [yellowregion] (0.25,0) to [out=90, in=-135] (1,1) to [out=135, in=-90] (0.25,2) to (0.25-\side,2) to (0.25-\side,0);
\draw[bnd](0.25,0) to [out=90, in=-135] (1,1) to [out=135, in=-90] (0.25,2);
\path [yellowregion] (1.75,0) to [out=90, in=-45] (1,1) to [out=45, in=-90] (1.75,2) to (1.75+\side,2) to (1.75+\side,0);
\draw[bnd](1.75,0) to [out=90, in=-45] (1,1) to [out=45, in=-90] (1.75,2);
\node [blob] at (1,1) {$H$};
\end{tz}}
&
\multicolumn{2}{c}{
\begin{tz}[string]
\path [blueregion] (0.25,0) to [out=90, in=-135] (1,1) to [out=135, in=-90] (0.25,2) to (0.25-\side,2) to (0.25-\side,0);
\draw[bnd] (0.25,0) to [out=90, in=-135] (1,1) to [out=135, in=-90] (0.25,2);
\path [blueregion] (1.75,0) to [out=90, in=-45] (1,1) to [out=45, in=-90] (1.75,2) to (1.75+\side,2) to (1.75+\side,0);
\draw[bnd] (1.75,0) to [out=90, in=-45] (1,1) to [out=45, in=-90] (1.75,2);
\node [blob] at (1,1) {$H$};
\path [yellowregion] (0.05-\side,-\yb) rectangle +(2.2+\side,2);
\end{tz}}
\end{array}
\]
In our pseudo-3d graphical notation, it can be hard to see if a rear sheet is actually connected to a vertex. In our diagrams, we will use the convention that all sheets drawn beneath a vertex are connected to it.

\paragraph{Interchangers.}
The vertex representing the crossing of wires at different depths is called an \textit{interchanger}:
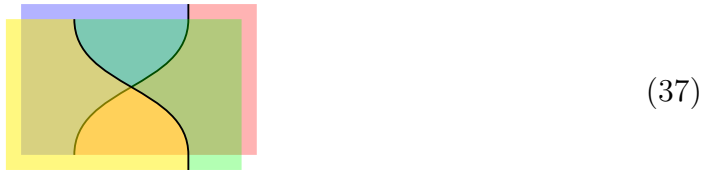
\def\xdeltaw{0.2}
\def \ydeltaw{0.2}
\begin{calign}\label{eq:interchanger} \begin{tz}[string]
\path[blueregion] (0.25, 0 ) to [out= up, in=down] (1.75,2-\ydeltaw) to (1.75,2) to  (0.25-\side, 2) to (0.25-\side,0);
\path[redregion]  (0.25, 0 ) to [out= up, in=down] (1.75,2-\ydeltaw) to (1.75,2) to   (1.75+\xdeltaw+\side, 2) to (1.75+\xdeltaw+\side,0);
\draw   (0.25, 0 ) to [out= up, in=down] (1.75,2-\ydeltaw) to (1.75,2);
\path[greenregion] (1.75,-\ydeltaw) to (1.75,0) to [out= up, in=down] (0.25,2-\ydeltaw) to (1.75+\side,2-\ydeltaw) to (1.75+\side, -\ydeltaw);
\path[yellowregion] (1.75,-\ydeltaw) to (1.75,0) to [out= up, in=down] (0.25,2-\ydeltaw) to (0.25-\xdeltaw-\side,2-\ydeltaw) to (0.25-\xdeltaw-\side,-\ydeltaw);
\draw  (1.75,-\ydeltaw) to (1.75,0) to [out= up, in=down] (0.25,2-\ydeltaw);
\end{tz}
\end{calign}
This is given canonically for all index values as the swap map $H \otimes J \to J \otimes H$.

We now show that interchangers are biunitary, with scalar $\lambda=1$.
\begin{proposition} \label{prop:interchanger}The interchanger \eqref{eq:interchanger} is biunitary.
\end{proposition}
\begin{proof}
Interchangers are unitary, as witnessed by the following equations:
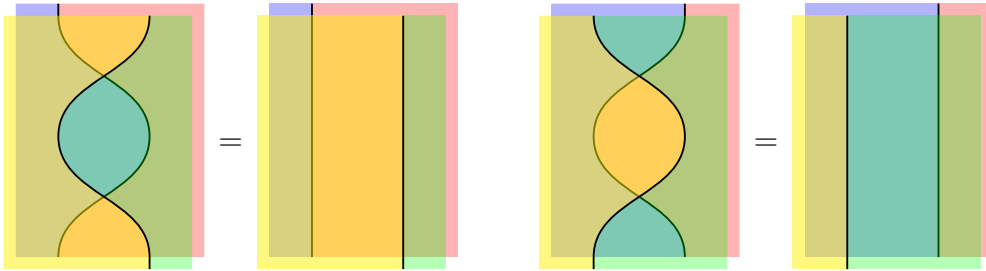
\begin{calign}
\nonumber
\begin{tz}[string,scale=0.8]
\path[redregion] (0.25,0) to [out= up, in=down] (1.75,2) to [out=up, in=down] (0.25,4) to (0.25,4+\ydelta) to (1.75+\xdelta+\side, 4+\ydelta) to (1.75+\xdelta+\side,0);
\path[blueregion] (0.25,0) to [out= up, in=down] (1.75,2) to [out=up, in=down] (0.25,4) to (0.25,4+\ydelta) to (0.25-\side,4+\ydelta) to (0.25-\side,0);
\draw   (0.25,0) to [out= up, in=down] (1.75,2) to [out=up, in=down] (0.25,4) to (0.25,4+\ydelta);
\path[yellowregion] (1.75,-\ydelta) to (1.75,0) to [out=up, in=down] (0.25,2) to [out= up, in=down] (1.75,4) to  (0.25-\xdelta-\side, 4) to (0.25-\xdelta-\side,-\ydelta);
\path[greenregion] (1.75,-\ydelta) to (1.75,0) to [out=up, in=down] (0.25,2) to [out= up, in=down] (1.75,4) to (1.75+\side, 4) to (1.75+\side, -\ydelta);
\draw  (1.75,-\ydelta) to (1.75,0) to [out=up, in=down] (0.25,2) to [out= up, in=down] (1.75,4);
\end{tz}
=
 \begin{tz}[string,scale=0.8]
 \path[blueregion] (0.25-\side, 0) rectangle (0.25, 4+\ydelta);
 \path[redregion] (0.25,0) rectangle (1.75+\xdelta+\side, 4+\ydelta);
 \draw (0.25,0) to (0.25,4+\ydelta);
 \path[greenregion] (1.75,-\ydelta) rectangle (1.75+\side, 4);
 \path[yellowregion] (0.25-\xdelta-\side, -\ydelta) rectangle (1.75,4);
 \draw (1.75,-\ydelta) to (1.75,4);
\end{tz}
&
\begin{tz}[string,scale=0.8]
\path[redregion] (1.75,0) to [out= up, in=down] (0.25,2) to [out=up, in=down] (1.75,4) to (1.75,4+\ydelta) to (1.75+\xdelta+\side, 4+\ydelta) to (1.75+\xdelta+\side,0);
\path[blueregion] (1.75,0) to [out= up, in=down] (0.25,2) to [out=up, in=down] (1.75,4) to (1.75,4+\ydelta) to (0.25-\side,4+\ydelta) to (0.25-\side,0);
\draw  (1.75,0) to [out= up, in=down] (0.25,2) to [out=up, in=down] (1.75,4) to (1.75,4+\ydelta);
\path[yellowregion] (0.25,-\ydelta) to (0.25,0) to [out=up, in=down] (1.75,2) to [out= up, in=down] (0.25,4) to  (0.25-\xdelta-\side, 4) to (0.25-\xdelta-\side,-\ydelta);
\path[greenregion] (0.25,-\ydelta) to (0.25,0) to [out=up, in=down] (1.75,2) to [out= up, in=down] (0.25,4) to (1.75+\side, 4) to (1.75+\side, -\ydelta);
\draw  (0.25,-\ydelta) to (0.25,0) to [out=up, in=down] (1.75,2) to [out= up, in=down] (0.25,4);
\end{tz}
=
 \begin{tz}[string,scale=0.8]
 \path[blueregion] (0.25-\side, 0) rectangle (1.75, 4+\ydelta);
 \path[redregion] (1.75,0) rectangle (1.75+\xdelta+\side, 4+\ydelta);
 \draw (1.75,0) to (1.75,4+\ydelta);
 \path[greenregion] (0.25,-\ydelta) rectangle (1.75+\side, 4);
 \path[yellowregion] (0.25-\xdelta-\side, -\ydelta) rectangle (0.25,4);
 \draw (0.25,-\ydelta) to (0.25,4);
\end{tz}
\end{calign}
They are also horizontally unitary, and thus biunitary, as witnessed by the following:
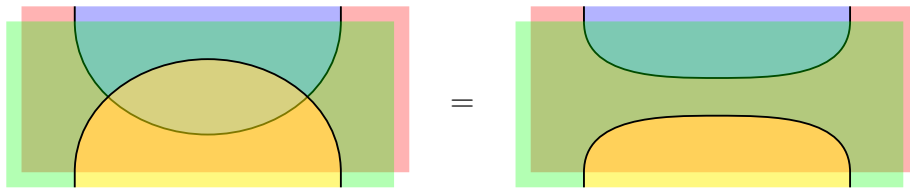
\begin{calign}
\nonumber
\begin{tz}[string]
\path[redregion] (0.25,2+\ydeltaw) to (0.25,2) to  [out= down, in= left] (2,0.5)  to [out = right, in = down] (3.75,2) to (3.75,2+\ydeltaw) to (3.75+\xdeltaw+\side,2+\ydeltaw) to (3.75+\xdeltaw+\side,0) to (0.25-\side, 0) to (0.25-\side, 2+\ydeltaw);
\path[blueregion,draw]  (0.25,2+\ydeltaw) to (0.25,2) to  [out= down, in= left] (2,0.5)  to [out = right, in = down] (3.75,2) to (3.75,2+\ydeltaw) ;
\path[greenregion] (0.25,-\ydeltaw) to (0.25,0) to [out = up, in = left] (2, 1.5) to [out = right, in = up] (3.75,0) to (3.75, -\ydeltaw) to (3.75+\side, -\ydeltaw) to (3.75+\side, 2) to (0.25-\xdeltaw-\side,2) to (0.25-\xdeltaw-\side, -\ydeltaw);
\path[yellowregion,draw]  (0.25,-\ydeltaw) to (0.25,0) to [out = up, in = left] (2, 1.5) to [out = right, in = up] (3.75,0) to (3.75, -\ydeltaw) ;
\end{tz}
\eqgap = \eqgap
\begin{tz}[string]
\path[redregion] (0.25,2+\ydeltaw) to (0.25,2) to  [out= down, in= left] (2,1.25)  to [out = right, in = down] (3.75,2) to (3.75,2+\ydeltaw) to (3.75+\xdeltaw+\side,2+\ydeltaw) to (3.75+\xdeltaw+\side,0) to (0.25-\side, 0) to (0.25-\side, 2+\ydeltaw);
\path[blueregion,draw]  (0.25,2+\ydeltaw) to (0.25,2) to  [out= down, in= left] (2,1.25)  to [out = right, in = down] (3.75,2) to (3.75,2+\ydeltaw) ;
\path[greenregion] (0.25,-\ydeltaw) to (0.25,0) to [out = up, in = left] (2, 0.75) to [out = right, in = up] (3.75,0) to (3.75, -\ydeltaw) to (3.75+\side, -\ydeltaw) to (3.75+\side, 2) to (0.25-\xdeltaw-\side,2) to (0.25-\xdeltaw-\side, -\ydeltaw);
\path[yellowregion,draw]  (0.25,-\ydeltaw) to (0.25,0) to [out = up, in = left] (2, 0.75) to [out = right, in = up] (3.75,0) to (3.75, -\ydeltaw) ;
\end{tz}
\end{calign}
The other horizontal unitarity equation follows similarly.
\end{proof}

\section{Biunitary composition}
\label{sec:composition}

The results of this section are all corollaries of the following simple idea.
\begin{theorem}
\label{thm:maintheorem}
Arbitrary finite diagonal composites of biunitaries are again biunitary.
\end{theorem}
\noindent
Since we have established in \autoref{sec:biunitarity} that biunitaries of various types correspond to different quantum structures, \autoref{thm:maintheorem} suggests the possibility of building new quantum structures from existing ones by diagonal composition. In \autoref{sec:diagonalcomposition}, we demonstrate that binary diagonal composites of biunitaries are again biunitary. We then consider the problem of diagonally composing the biunitaries corresponding to Hadamard matrices, quantum Latin squares, unitary error bases and controlled families to produce other such structures, investigating binary composites in \autoref{sec:binaryquantum}, ternary composites in \autoref{sec:ternaryquantum}, and higher composites in \autoref{sec:higherquantum}. In \autoref{sec:infinity}, we argue that our methods gives rise to an infinite number of genuinely distinct constructions.

A \textit{planar tiling}~\cite{Dawson:1993} is a partition of a rectangle by a finite number of rectangles, and gives the correct structure to describe the possible forms of an arbitrary finite diagonal composite of biunitaries.\footnote{In particular, this implies that biunitaries can be organized as a double category~\cite{Dawson:1993}.} This notation closely resembles Ocneanu's original paragroup notation for biunitaries and their composition~\cite{Ocneanu:1989}. The following are examples of planar tilings, which we always draw in a diagonal fashion to better match the biunitary pictures:
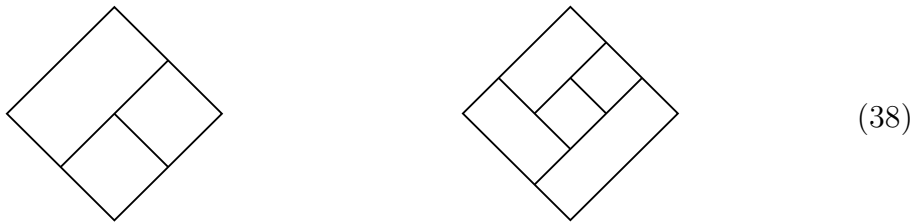
\ignore{ \emph{planar tiling} is useful to describe arbitrary diagonal composites. 
\begin{equation}
\begin{tz}[string,scale=0.7,rotate=-45]
\draw(0,0) rectangle (2,2);
\node[rotate=-45] at (2.5,1){$\leadsto$};
\node[rotate=-135] at (1,-0.5) {$\leadsto$};
\draw (3,0) rectangle (5,2);
\draw (4,0) to +(0,2);
\draw (0,-3) rectangle (2,-1);
\draw (0,-2) to +(2,0);
\end{tz}
\end{equation}
}
\begin{calign}
\begin{tz}[string,rotate=-45]
\draw (0,0) rectangle (2,2);
\draw (1,0) to +(0,2);
\draw (1,1) to (2,1);
\end{tz}
&
\begin{tz}[string,rotate=-45,scale=2/3]
\draw (0,-1) rectangle (3,2);
\draw (1,0) to +(0,2);
\draw (1,1) to (2,1);
\draw (0,0) to +(2,0);
\draw (2,2) to +(0,-3);
\end{tz}
\end{calign}
 For the more complicated biunitary composites in \autoref{fig:ternary}, we give the corresponding planar tiling to make the structure clear.

\subsection{Diagonal composition}
\label{sec:diagonalcomposition}

It is straightforward to see that the diagonal composite of two biunitaries is again biunitary.
\begin{theorem}\label{thm:composition}
Let $U$, $V$ and $W$ be biunitaries of the following types:
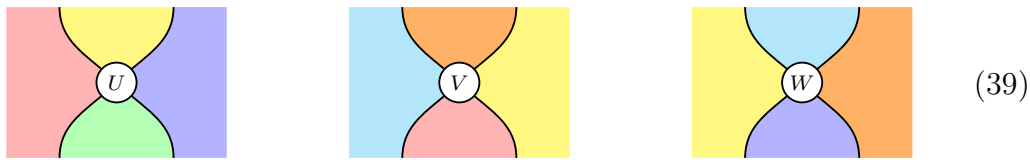
\begin{calign}
\begin{tz}[string]
\path[redregion] (0.25,0) to [out=up, in=-135] (1,1)to [out= 135, in=down] (0.25,2) to (0.25-\side,2) to (0.25-\side,0);
\path[blueregion] (1.75,0) to [out=90, in=-45]  (1,1) to [out= 45, in=-90] (1.75,2) to (1.75+\side,2) to (1.75+\side,0);
\path[greenregion,draw] (0.25,0) to [out=90, in=-135] (1,1)to [out=-45, in=90] (1.75,0);
\path[yellowregion,draw] (0.25,2) to [out=-90, in=135]  (1,1) to   [out=45, in=-90] (1.75,2);
\node[blob] at (1,1) {$U$};
\end{tz}
&
\begin{tz}[string]
\path[cyanregion] (0.25,0) to [out=up, in=-135] (1,1)to [out= 135, in=down] (0.25,2) to (0.25-\side,2) to (0.25-\side,0);
\path[yellowregion] (1.75,0) to [out=90, in=-45]  (1,1) to [out= 45, in=-90] (1.75,2) to (1.75+\side,2) to (1.75+\side,0);
\path[redregion,draw] (0.25,0) to [out=90, in=-135] (1,1)to [out=-45, in=90] (1.75,0);
\path[orangeregion,draw] (0.25,2) to [out=-90, in=135]  (1,1) to   [out=45, in=-90] (1.75,2);
\node[blob] at (1,1) {$V$};
\end{tz}
&
\begin{tz}[string]
\path[yellowregion] (0.25,0) to [out=up, in=-135] (1,1)to [out= 135, in=down] (0.25,2) to (0.25-\side,2) to (0.25-\side,0);
\path[orangeregion] (1.75,0) to [out=90, in=-45]  (1,1) to [out= 45, in=-90] (1.75,2) to (1.75+\side,2) to (1.75+\side,0);
\path[blueregion,draw] (0.25,0) to [out=90, in=-135] (1,1)to [out=-45, in=90] (1.75,0);
\path[cyanregion,draw] (0.25,2) to [out=-90, in=135]  (1,1) to   [out=45, in=-90] (1.75,2);
\node[blob] at (1,1) {$W$};
\end{tz}
\end{calign}
Then the following diagonal composites are biunitary, with respect to the indicated partitions of the input and output wires:
\begin{calign}\label{eq:biunitaryComposite}
\begin{tz}[string]
\path[blueregion] (1.75,0) to [out=90, in=-45] (1,1) to [out=45, in=-90]  (1.25+\hrt,2.5) to (1.25+\hrt+\side,2.5) to (1.25+\hrt+\side,0);
\path[cyanregion] (0.25-\hrt,0) to [out=90, in=-135] (0.5,1.5) to [out=135, in=-90] (-0.25,2.5) to (-0.25-\side,2.5) to (-0.25-\side,0);
\path[redregion,draw] (0.25,0) to [out=90, in=-135] (1,1) to (0.5,1.5) to [out=-135, in=90] (0.25-\hrt,0);
\path[ orangeregion, draw] (-0.25,2.5) to [out=-90, in=135] (0.5,1.5) to [out=45, in=-90] (1.25,2.5);
\path[greenregion,draw] (0.25,0) to [out=90, in=-135] (1,1) to [out=-45, in=90] (1.75,0);
\path[yellowregion,draw] (1.25,2.5) to [out=-90, in=45] (0.5,1.5) to (1,1)to [out=45, in=-90] (1.25+\hrt,2.5);
\node[blob] at (1,1) {$U$};
\node[blob] at (0.5,1.5) {$V$};
\node [rotate=-90] at (1.6,2.75) {$\left\{ \vbox to 0.6cm {} \right.$};
\node [rotate=-90] at (-0.2,2.75) {$\left\{ \vbox to 0.4cm {} \right.$};
\node [rotate=90] at (-0.1,-0.25) {$\left\{ \vbox to 0.6cm {} \right.$};
\node [rotate=90] at (1.75,-0.25) {$\left\{ \vbox to 0.4cm {} \right.$};
\end{tz}
&
\begin{tz}[string]
\path[orangeregion] (1.75+\hrt,0) to [out=90, in=-45] (1.5,1.5) to [out=45, in=-90] (2.25,2.5) to (2.25+\side,2.5) to (2.25+\side,0);
\path[redregion] (0.25,0) to [out=90, in=-135] (1,1) to [out=135, in=-90] (0.75-\hrt, 2.5) to (0.75-\hrt-\side,2.5) to (0.75-\hrt-\side,0);
\path[greenregion,draw] (0.25,0) to [out=90, in=-135] (1,1) to [out=-45, in=90] (1.75,0);
\path[ blueregion,draw] (1.75,0) to [out=90, in=-45] (1,1) to (1.5,1.5) to [out=-45, in=90]  (1.75+\hrt,0);
\path[cyanregion, draw] (0.75,2.5) to [out=-90, in=135] (1.5,1.5) to [out=45, in=-90] (2.25,2.5);
\path[yellowregion,draw] (0.75,2.5) to [out=-90, in=135] (1.5,1.5) to (1,1)to [out=135, in=-90] (0.75-\hrt,2.5);
\node[blob] at (1,1) {$U$};
\node[blob] at (1.5,1.5) {$W$};
\node [rotate=90] at (2.1,-0.25) {$\left\{ \vbox to 0.6cm {} \right.$};
\node [rotate=90] at (0.3,-0.25) {$\left\{ \vbox to 0.4cm {} \right.$};
\node [rotate=-90] at (0.4,2.75) {$\left\{ \vbox to 0.6cm {} \right.$};
\node [rotate=-90] at (2.25,2.75) {$\left\{ \vbox to 0.4cm {} \right.$};
\end{tz} 
\end{calign}
\end{theorem}
\begin{proof}
We will prove that 
\[\begin{tz}[string,scale=1]

\path[redregion,draw] (0.25,0) to [out=90, in=-135] (1,1) to (0.85,1.15) to [out=-135, in=90] (0.25-0.3*\hrt,0);

\path[ blueregion] (1.75,0) to [out=90, in=-45] (1,1) to [out=45, in=-90]  (1.6+0.3*\hrt,2.15) to (1.75+\side,2.15) to (1.75+\side,0);

\path[orangeregion, draw] (0.1,2.15) to [out=-90, in=135] (0.85,1.15) to [out=45, in=-90] (1.6,2.15);

\path[greenregion,draw] (0.25,0) to [out=90, in=-135] (1,1) to [out=-45, in=90] (1.75,0);

\path[yellowregion,draw] (1.6,2.15) to [out=-90, in=45] (0.85,1.15) to (1,1)to [out=45, in=-90] (1.6+0.3*\hrt,2.15);

\path[cyanregion] (0.25-0.3*\hrt,0) to [out=90, in=-135] (0.85,1.15) to [out=135, in=-90] (0.1,2.15) to (0.1-\side,2.15) to (0.1-\side,0);

\node[blob]at (0.925,1.075) {$V {\ast}  U$};
\end{tz}:=
\begin{tz}[string,scale=.86]
\path[blueregion] (1.75,0) to [out=90, in=-45] (1,1) to [out=45, in=-90]  (1.25+\hrt,2.5) to (1.25+\hrt+\side,2.5) to (1.25+\hrt+\side,0);
\path[cyanregion] (0.25-\hrt,0) to [out=90, in=-135] (0.5,1.5) to [out=135, in=-90] (-0.25,2.5) to (-0.25-\side,2.5) to (-0.25-\side,0);
\path[redregion,draw] (0.25,0) to [out=90, in=-135] (1,1) to (0.5,1.5) to [out=-135, in=90] (0.25-\hrt,0);
\path[ orangeregion, draw] (-0.25,2.5) to [out=-90, in=135] (0.5,1.5) to [out=45, in=-90] (1.25,2.5);
\path[greenregion,draw] (0.25,0) to [out=90, in=-135] (1,1) to [out=-45, in=90] (1.75,0);
\path[yellowregion,draw] (1.25,2.5) to [out=-90, in=45] (0.5,1.5) to (1,1)to [out=45, in=-90] (1.25+\hrt,2.5);
\node[blob] at (1,1) {$U$};
\node[blob] at (0.5,1.5) {$V$};
\end{tz}\] is biunitary; the proof for the other composite is completely analogous. The composite $V{\ast\,}U$ is vertically unitary, since it is the vertical composite of two unitary vertices. For horizontal unitarity, consider the anticlockwise rotation of $V {\ast\,} U$:
\def\scalel {0.9} 
\[ 
\def\sidew{1.79*\side}
\begin{tz}[string,xscale=0.5,scale=1.25*\scalel]
\coordinate (L) at (1-0.2*\hrt,1+0.1*\hrt);
\coordinate (R) at (1+0.2*\hrt,1-0.1*\hrt);

\path[redregion,draw] (0.05,0) to [out=90, in=-135] (L) to (R) to [out=-135, in=90]  (0.45,0);

\path[yellowregion,draw] (1.55,2) to [out= -90, in=45] (L) to (R) to [out=45, in=-90]  (1.95,2);

\path[cyanregion,draw] (0.05,0) to [out=up, in=-135] (L) to [out=135, in=right] (-0.5,1.7) to [out=left, in=up] (-1.5,1) to (-1.5,0);

\path[blueregion,draw] (1.95,2) to [out=down, in=45] (R) to [out=-45, in=left] (2.5, 0.3)  to [out= right, in=down] (3.5,1) to (3.5,2);

\path[greenregion] (0.45,0) to [out=up, in=-135] (R) to [out=-45, in=left] (2.5,0.3) to [out= right, in=down] (3.5,1) to (3.5,2) to (3.5 + \sidew, 2) to (3.5+\sidew, 0);

\path[orangeregion] (1.55,2) to [out=down, in=45] (L) to [out= 135, in=right] (-0.5,1.7) to [out= left, in=up] (-1.5,1) to (-1.5,0) to (-1.5 -\sidew, 0) to (-1.5-\sidew, 2);

\node[blob] at (1,1) {$V {\ast\,} U$};
\end{tz}
\,=\,
\def\sidew{2*0.714*\side}
\begin{tz}[string,xscale=0.5,scale=\scalel]
\node[blob] (U) at (1,1){$U$};
\node[blob] (V) at (0,1.5) {$V$};
\path [greenregion] (0.25,0) to [out=up, in=-135] (1,1) to [out=-45, in=left] (2.75,0.3) to [out=right, in=down] (4,1) to (4,2.5) to (4+\sidew,2.5) to (4+\sidew,0);

\path[orangeregion] (0.75,2.5) to [out= down, in=45] (V.center) to [out= 135, in=right] (-1.75, 2.2) to [out=left, in=up] (-3, 1.5) to (-3, 0) to (-3-\sidew, 0) to (-3-\sidew, 2.5);

\path[blueregion] (0.75+2*\hrt,2.5) to [out=down, in=45] (1,1) to [out= -45, in= left] (2.75,0.3) to [out=right, in=down] (4,1) to (4,2.5);
\draw[string] (1,1) to [out= -45, in= left] (2.75,0.3) to [out=right, in=down] (4,1) to (4,2.5);

\path[cyanregion] (0.25-2*\hrt, 0) to [out=up, in=-135] (V.center) to [out= 135, in=right] (-1.75, 2.2) to [out=left, in=up] (-3, 1.5) to (-3, 0);
\draw[string] (V.center) to [out= 135, in=right] (-1.75, 2.2) to [out=left, in=up] (-3, 1.5) to (-3, 0);

\path[yellowregion,draw,string] (0.75,2.5) to [out=down, in=45] (V.center) to (U.center) to [out=45, in=down] (0.75+2*\hrt, 2.5);

\path[redregion,draw,string] (0.25,0) to [out= up, in=-135] (U.center) to (V.center)  to [out= -135, in=up] (0.25-2*\hrt, 0);
\end{tz}
\,\superequals{eq:colouredcupscaps}\,
\def\sidew{2*\side}
\begin{tz}[string, xscale=0.5,scale=0.714*\scalel]
\node[blob] (V) at (1,1) {$V$};
\node[blob] (U) at (6.5,2.5) {$U$};

\path [redregion] (0.25,0) to [out=up, in=-135] (1,1) to [out=-45, in=left] (2.3,0.3) to [out=right, in=left] (5.2,3.2) to [out=right, in=135] (U.center) to [out= -135, in=up]  (5.5,0);

\path[orangeregion] (-1.25,0) to (-1.25,1) to [out= up, in=left] (-0.3,1.7) to [out= right, in=135] (V.center) to [out= 45, in=down] (2,3.5) to (-1.25-\sidew, 3.5) to (-1.25-\sidew, 0);

\path[greenregion] (5.5,0) to [out= up , in=-135] (U.center) to [out=-45, in=left] (7.8,1.8) to [out=right, in=down] (8.75,2.5) to (8.75,3.5) to (8.75+\sidew,3.5) to (8.75+\sidew, 0);

\path[blueregion,draw,string] (7.25, 3.5) to [out=down, in=45] (U.center) to [out=-45, in=left] (7.8,1.8) to [out=right, in=down] (8.75,2.5) to (8.75,3.5);

\path[cyanregion,draw, string] (0.25,0) to [out=up, in=-135] (1,1) to [out=135, in=right] (-0.3, 1.7) to [out=left, in=up] (-1.25,1) to (-1.25,0);

\path[yellowregion] (2, 3.5) to [out=down, in=45] (1,1) to [out= -45, in= left] (2.3,0.3) to [out=right, in=left] (5.2,3.2) to [out=right, in=135] (U.center) to [out= 45, in=down] (7.25,3.5); 
\draw[string] (2, 3.5) to [out=down, in=45] (1,1) to [out= -45, in= left] (2.3,0.3) to [out=right, in=left] (5.2,3.2) to [out=right, in=135] (U.center) to [out=-135, in=up] (5.5,0);
\end{tz}
\]
This is unitary up to a scalar, since by \autoref{thm:rotationfactor} it is the vertical composite of two vertices which are unitary up to a scalar. By \autoref{thm:rotationfactor}, we conclude that $V{\ast\,}U$ is biunitary.
\end{proof}

\noindent
Except for the pinwheel composite\footnote{The \textit{pinwheel composite} is a way to compose five 2-morphisms in a double category, in a way which cannot be described in terms of repeated binary composites.}~\cite{Dawson:1993}, which can be handled separately, this shows that \autoref{thm:maintheorem} holds.

\subsection{Binary composites}
\label{sec:binaryquantum}

We give a number of quantum constructions listed in \autoref{fig:binarydetailed1} and \autoref{fig:binarydetailed2}, each involving the diagonal composite of two biunitaries. Correctness of all these constructions follows as corollaries from \autoref{thm:composition}, and the results of \autoref{sec:characterizing} as summarized in \autoref{fig:biunitarytypes}.

\paragraph{Quantum Latin squares.} We begin by presenting two quantum Latin square constructions. The following construction produces a quantum Latin square from two Hadamard matrices, generalizing \cite[Definition 2.3]{Banica:2007} and \cite[Definition~10]{Musto:2015}.
\begin{figure}
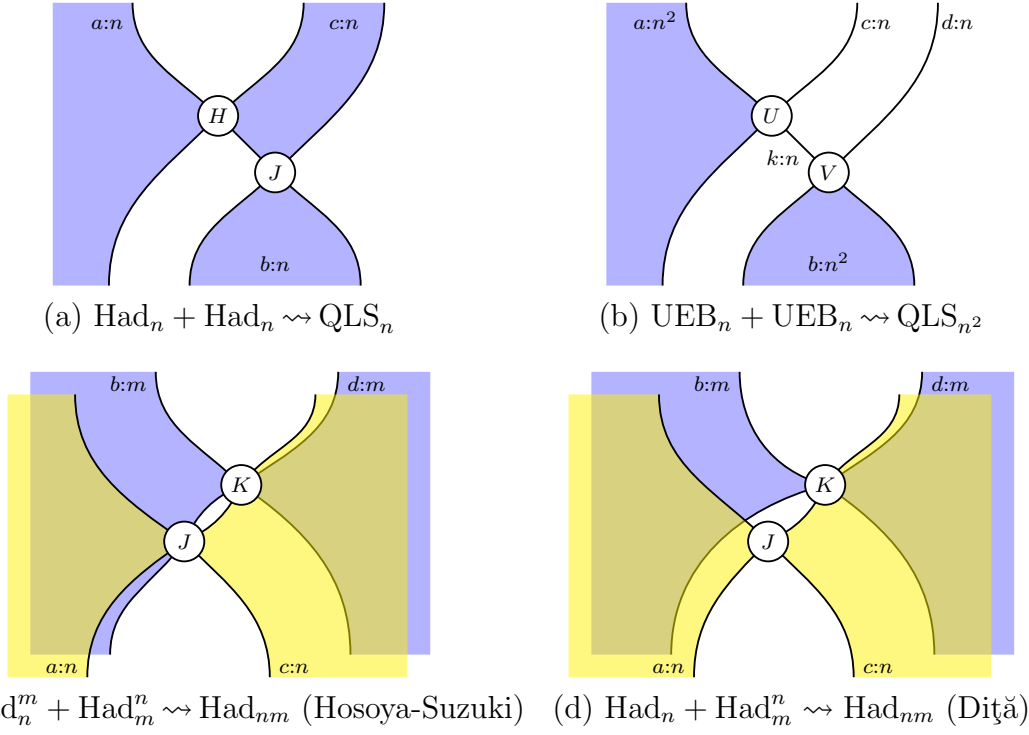

\begin{calign}
\nonumber
\begin{tz}[xscale=-1,string,scale=1.5]
\path[blueregion,bnd] (0.25,0) to [out=90, in=-135] (1,1) to [out=-45, in=90] (1.75,0);
\path[blueregion] (1.75+\hrt,0) to [out=90, in=-45] (1.5,1.5) to [out=45, in=-90] (2.25,2.5) to (2.25+\side,2.5) to (2.25+\side,0);
\draw[bnd] (1.75+\hrt,0) to [out=90, in=-45] (1.5,1.5) to [out=45, in=-90] (2.25,2.5);
\path[blueregion,bnd] (0.75,2.5) to [out=-90, in=135] (1.5,1.5) to (1,1) to [out=135, in=-90] (0.75-\hrt, 2.5);
\node[blob] at (1,1) {$J$};
\node[blob] at (1.5,1.5) {$H$};
\node[below, dimension] at (0.75 -0.5*\hrt,2.3) {$\smash{c{:}n}$};
\node[below left, dimension] at (2.3,2.3) {$\smash{a{:}n}$};
\node[above,dimension] at (1., 0.1) {$b{:}n$};
\end{tz}
& 
\begin{tz}[xscale=-1,string,scale=1.5]
\path[blueregion] (1.75+\hrt,0) to [out=90, in=-45] (1.5,1.5) to [out=45, in=-90] (2.25,2.5) to (2.25+\side,2.5) to (2.25+\side,0);
\draw[bnd] (1.75+\hrt,0) to [out=90, in=-45] (1.5,1.5) to [out=45, in=-90] (2.25,2.5);
\path[blueregion,bnd] (0.25,0) to [out=90, in=-135] (1,1) to [out=-45, in=90] (1.75,0);
\draw[string] (0.75,2.5) to [out=-90, in=135] (1.5,1.5) to (1,1) to [out=135, in=-90] (0.75-\hrt, 2.5);
\node[blob]at (1,1) {$V$};
\node[blob] at (1.5,1.5) {$U$};
\node[dimension, below right] at (0.75, 2.3) {$\smash{c{:}n}$};
\node[dimension, below right] at (0.75-\hrt, 2.3) {$\smash{d{:}n}$};
\node[below left,dimension] at (2.3, 2.3) {$\smash{a{:}n^2}$};
\node[dimension, above] at (1., 0.1) {$b{:}n^2$};
\node[below left,dimension] at (1.23,1.23) {$k{:}n$};
\end{tz}   
\\[-1pt]\nonumber 
\text{(a) }\HAD_n + \HAD_n {\,\leadsto\,} \QLS_n
&
\text{(b) } \UEB_n + \UEB_n {\,\leadsto\,}\QLS_{n^2}
\\[10pt]
\nonumber
\begin{tz}[string,scale=1.5]
\path[blueregion] (0.75,2.5) to [out=-90, in=135] (1.5,1.5) to [out=-155, in=65] (1,1)to [out=-125, in=90] (0.35,0) to (0.35-\side,0) to (0.35-\side,2.5);
\draw[bnd](0.75,2.5) to [out=-90, in=135] (1.5,1.5) to [out=-155, in=65] (1,1)to [out=-125, in=90] (0.35,0) ;
\path[yellowregion] (0.15,-\ydelta) to [out=90, in=-145] (1,1) to [out=145, in=-90] (0.75-\hrt, 2.5-\ydelta) to (0.15-\side,2.5-\ydelta) to (0.15-\side,-\ydelta);
\draw[bnd](0.15,-\ydelta) to [out=90, in=-145] (1,1) to [out=145, in=-90] (0.75-\hrt, 2.5-\ydelta);
\path[blueregion] (2.35,2.5) to [out=-90, in=35] (1.5,1.5) to [out=-35, in=90] (1.75 + \hrt,0) to (1.75+\hrt+\side,0) to (1.75+\hrt+\side,2.5);
\draw[bnd] (2.35,2.5) to [out=-90, in=35] (1.5,1.5) to [out=-35, in=90] (1.75 + \hrt,0);
\path[ yellowregion] (1.75,-\ydelta) to [out=90, in=-45] (1,1) to [out=25, in=-115] (1.5,1.5) to [out=55, in=-90] (2.15,2.5-\ydelta) to (1.55+\hrt+\side,2.5-\ydelta) to (1.55+\hrt+\side,-\ydelta);
\draw[bnd] (1.75,-\ydelta) to [out=90, in=-45] (1,1) to [out=25, in=-115] (1.5,1.5) to [out=55, in=-90] (2.15,2.5-\ydelta);
\node[blob] at (1,1) {$J$};
\node[blob] at (1.5,1.5) {$K$};
\node[below left,dimension] at (0.7,2.36) {$\smash{b{:}m}$};
\node[dimension,below right] at (2.4,2.36) {$\smash{d{:}m}$};
\node[dimension, above right] at (1.8,-0.19) {$\smash{c{:}n}$};
\node[dimension, above left] at (0.1,-0.19) {$\smash{a{:}n}$};
\end{tz}
&
\begin{tz}[string,scale=1.5]
\path[blueregion] (0.75,2.5) to [out=-90, in=165] (1.5,1.5) to [out=-165, in=90] (0.15,0) to (0.15-\side,0) to (0.15-\side,2.5);
\draw[bnd]   (0.75,2.5) to [out=-90, in=165] (1.5,1.5) to [out=-165, in=90] (0.15,0);
\path[yellowregion] (0.35,-\ydelta) to [out=90, in=-135] (1,1) to [out=135, in=-90] (0.75-\hrt, 2.5-\ydelta) to (-0.05-\side,2.5-\ydelta) to (-0.05-\side,-\ydelta);
\draw[bnd](0.35,-\ydelta) to [out=90, in=-135] (1,1) to [out=135, in=-90] (0.75-\hrt, 2.5-\ydelta);
\path[blueregion] (2.35,2.5) to [out=-90, in=35] (1.5,1.5) to [out=-35, in=90] (1.75 + \hrt,0) to (1.75+\hrt+\side,0) to (1.75+\hrt+\side,2.5);
\draw[bnd] (2.35,2.5) to [out=-90, in=35] (1.5,1.5) to [out=-35, in=90] (1.75 + \hrt,0);
\path[ yellowregion] (1.75,-\ydelta) to [out=90, in=-45] (1,1) to [out=25, in=-115] (1.5,1.5) to [out=55, in=-90] (2.15,2.5-\ydelta) to (1.55+\hrt+\side,2.5-\ydelta) to (1.55+\hrt+\side,-\ydelta); 
\draw[bnd](1.75,-\ydelta) to [out=90, in=-45] (1,1) to [out=25, in=-115] (1.5,1.5) to [out=55, in=-90] (2.15,2.5-\ydelta);
\node[blob] at (1,1) {$J$};
\node[blob] at (1.5,1.5) {$K$};
\node[below left,dimension] at (0.7,2.36) {$\smash{b{:}m}$};
\node[dimension,below right] at (2.4,2.36) {$\smash{d{:}m}$};
\node[dimension, above right] at (1.8,-0.19) {$\smash{c{:}n}$};
\node[dimension, above left] at (0.3,-0.19) {$\smash{a{:}n}$};
\end{tz}
\\[-1pt]\nonumber
\text{(c) }\HAD_n^m + \HAD_m^n {\,\leadsto\,} \HAD_{nm} \text{ (Hosoya-Suzuki)} 
&
\text{(d) }\HAD_n + \HAD^n_m \rightsquigarrow \HAD_{nm} \text{ (\Dita)}
\end{calign}
\vspace{-20pt}
\caption{Binary constructions of quantum Latin squares and Hadamard matrices.\label{fig:binarydetailed1}}
\end{figure}%
\begin{cor}[$\HAD_n+\HAD_n\leadsto\QLS_n$]\label{thm:H+H=Q} The construction of \autoref{fig:binarydetailed1}(a) produces an $n$\-dimensional quantum Latin square 
\begin{calign}
\label{eq:H+H=Q}
Q_{a,b,c} = \frac{1}{\sqrt{n}}\,H_{a,c}\,J_{c,b}
\end{calign}
from the following data, with $a,b,c \in [n]$: 
\begin{itemize}
\item $H_{a,c}$ and $J_{c,b} \in \HAD_n$, $n$\-dimensional Hadamard matrices.
\end{itemize}
\end{cor}

\noindent
The factor $\frac{1}{\sqrt{n}}$ arises as described in \autoref{thm:rotationfactor}, since the biunitary $J$ is of rotated Hadamard type. Such a biunitary is a \textit{unitary} matrix; given an ordinary Hadamard matrix, we need to rescale it by a factor of $\frac{1}{\sqrt{n}}$ to obtain such a unitary.

As the first of many such corollaries in this paper, we show here how to make use of this data explicitly. Choosing $n=2$, let $H$ and $J$ be the following Hadamard matrices:
\begin{calign}
H = 
\frac 1 {\sqrt{2}}
\begin{pmatrix}
1 & 1 \\ 1 & -1
\end{pmatrix}
&J = 
\frac 1 {\sqrt{2}}
\begin{pmatrix}
1 & i \\ i & 1
\end{pmatrix}
\end{calign}
Then applying the formula~\eqref{eq:H+H=Q}, and recalling from \autoref{def:QLS} that a 2\-dimensional quantum Latin square $Q$ comprises a 2-by-2 grid of vectors $\ket{Q_{a,b}} \in \C^2$ with $Q_{a,b,c} = \braket{c}{Q_{a,b}}$, we obtain $Q$ explicitly as follows:
\begin{equation}
Q =
\begin{pmatrix}
{\displaystyle \frac 1 {\sqrt 2}}
\begin{pmatrix} 1\\i \end{pmatrix}
&
{\displaystyle \frac 1 {\sqrt 2}}
\begin{pmatrix} 1\\-i \end{pmatrix}
\\[15pt]
{\displaystyle \frac 1 {\sqrt 2}}
\begin{pmatrix} i\\1 \end{pmatrix}
&
{\displaystyle \frac 1 {\sqrt 2}}
\begin{pmatrix} i\\-1 \end{pmatrix}
\end{pmatrix}
\end{equation}
It can be checked that every row and column yields an orthonormal basis of $\C^2$, as required.

The next construction, which we believe to be new, produces a quantum Latin square from two unitary error bases.
\begin{cor}[$\UEB_n+\UEB_n\leadsto\QLS_{n^2}$] The construction of \autoref{fig:binarydetailed1}(b) produces an $n^2$\-dimensional quantum Latin square
\begin{equation}
Q_{a,b,cd} = \frac{1}{\sqrt{n}}\,\sum_{k \in [n]}\,U_{a,c,k}\, V_{b,k,d}
\end{equation}
from the following data, with $a,b\in [n^2]$ and $c,d\in [n]$:
\begin{itemize}
\item $U_{a,c,k}$ and $V_{b,k,d}\in \UEB_n$, $n$\-dimensional unitary error bases.
\end{itemize}
\end{cor}
\noindent
As with \autoref{thm:H+H=Q}, the factor $\smash{\frac{1}{\sqrt{n}}}$ arises since the biunitary $V$ is of rotated UEB type.

Note that we concatenate indices corresponding to tensor products of Hilbert spaces or products of indexing sets; for example, for a QLS on a Hilbert space $V\otimes W$, the coefficient of the basis vector $\ket{i,j}=\ket{i} \otimes \ket{j}$ in the $(a,b)$th position of the quantum Latin square will be written as $Q_{a,b,ij}$. Similarly, if the indexing set of a UEB is the product of two sets $[n]\times [m]$ we denote its $(a,b)$th element by $U_{ab}$ with coefficients $U_{ab,i,j}$.

\paragraph{Hadamard matrices.} The following construction produces a single Hadamard matrix from two controlled families.
\begin{cor}[Hosoya-Suzuki \cite{Hosoya:2003}, $\HAD_n^m + \HAD_m^n \leadsto \HAD_{nm}$] The construction of \autoref{fig:binarydetailed1}(c) produces an $nm$\-dimensional Hadamard matrix 
\begin{equation} \label{eq:Hosoya}H_{ab,cd} = J_{a,c}^b \, K_{b,d}^c
\end{equation}
from the following data, with $a,c\in[n]$ and $b,d\in[m]$:
\begin{itemize}
\item $J^b_{a,c}\in \HAD^m_n$, an $m$\-controlled family of $n$\-dimensional Hadamard matrices;
\item $K_{b,d}^c\in \HAD^n_m$, an $n$\-controlled family of $m$\-dimensional Hadamard matrices.
\end{itemize}
\end{cor}

\noindent
This construction was introduced in 2003 by Hosoya and Suzuki \cite{Hosoya:2003} under the name \textit{generalized tensor product}. Originally, they defined their tensor product as a block matrix $\left(J^1,\ldots,J^m\right) \otimes \left(K^1,\ldots, K^n\right)$ with $(i,j)$th block given by 
\begin{equation}\mathrm{diag}\left( J^1_{i,j}, \ldots, J^m_{i,j}\right) K^j.\end{equation}
This coincides with \eqref{eq:Hosoya}.

A better known special case of this construction, due to \Dita \cite{Dita:2004}, is a central tool in the study and classification of Hadamard matrices; we give it explicitly in \autoref{fig:binarydetailed1}(d).
\Dita's construction uses an $n$\-dimensional Hadamard matrix $J$ and an $n$\-controlled family of $m$\-dimensional Hadamard matrices $K^1,\ldots, K^n$ to obtain the Hadamard matrix $J\otimes (K^1,\ldots,K^n)$. The difference is that $J$ is a single Hadamard matrix in \Dita's construction, rather than a controlled family of Hadamard matrices.

\paragraph{Unitary error bases.} We now turn our attention to unitary error bases. By a manual combinatorial check,\footnote{Such a check is possible since there are only finitely many ways these structures can be composed.} it can be verified that the constructions in \autoref{fig:binarydetailed2} are the only possible binary constructions of UEBs using only Hadamard matrices, UEBs or QLSs and controlled families thereof.

\begin{figure}[h]
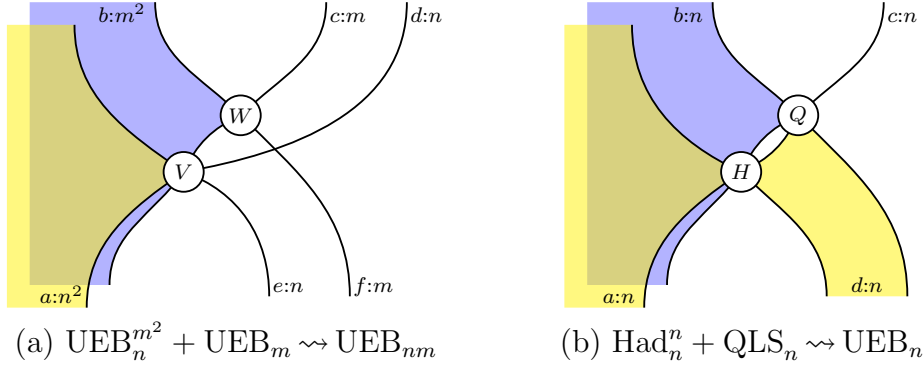

\begin{calign}
\nonumber
\begin{tz}[string,scale=1.5]
\path[blueregion] (0.75,2.5) to [out=-90, in=135] (1.5,1.5) to [out=-155, in=65] (1,1)to [out=-125, in=90] (0.35,0) to (0.35-\side,0) to (0.35-\side,2.5);
\draw[bnd] (0.75,2.5) to [out=-90, in=135] (1.5,1.5) to [out=-155, in=65] (1,1)to [out=-125, in=90] (0.35,0) ;
\path[yellowregion] (0.15,-\ydelta) to [out=90, in=-145] (1,1) to [out=145, in=-90] (0.75-\hrt, 2.5-\ydelta) to (0.15-\side,2.5-\ydelta) to (0.15-\side,-\ydelta);
\draw[bnd] (0.15,-\ydelta) to [out=90, in=-145] (1,1) to [out=145, in=-90] (0.75-\hrt, 2.5-\ydelta);
\draw (1.75,-0.5*\ydelta) to [out=90, in=-20] (1,1) to [out=10, in=-90] (2.25 +\hrt,2.5);
\draw  (1.75+\hrt,-0.5*\ydelta) to [out=90, in=-45] (1.5,1.5) to [out=45, in=-90] (2.25,2.5);
\node[blob] at (1,1) {$V$};
\node[blob] at (1.5,1.5) {$W$};
\node[dimension, below left] at (0.7, 2.36) {$\smash{b{:}m^2}$};
\node[dimension, below right] at (2.25, 2.36) {$\smash{c{:}m}$};
\node[dimension, below right] at (2.25 +\hrt,2.36) {$\smash{d{:}n}$};
\node[dimension,above left] at (0.15,-0.19 ) {$\smash{a{:}n^2}$};
\node[dimension, above right] at (1.75,-0.09) {$\smash{e{:}n} $};
\node[dimension, above right] at (1.75+\hrt,-0.09) {$\smash{f{:}m}$};
\end{tz}
&
\begin{tz}[string,scale=1.5]
\path[blueregion] (0.75,2.5) to [out=-90, in=135] (1.5,1.5) to [out=-155, in=65] (1,1)to [out=-125, in=90] (0.35,0) to (0.35-\side,0) to (0.35-\side,2.5);
\draw[bnd]  (0.75,2.5) to [out=-90, in=135] (1.5,1.5) to [out=-155, in=65] (1,1)to [out=-125, in=90] (0.35,0);
\path[yellowregion] (0.15,-\ydelta) to [out=90, in=-145] (1,1) to [out=155, in=-90] (0.75-\hrt, 2.5-\ydelta) to (0.15-\side,2.5-\ydelta) to (0.15-\side,-\ydelta);
\draw[bnd](0.15,-\ydelta) to [out=90, in=-145] (1,1) to [out=155, in=-90] (0.75-\hrt, 2.5-\ydelta);
\path[ yellowregion,bnd] (1.75,-0.5*\ydelta) to [out=90, in=-45] (1,1) to [out=25, in=-115] (1.5,1.5) to [out=-45, in=90]  (1.75+\hrt,-0.5*\ydelta);
\draw (1.5,1.5) to [out=45, in=-90] (2.25,2.5);
\node[blob] at (1,1) {$H$};
\node[blob] at (1.5,1.5) {$Q$};
\node[dimension, below left] at (0.7,  2.36) {$\smash{b{:}n}$};
\node[dimension, below right] at (2.25, 2.36) {$\smash{c{:}n}$};
\node[dimension,above]at (1.75+0.5*\hrt,-0.09) {$\smash{d{:}n}$};
\node[dimension,above left] at (0.1,-0.19) {$\smash{a{:}n}$};
\end{tz}
\\[-2pt]\nonumber
\text{(a) }\UEB_n^{m^2} + \UEB_m {\,\leadsto\,} \UEB_{nm} 
&
\text{(b) }\HAD_n^n + \QLS_n {\,\leadsto\,} \UEB_n  
\end{calign}
\vspace{-20pt}
\caption{Binary biunitary constructions of unitary error bases.\label{fig:binarydetailed2}}
\end{figure}

The following construction can be seen as the UEB analog of \Dita's construction given in~\autoref{fig:binarydetailed1}(d).

\begin{cor}[$\UEB_n^{m^2} + \UEB_m \leadsto \UEB_{nm}$] 
The construction of \autoref{fig:binarydetailed2}(a) produces an $nm$\-dimensional unitary error basis 
\begin{equation}
U_{ab,cd,ef} = V_{a,d,e}^b\, W^{}_{b,c,f}
\end{equation}
from the following data, with $a\in [n^2]$, $b\in [m^2]$, $c,f\in [m]$ and $d,e\in [n]$:
\begin{itemize}
\item $V_{a,d,e}^b\in \UEB^{m^2}_n$, an $m^2$\-controlled family of $n$\-dimensional unitary error bases;
\item $W_{b,c,f}\in \UEB_m$, an $m$\-dimensional unitary error basis.
\end{itemize}
\end{cor}

\noindent
In \autoref{fig:binarydetailed2}(a), we have used biunitarity of the interchanger as established in \autoref{prop:interchanger}.

It is also possible to compose biunitaries of different types to obtain unitary error bases, as shown by the following biunitary characterization of an existing construction, the \textit{quantum shift-and-multiply} method, which simultaneously generalizes the shift-and-multiply method~\cite{Werner:2001} and the Hadamard method~\cite[Definition 33]{Musto:2015}.
\begin{cor}[Musto \& V.~\cite{Musto:2015}, $\HAD_n^n+\QLS_n \leadsto \UEB_n$]\label{cor:QSM} The construction of \autoref{fig:binarydetailed2}(b) produces an $n$\-dimensional unitary error basis
\begin{equation} U_{ab,c,d} = H^b_{a,d}\, Q_{b,d,c}
\end{equation}
 from the following data, with $a,b,c,d\in [n]$:
\begin{itemize}
\item $H^b_{a,d}\in \HAD^n_n$, an $n$\-controlled family of $n$\-dimensional Hadamard matrices;
\item $Q_{b,d,c}\in \QLS_n$, an $n$\-dimensional quantum Latin square.
\end{itemize}
\end{cor}
\subsection{Ternary constructions}
\label{sec:ternaryquantum}

We can easily obtain higher arity constructions by iterating some of the binary constructions of \autoref{fig:binarydetailed1} and \autoref{fig:binarydetailed2}. For example, combining the constructions of \autoref{fig:binarydetailed1}(a) and \autoref{fig:binarydetailed2}(b) yields the following unitary error basis construction:
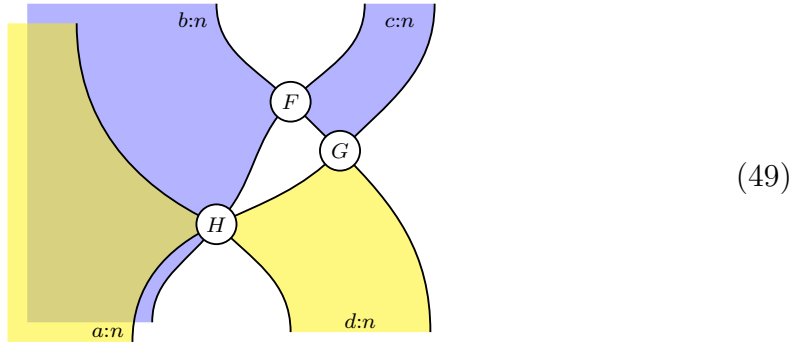
\begin{equation}\label{eq:HadamardMethod}
\begin{tz}[string,scale=1.3]
\path[blueregion,bnd] (2.5,3.25) to [out=-90, in=45] (1.75,2.25) to (2.25,1.75)  to [out=45, in=-90] (2.5 +\hrt,3.25);
\path[blueregion] (1.,3.25) to [out=-90, in=135] (1.75,2.25) to [out=-135, in=45] (1,1) to [out=-125, in=90] (0.35,0) to  (1.2-2*\hrt-\side,0) to (1.2-2*\hrt-\side,3.25);
\draw[bnd]  (1.,3.25) to [out=-90, in=135] (1.75,2.25) to [out=-135, in=45] (1,1) to [out=-125, in=90] (0.35,0);
\path[yellowregion,bnd] (1.75,-0.5*\ydelta) to [out=90, in=-45] (1,1) to [out=25, in=-135] (2.25,1.75)  to [out=-45, in=90] (1.75+2*\hrt  ,-0.5*\ydelta);
\path[ yellowregion] (0.15,-\ydelta) to [out=90, in=-155] (1,1) to [out=155, in=-90] (1-2*\hrt,3.25-\ydelta) to (1-2*\hrt-\side, 3.25-\ydelta) to (1-2*\hrt-\side, -\ydelta);
\draw[bnd]  (0.15,-\ydelta) to [out=90, in=-155] (1,1) to [out=155, in=-90] (1-2*\hrt,3.25-\ydelta);
\node[blob] at (1,1) {$H$};
\node[blob] at (1.75,2.25) {$F$};
\node[blob] at (2.25,1.75) {$G$};
\node[dimension, below left] at (0.95, 3.05) {$\smash{b{:}n}$};
\node[dimension, below] at (2.5+0.5*\hrt, 3.05) {$\smash{c{:}n}$};
\node[dimension, above left] at (0.1,-0.19) {$\smash{a{:}n}$};
\node[dimension,above] at (1.75+\hrt, -0.09) {$\smash{d{:}n}$};
\end{tz}
\end{equation}
In index notation this corresponds to the expression
\begin{equation}
U_{ab,c,d} = \frac{1}{\sqrt{n}} \, H^b_{a,d}\, F_{b,c}\, G_{c,d}
\end{equation}
built from the following data, with $a,b,c,d\in [n]$:
\begin{itemize}
\item $H^b_{a,d}\in \HAD^n_n$, an $n$\-controlled family of $n$\-dimensional Hadamard matrices;
\item $F_{b,c}$ and $G_{c,d} \in \HAD_n$, $n$\-dimensional Hadamard matrices.
\end{itemize}
This generalizes the Hadamard method~\cite[Definition 33]{Musto:2015}. By definition, this construction factors through the quantum shift-and-multiply method of \autoref{fig:binarydetailed2}(b).

More interestingly, there are ternary constructions that do not arise by iterating binary constructions involving our basic quantum structures. In this subsection, we list all ternary biunitary constructions of unitary error bases from Hadamard matrices, unitary error bases, quantum Latin squares and controlled families thereof, which do not factor through the constructions of \autoref{fig:binarydetailed2}. We summarize them in \autoref{fig:ternary}. By performing an exhaustive manual check, up to equivalence as defined in \autoref{sec:equivalence}, we assert that this list is complete, although we do not prove completeness in a formal way. To our knowledge, all constructions in this section are new. As before, all these results are corollaries of \autoref{thm:composition}, and the results of \autoref{sec:characterizing} as summarized in \autoref{fig:biunitarytypes}. To improve readability, we indicate the form of the compositions by corresponding tiling diagrams.

\def\vp{\vphantom{p|}}
\begin{figure}[t!]
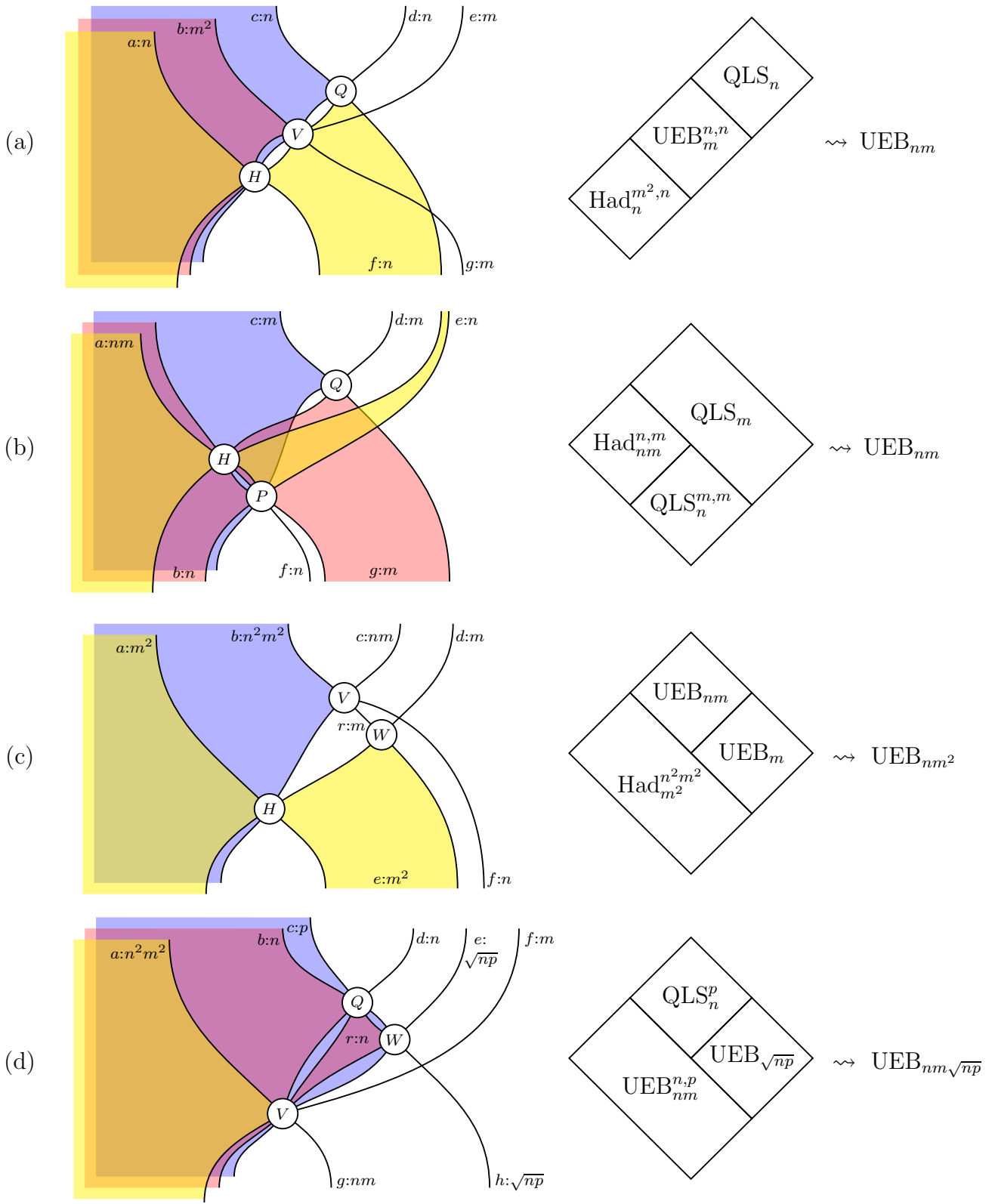

\[\hspace{-5cm}
\begin{array}{ccl}
\text{(a)}
&
\def \sidew {1.3} 
\def \xd{0.15}
\def \yd{0.15}
\begin{tz}[string,scale=1.5]
\clip (0.05-\sidew, -2*\yd) rectangle (2.93+2*\hrt, 3.2);

\path[blueregion] (1.25, 3) to [out=down, in=135] (2,2) to [out=-155, in=65] (1.5,1.5) to [out=-155, in=65] (1,1) to [out=-125, in=up] (0.25+\xd,0) to (0.25+\xd-\sidew,0) to (0.25+\xd-\sidew, 3);
\draw[bnd]  (1.25, 3) to [out=down, in=135] (2,2) to [out=-155, in=65] (1.5,1.5) to [out=-155, in=65] (1,1) to [out=-125, in=up] (0.25+\xd,0);

\path[redregion] (1.25 - \hrt, 3-\yd) to [out=down, in=135] (1.5,1.5) to [out=-190, in=100]  (1,1) to [out=-135, in=up] (0.25, -\yd) to (0.25-\sidew, -\yd) to (0.25-\sidew, 3-\yd);
\draw[bnd]  (1.25 - \hrt, 3-\yd) to [out=down, in=135] (1.5,1.5) to [out=-190, in=100]  (1,1) to [out=-135, in=up] (0.25, -\yd);

\path [yellowregion] (1.25 -2*\hrt, 3-2*\yd) to [out=down, in=135] (1,1) to [out=-145, in=up] (0.25-\xd, -2*\yd) to (0.25-\xd-\sidew, -2*\yd) to (0.25-\xd-\sidew, 3-2*\yd);
\draw[bnd]  (1.25 -2*\hrt, 3-2*\yd) to [out=down, in=135] (1,1) to [out=-145, in=up] (0.25-\xd, -2*\yd);

\path[yellowregion,bnd] (1.75 + 2*\hrt, -\yd) to [out=90, in=-45] (2,2) to [out=-115, in=25] (1.5,1.5) to [out=-115, in=25] (1,1) to [out=-25, in=90] (1.75, -\yd);

\draw (2,2) to [out=45, in=-90] (2.75,3);
\draw (2 + 2*\hrt,-\yd) to [out=90, in=-45] (1.5,1.5) to [out=15, in=-90] (2+2*\hrt,3); 

\node[blob] at (1,1) {$H$};
\node[blob] at (1.5,1.5) {$V$};
\node[blob] at (2,2) {$Q$};

\node [above,dimension] at (1.75 +\hrt, -\yd+0.05) {$\smash{f{:}n}$};
\node [above right,dimension] at (2+2*\hrt,-\yd+0.05) {$\smash{g{:}m}$};
\node [below left,dimension] at (1.25-2*\hrt, 3-2*\yd-0.15) {$\smash{a{:}n}$};
\node [below left,dimension] at (1.25-\hrt, 3-\yd-0.15) {$\smash{b{:}m^2}$};
\node[below left,dimension] at (1.25,3-0.15) {$\smash{c{:}n}$};
\node [below  right,dimension] at (2.75, 2.85) {$\smash{d{:}n}$};
\node [below right,dimension] at (2.75+\hrt, 2.85) {$\smash{e{:}m}$};

\end{tz}
&
\begin{tz}[string,rotate = 45,scale = 1.5]
\draw (0,0) rectangle (1,1);
\draw (1,0) rectangle (2,1);
\draw (2,0) rectangle (3,1);
\node at (0.5,0.5) {$\HAD_{n}^{m^2,n}$};
\node at (1.5,0.5){$\UEB_m^{n,n}$};
\node at (2.5,0.5) {$\QLS_n$};
\end{tz} \rightsquigarrow\,\UEB_{nm}
\\
\text{(b)}
&
\def \sidew {1.65}   
\def \xd{0.15}
\def \yd{0.15}
\begin{tz}[string,scale=1.3]
\clip (-0.05-\sidew, -2*\yd) rectangle (3.37+2*\hrt, 3.7);

\path[blueregion] (1.25,3.5) to [out=down, in=135] (2,2.5) to [out=-175, in=55] (1,1) to [out=-125, in=up] (0.25+\xd, 0) to (0.25+\xd-\sidew,0) to (0.25+\xd-\sidew, 3.5);
\draw[bnd] (1.25,3.5) to [out=down, in=135] (2,2.5) to [out=-175, in=55] (1,1) to [out=-125, in=up] (0.25+\xd, 0);

\path [redregion] (1.35 -2.5*\hrt,3.5-\yd) to [out=down, in=125] (0.5,1.5) to [out=-65, in=165] (1,1) to [out=-145, in=up] (0.25, -\yd) to (0.25-\sidew, -\yd) to (0.25-\sidew, 3.5-\yd);
\draw[bnd] (1.35 -2.5*\hrt,3.5-\yd) to [out=down, in=125] (0.5,1.5) to [out=-65, in=165] (1,1) to [out=-145, in=up] (0.25, -\yd);

\path[yellowregion] (1.15-2.5*\hrt, 3.5-2*\yd) to [out=down, in=145] (0.5,1.5) to [out=-145, in=up] (0.25-\hrt,-2*\yd) to (0.25-\xd-\sidew,-2*\yd) to (0.25-\xd-\sidew, 3.5-2*\yd);
\draw[bnd]  (1.15-2.5*\hrt, 3.5-2*\yd) to [out=down, in=145] (0.5,1.5) to [out=-145, in=up] (0.25-\hrt,-2*\yd);

\path[redregion,bnd] (1.85, -\yd) to [out=90, in=-35] (1,1) to [out=135, in=-45] (0.5,1.5) to [out=55, in=-125] (2,2.5) to [out=-45, in=90] (1.75 + 2.5*\hrt,-\yd);

\path[yellowregion, bnd] (2.7+\hrt, 3.5) to [out=down, in=30] (0.5,1.5) to [out=-20, in=110] (1,1) to [out=35, in=-90] (2.8+\hrt,3.5);

\draw (1,1) to [out=-55, in=90] (1.65, -\yd);
\draw (2,2.5) to [out=45, in=-90] (2.75,3.5);

\node[blob] at (1,1) {$P$};
\node[blob] at (0.5,1.5) {$H$};
\node[blob] at (2,2.5) {$Q$};

\node[below left, dimension] at (1.1-2.5*\hrt,3.35-2*\yd) {$\smash{a{:}nm}$};
\node[above left, dimension] at (0.15,-\yd+0.01) {$\smash{b{:}n}$};
\node[above left, dimension] at (1.6, -\yd+0.05) {$\smash{f{:}n}$};
\node[above, dimension] at (1.75+1.25*\hrt, -\yd+0.05) {$\smash{g{:}m}$};
\node[below left,dimension] at (1.25, 3.35) {$\smash{c{:}m}$};
\node[below right, dimension] at (2.75, 3.35) {$\smash{d{:}m}$};
\node[below right,dimension] at (2.85+\hrt, 3.35) {$\smash{e{:}n}$};

\end{tz}
 &
\begin{tz}[string,rotate=45,scale=1.5]

\draw(0,0) rectangle (1,1);
\draw (0,1) rectangle (1,2);
\draw (1,0) rectangle (2,2);
\node at (0.5,0.5) {$\QLS_n^{m,m}$};
\node at (0.5,1.5) {$\HAD_{nm}^{n,m}$};
\node at (1.5, 1){$\QLS_m$};
\end{tz} \, \rightsquigarrow\,\UEB_{nm}
\\
\text{(c)}
&
\def \sidew {1.7}   
\def \xd{0.15}
\def \yd{0.15}
\begin{tz}[string,scale=1.3]
\clip (0.15-\sidew,-2*\yd) rectangle (2.3+3*\hrt, 3.5+2*\yd);

\path[blueregion] (1.25,3.5) to [out=down, in= 135] (2,2.5) to [out=-135, in=55] (1,1) to [out= -125, in=up] (0.35,0) to (0.35-\sidew, 0) to (0.35-\sidew, 3.5);
\draw[bnd] (1.25,3.5) to [out=down, in= 135] (2,2.5) to [out=-135, in=55] (1,1) to [out= -125, in=up] (0.35,0) ;

\path[yellowregion] (1.25-2.5*\hrt, 3.5-\yd) to [out= down , in=135] (1,1) to [out= -145, in=up] (0.15,-\yd) to (0.35-\xd-\sidew, -\yd) to (0.35-\xd-\sidew, 3.5-\yd);
\draw[bnd] (1.25-2.5*\hrt, 3.5-\yd) to [out= down , in=135] (1,1) to [out= -145, in=up] (0.15,-\yd);

\path[yellowregion,bnd] (1.75,-0.5*\yd) to [out=up, in=-45] (1,1) to [out= 35, in=-135] (2.5,2) to [out= -45, in= up] (1.75+2.5*\hrt, -0.5*\yd);

\draw (2.75,3.5) to [out= down, in=45] (2,2.5) to (2.5,2) to [out= 45, in=down] (2.75+\hrt, 3.5);
\draw (2,2.5) to [out= -10, in=up] (1.75+3*\hrt, -0.5*\yd);

\node[blob] at (1,1) {$H$};
\node[blob] at (2.5,2) {$W$};
\node[blob] at (2,2.5) {$V$};

\node[above,dimension] at (1.75+1.25*\hrt,-0.05) {$\smash{e{:}m^2}$};
\node[above right,dimension] at (1.75+3*\hrt,-0.05) {$\smash{f{:}n}$};
\node[below left, dimension] at (1.25-2.5*\hrt,3.15) {$\smash{a{:}m^2}$};
\node[below left,dimension] at (1.25,3.3) {$\smash{b{:}n^2m^2}$};
\node[below left, dimension] at (2.7,3.3) {$\smash{c{:}nm}$};
\node[below right, dimension] at (2.75+\hrt,3.3) {$\smash{d{:}m}$};

\node[below left, dimension] at (2.325,2.1) {$\smash{r{:}m}$};

\end{tz}
& 
\begin{tz}[string, rotate=45,scale=1.5]
\draw (0,0) rectangle  (1,2);
\draw (1,0) rectangle  (2,1);
\draw (1,1) rectangle (2,2);
\node at (0.5,1) {$\HAD^{n^2m^2}_{m^2}$};
\node at (1.5, 0.5) {$\UEB_{m}$};
\node at (1.5,1.5){$\UEB_{nm}$};
\end{tz}~\rightsquigarrow ~\UEB_{nm^2}
\\
\text{(d)}
&
\tikzset{helper/.style={}}
\tikzset{helper/.style={opacity=0}}
\def \sidew {1.8}   
\def\angledelta{15} 
\begin{tz}[string,scale=1.3]
\path [use as bounding box] (-0.05-\sidew, -0.2) rectangle (2.95+2*\hrt, 3.7);
\node[blob] (V) at (1,1) {$V$};
\node[blob] (W) at (2.5,2) {$W$};
\node[blob] (Q) at (2,2.5) {$Q$};

\path [blueregion] (1.37, 3.65) node [helper] (1) {1} to [out=down, in=125] (Q.center) to [out=-45+\angledelta, in=135-\angledelta] (W.center) to [out=-135+\angledelta, in=45-\angledelta] (V.center) to [out=-135+\angledelta, in=up] (0.35,0.15) node [helper] (2) {2} to (0.3-\sidew,0 |- 2.center) to (0.3-\sidew,0 |- 1.center);
\draw [bnd] (1.center) to [out=down, in=125] (Q.center) to [out=-45+\angledelta, in=135-\angledelta] (W.center) to [out=-135+\angledelta, in=45-\angledelta] (V.center) to [out=-135+\angledelta, in=90] (2.center) ;

\path [redregion] (1,3.5) node [helper] (3) {3} to [out=down, in=145] (Q.center) to [out=-135, in=45+2*\angledelta] (V.center) to [out=-135, in=90] (0.15,-0) node [helper] (4) {4} to (0.15-\sidew,0 |- 4.center) to (0.15-\sidew,0 |- 3.center);
\draw [bnd] (3.center) to [out=down, in=145] (Q.center) to [out=-135, in=45+2*\angledelta] (V.center) to [out=-135, in=90] (4.center);

\path[yellowregion] (1.25-2.5*\hrt, 3.35) node [helper] (6) {6} to [out=down, in=135] (V.center) to [out=-135-\angledelta, in=up] (-0.05, -0.15) node [helper] (5) {5} to (-0.0-\sidew,0 |- 5.center) to (-0.0-\sidew,0 |- 6.center);
\draw [bnd] (6.center) to [out= down, in=135] (V.center) to [out=-135-\angledelta, in=up] (-0.05, -0.2) ;

\path [redregion, bnd] (V.center) to [out=45+\angledelta, in=-115] (Q.center) to [out=-45-\angledelta, in=135+\angledelta] (W.center) to [out=-135-\angledelta, in=45] (1,1) ;

\draw [string] (Q.center) to [out=45, in=down] (2.75,3.5 |- 3.center) node [helper] (7) {7};
\draw [string] (2.75+\hrt,3.5 |- 3.center) node [helper] (8) {8} to [out=down, in=45] (2.5,2) to [out=-45, in=up] (1.65+3*\hrt,0 |- 4.center) node [helper] (10) {10};
\draw [string] (1.65,0) node [helper] (11) {11} to [out=up, in=-35] (V.center) to [out=45-2*\angledelta, in=-90] (2.75+2*\hrt,3.5) node [helper] (9) {9};

\node[below left, dimension] at ([yshift=-1pt]6.center) {${a{:}n^2m^2}$};
\node[above right, dimension] at ([xshift=1pt]11.center) {$\smash{g{:}nm}$};
\node[above right, dimension] at ([xshift=1pt]10.center) {$\smash{h{:}\sqrt{np}}$};

\node[below left, dimension] at ([yshift=-1pt]3.center) {${b{:}n}$};
\node[below left, dimension] at (1.center) {${c{:}p}\vp$};
\node[below right, dimension] at (7.center) {${d{:}n}\vp$};
\node[below right, dimension, align=center] at ([xshift=-2pt] 8.center) {$\begin{array}{@{}c@{}}e{:}\vp\\\sqrt{np}\end{array}$};
\node[below right,dimension] at ([xshift=1pt]9.center) {$f{:}m\vp$};

\node[dimension] at (2.0,1.95) {$\smash{r{:}n}$};

\end{tz}
&
\begin{tz}[string, rotate=45,scale=1.5]
\draw (0,0) rectangle  (1,2);
\draw (1,0) rectangle  (2,1);
\draw (1,1) rectangle (2,2);
\node at (0.5,1) {$\UEB^{n,p}_{nm}$};
\node at (1.5, 0.5) {$\UEB_{\sqrt{np}}$};
\node at (1.5,1.5){$\QLS^p_{n}$};
\end{tz}~\rightsquigarrow ~\UEB_{nm\sqrt{np}}
\end{array}
\hspace{-5cm}
\]
\vspace{-10pt}
\caption{\label{fig:ternary}An overview of all ternary unitary error basis constructions.}
\end{figure}

The constructions of \autoref{fig:ternary}(a) and \autoref{fig:ternary}(b) can be seen as slight alterations of constructions that factor through the constructions of \autoref{fig:binarydetailed2}, while the other constructions in \autoref{fig:ternary} do not seem to have binary analogues.

\begin{cor}[$\HAD_{n}^{m^2,n} + \UEB_{m}^{n,n} + \QLS_n \leadsto\UEB_{nm}$] The construction of \autoref{fig:ternary}(a) produces an $nm$\-dimensional UEB
\begin{calign}
U_{abc,de,fg} = H^{b,c} _{a,f} \, V^{c,f} _{b,e,g} \, Q _{c,f,d} ^{\vphantom{c,d}}
\end{calign}
from the following data, with $a,c,d,f \in [n]$, $b \in [m^2]$ and $e,g \in [m]$:
\begin{itemize}
\item $H^{b,c}_{a,f}\in \HAD^{m^2,n}_n$, an $(m^2,n)$\-controlled family of $n$\-dimensional Hadamard matrices;
\item $V^{c,f}_{b,e,g}\in \UEB^{n,n}_m$, an $(n,n)$\-controlled family of $m$\-dimensional unitary error bases; 
\item $Q _{c,f,d}\in \QLS_n$, an $n$\-dimensional quantum Latin square.
\end{itemize}
\end{cor}

\noindent
If the UEB $V$ were not controlled, the construction would be the tensor product of $V$ with the quantum shift-and-multiply UEB obtained as in \autoref{fig:binarydetailed2}(b).

The following construction is also related to one of the binary constructions.
\begin{cor}[$\HAD_{nm}^{n,m} + \QLS_n^{m,m} + \QLS_m \leadsto \UEB_{nm}$]
The construction of \autoref{fig:ternary}(b) produces an $nm$\-dimensional UEB
\begin{equation}
U_{abc,de,fg}=H_{a,eg}^{b,c}\, P_{e,b,f}^{c,g}\, Q_{c,g,d}^{\vphantom{b}}
\end{equation}
from the following data, with $a\in [nm]$ $b,e,f\in [n]$ and $c,d,g\in [m]$:
\begin{itemize}
\item $H^{b,c}_{a,eg} \in \HAD_{nm}^{n,m}$, an $(n,m)$\-controlled family of $nm$\-dimensional Hadamard matrices;
\item $P_{e,b,f}^{c,g} \in \QLS_n^{m,m}$, an $(m,m)$\-controlled family of $n$\-dimensional quantum Latin squares; 
\item $Q _{c,g,d} \in \QLS_m$, an $m$\-dimensional quantum Latin square.
\end{itemize}
\end{cor}

\noindent
In fact, taking the partial transpose of the resulting UEB (that is, bending the  $d$ wire down and the $g$ wire up) leads to the quantum shift-and-multiply UEB generated from the Hadamard matrices $H$ and a quantum Latin square obtained from the controlled tensor product (the QLS analogue of \Dita's construction) of $P$ and $Q$. This relationship is surprising since taking the partial transpose does not in general preserve biunitarity.     

The following is geometrically the simplest of our ternary constructions. 
It involves a closed wire, so the index expression includes a sum.
\begin{cor}[$ \HAD_{m^2}^{n^2m^2} +  \UEB_{nm}+\UEB_m \leadsto\UEB_{nm^2}$] The construction of \autoref{fig:ternary}(c) produces an $nm^2$\-dimensional UEB 
 \begin{equation}
 U_{ab,cd,ef} = \sum_{r\in [m]} \, H_{a,e}^b \, V_{b,c,rf}^{\vphantom{b}} \, W_{e,r,d}^{\vphantom{b}}
 \end{equation}
from the following data, with $a,e\in [m^2]$, $b\in [n^2m^2]$, $c\in [nm]$, $d\in [m]$, and $f\in [n]$:
\begin{itemize}
\item $H^{b}_{a,e}\in \HAD^{n^2m^2}_{m^2}$, an $n^2m^2$\-controlled family of $m^2$\-dimensional Hadamard matrices;
\item $V_{b,c,rf}\in \UEB_{nm}$, an $nm$\-dimensional unitary error basis; 
\item $W _{e,r,d}\in \UEB_m$, an $m$\-dimensional unitary error basis.
\end{itemize}

   \end{cor}


Our final ternary construction is the first to involve a sum over a closed region, which again gives rise to a summation.

\begin{cor}[$\UEB_{nm}^{n,p} +\QLS_n^p+ \UEB_{\sqrt{np}}   \leadsto \UEB_{nm\sqrt{np}}$]  For $n,m,p\in \mathbb{N}$ such that $\sqrt{np}\in \mathbb{N}$, the construction of \autoref{fig:ternary}(d) produces an $nm\sqrt{np}$\-dimensional UEB
\begin{equation}
U_{abc,def,gh} := \sum_{r\in [n]} V_{a,rf,g}^{b,c}\, Q_{b,r,d}^c  \, W_{rc,e,h}^{\vphantom{b}} 
\end{equation}
from the following data, with $a\in [n^2m^2]$, $b,d\in [n], c\in [p]$, $e,h\in [\sqrt{np}]$, $f\in [m]$, and $g\in [nm]$:
\begin{itemize}
\item $V_{a,rf,g}^{b,c}\in \UEB_{nm}^{n,p}$, an $(n,p)$\-controlled family of $nm$\-dimensional unitary error bases;
\item $Q_{b,r,d}^c\in \QLS^p_n$, an $p$\-controlled family of $n$\-dimensional quantum Latin squares;
\item $W_{rc,e,h}\in \UEB_{\sqrt{np}}$, an $\sqrt{np}$\-dimensional unitary error basis.

\end{itemize}
 \end{cor}

\noindent
A particularly simple case of this final construction is the following, which plays a role in our argument in \autoref{sec:infinity} that our methods give rise to infinitely many distinct constructions:
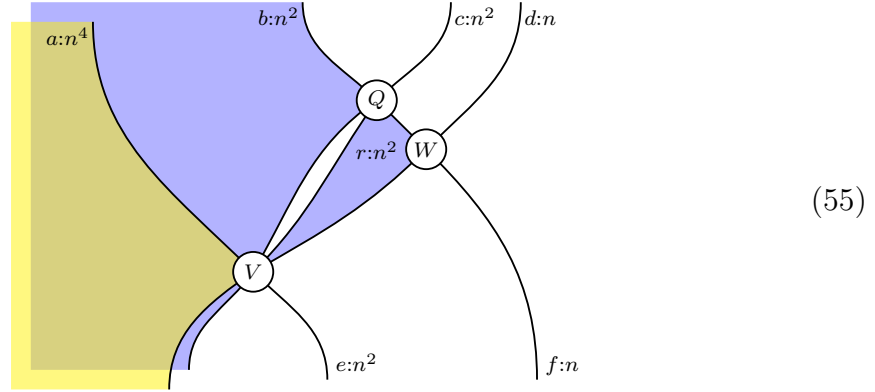
\def\sidew {1.6} 
\def\dist{3*\hrt} 
\begin{equation}
\label{eq:particularlysimple}
\begin{tz}[string,scale=1.3]
\path[blueregion] (0.25,2) to [out=down, in=135] (1,1) to [out=-145, in=60] (-0.25,-0.75) to [out=-125, in=up] (-0.9, -1.75)  to (-0.9-\sidew, -1.75) to (-0.9-\sidew, 2);
\draw[bnd] (0.25,2) to [out=down, in=135] (1,1) to [out=-145, in=60] (-0.25,-0.75) to [out=-125, in=up] (-0.9, -1.75) ;
\path[yellowregion] (0.25-\dist, 2-\ydelta) to [out=down, in=135] (-0.25,-0.75) to [out=-145, in=up] (-1.1, -1.75-\ydelta) to (-1.1-\sidew,-1.75-\ydelta) to (-1.1-\sidew,2-\ydelta);
\draw[bnd] (0.25-\dist, 2-\ydelta) to [out=down, in=135] (-0.25,-0.75) to [out=-145, in=up] (-1.1, -1.75-\ydelta);
\path[blueregion,bnd] (-0.25,-0.75) to [out=45, in=-125] (1,1) to (1.5,0.5) to [out=-135, in=30] (-0.25,-0.75);

\draw (-0.25,-0.75) to [out=-45, in=up] (0.5,-1.75-0.5*\ydelta);
\draw (1.5,0.5) to [out=-45, in=up] (0.5+3*\hrt, -1.75-0.5*\ydelta);
\draw (1,1) to [out=45, in=down] (1.75,2);
\draw (1.5,0.5) to [out=45, in=down] (1.75+\hrt, 2);

\node[blob] at (1,1) {$Q$};
\node[blob] at (1.5,0.5) {$W$};
\node[blob] at (-0.25,-0.75) {$V$};

\node[above right,dimension] at (0.55,-1.8) {$\smash{e{:}n^2}$};
\node[above right,dimension] at (0.55+3*\hrt,-1.8) {$\smash{f{:}n}$};
\node[below left,dimension] at (0.25-\dist,1.6) {$\smash{a{:}n^4}$};
\node[below left,dimension] at (0.25,1.8) {$\smash{b{:}n^2}$};
\node[below right, dimension] at (1.75,1.8) {$\smash{c{:}n^2}$};
\node[below right, dimension] at (1.75+\hrt ,1.8) {$\smash{d{:}n}$};
\node[left, dimension] at (1.25,0.4) {$\smash{r{:}n^2}$};
\end{tz}\end{equation}
This produces an $n^3$\-dimensional UEB 
\begin{equation} U_{ab,cd,ef} := \sum_{r\in [n^2]} V_{a,r,e}^b \, Q_{b,r,c}^{\vphantom{b}} \, W_{r,d,f}^{\vphantom{b}} 
\end{equation}
from the following data, with $a\in [n^4];b,c,e\in [n^2]$ and $d,f\in [n]$:
\begin{itemize}
\item $V_{a,r,e}^{b}\in\UEB^{n^2}_{n^2}$, an $n^2$\-controlled family of $n^2$\-dimensional unitary error bases;
\item $Q_{b,r,c}\in \QLS_{n^2}$, an $n^2$\-dimensional quantum Latin squares;
\item $W_{r,d,f}\in\UEB_n$, an $n$\-dimensional unitary error basis.

\end{itemize}

\subsection{Higher constructions}
\label{sec:higherquantum}

Interesting biunitary composites exist for higher arity, and are easy to discover by \textit{ad hoc} experimentation. We illustrate two examples which seem particularly elegant, and which, to our knowledge, are new constructions. Once again, these results are all corollaries of \autoref{thm:maintheorem}, and the results of \autoref{sec:characterizing} as summarized in \autoref{fig:biunitarytypes}.

\begin{figure}
\makebox[\textwidth][c]{\begin{minipage}{1.25\textwidth}
\begin{calign}
\nonumber
\def\sidew{0.8} 
\def \xd{0.1}
\def \yd{0.1}
\begin{tz}[string,scale=4.13/3.1*1.3]
\path[blueregion] (0.25,3) to [out=-90, in=135] (1,2) to [out=-155, in=65] (0.5,1.5) to [out=-135, in=90] (0.25+0.5*\xd-\hrt, 0) to (0.25-\hrt-\sidew,0) to (0.25-\hrt-\sidew,3);
\draw[bnd]  (0.25,3) to [out=-90, in=135] (1,2) to [out=-155, in=65] (0.5,1.5) to [out=-135, in=90] (0.25+0.5*\xd-\hrt, 0) ;
\path[redregion] (0.25,-\yd) to [out=90, in=-135] (1,1) to [out=155, in=-65] (0.5,1.5) to [out=135, in=-90] (0.25+0.5*\xd-\hrt,3-\yd) to (0.25-\xd-\hrt-\sidew,3-\yd) to (0.25-\xd-\hrt-\sidew,-\yd);
\draw[bnd] (0.25,-\yd) to [out=90, in=-135] (1,1) to [out=155, in=-65] (0.5,1.5) to [out=135, in=-90] (0.25+0.5*\xd-\hrt,3-\yd) ;
\path[yellowregion] (0.25-0.5*\xd-\hrt, 3-2*\yd) to [out=-90, in=155] (0.5,1.5) to [out=-155, in=90] (0.25-0.5*\xd-\hrt, -2*\yd) to (0.25-2*\xd-\hrt-\sidew,-2*\yd) to (0.25-2*\xd-\hrt-\sidew,3-2*\yd);
\draw[bnd](0.25-0.5*\xd-\hrt, 3-2*\yd) to [out=-90, in=155] (0.5,1.5) to [out=-155, in=90] (0.25-0.5*\xd-\hrt, -2*\yd);
\path[yellowregion,bnd] (0.5,1.5) to [out=25, in=-115] (1,2) to [out=-65, in=155] (1.5,1.5) to [out=-155, in=65] (1,1) to [out=115, in=-25] (0.5,1.5);
\draw (1,1) to [out=-45, in=90] (1.75,-\yd);
\draw (1.75 +0.5*2^0.5,-\yd) to [out=90, in=-45] (1.5,1.5) to [out=45, in=-90] (1.75+0.5 *2^0.5,3);
\draw (1,2) to [out=45, in=-90] (1.75,3);
\node[blob] at (1,1) {$Q$};
\node[blob]at (1.5,1.5) {$V$};
\node[blob] at (0.5,1.5) {$H$};
\node[blob] at (1,2) {$P$};
\node[below left,dimension] at (0.25,2.85) {$\smash{c{:}n^2}$};
\node[below right, dimension] at (1.75,2.85) {$\smash{d{:}n^2}$};
\node[below right, dimension] at (1.75+\hrt,2.85){$\smash{e{:}n}$};
\node[below left, dimension] at (0.15-\hrt,2.85-2*\yd){$\smash{a{:}n^2}$};
\node[above left, dimension] at (0.25,-0.05) {$\smash{b{:}n^2}$};
\node[above right, dimension] at (1.75,-0.05) {$\smash{f{:}n^2}$};
\node[above right, dimension] at (1.75+\hrt, -0.05) {$\smash{g{:}n}$};
\node[dimension] at (1.02,1.45) {$\smash{r{:}n^2}$};
\end{tz}
&
\def\sidew{0.8}  
\def \xd{0.13}
\def \yd{0.13}
\begin{tz}[string,scale = 1.3]
\path[yellowregion, bnd] (1.75 +\hrt, 4) to [out=-90, in=45] (1.5,2.5) to (2,2) to [out=45, in=-90] (1.75 + 2*\hrt,4);
\path[yellowregion,bnd] (1.75 +\hrt, -\yd) to [out=90, in=-45] (1.5,1.5) to (2,2) to [out=-45, in=90] (1.75 + 2*\hrt,-\yd);
\path[yellowregion,bnd] (1,2) to [out= -65, in=155] (1.5,1.5) to [out=-155, in=65] (1,1) to [out=115, in=-25] (0.5,1.5) to [out=25, in=-115] (1,2);
\path[yellowregion,bnd] (1,3) to [out=-65, in=155] (1.5,2.5) to [out=-155, in=65] (1,2) to [out=115, in=-25] (0.5,2.5) to [out=25, in=-115] (1,3);
\path[blueregion] (0.25,4) to [out=down, in=135] (1,3) to [out=-155, in=65] (0.5,2.5) to [out=-125, in=90] (0.35-2*\hrt,0) to (0.25-2*\hrt-\sidew,0) to (0.25-2*\hrt-\sidew,4);
\draw[bnd]   (0.25,4) to [out=down, in=135] (1,3) to [out=-155, in=65] (0.5,2.5) to [out=-125, in=90] (0.35-2*\hrt,0);
\path[redregion] (0.25,-\yd) to [out=90, in=-135] (1,1) to [out=155, in=-65] (0.5,1.5) to [out=125, in=-90] (0.35 -2*\hrt,4-\yd) to (0.25-\xd-2*\hrt-\sidew,4-\yd) to (0.25-\xd-2*\hrt-\sidew,-\yd);
\draw[bnd](0.25,-\yd) to [out=90, in=-135] (1,1) to [out=155, in=-65] (0.5,1.5) to [out=125, in=-90] (0.35 -2*\hrt,4-\yd);
\path[greenregion] (0.25-\hrt,-2*\yd) to [out=up, in=-135] (0.5,1.5) to [out=145, in=down] (0.15-2*\hrt,4-2*\yd) to (0.25-2*\xd-2*\hrt-\sidew, 4-2*\yd) to (0.25-2*\xd-2*\hrt-\sidew, -2*\yd);
\draw[bnd] (0.25-\hrt,-2*\yd) to [out=up, in=-135] (0.5,1.5) to [out=145, in=down] (0.15-2*\hrt,4-2*\yd);
\path[yellowregion] (0.25-\hrt,4-3*\yd) to [out=-90, in=135] (0.5,2.5) to [out=-145, in=90] (0.15-2*\hrt,-3*\yd) to (0.25-3*\xd-2*\hrt-\sidew, -3*\yd) to (0.25-3*\xd-2*\hrt-\sidew, 4-3*\yd) to (0.25-\hrt,4-3*\yd);
\draw[bnd] (0.25-\hrt,4-3*\yd) to [out=-90, in=135] (0.5,2.5) to [out=-145, in=90] (0.15-2*\hrt,-3*\yd) ;
\draw (1.75,4) to [out=-90, in=45] (1,3);
\draw (1.75,-\yd) to [out=90, in=-45] (1,1);
\node[blob]at (1,1) {$P$};
\node[blob] at (1.5,1.5) {$C$};
\node[blob] at (0.5,1.5) {$K$};
\node[blob] at (1,2) {$D$};
\node[blob]at (1,3) {$Q$};
\node[blob] at (0.5,2.5) {$H$};
\node[blob]at (2,2) {$A$};
\node[blob] at (1.5,2.5) {$B$};
\node[below left, dimension] at (0.25,3.8) {$\smash{d{:}n}$};
\node[below right, dimension] at (1.75,3.8) {$\smash{e{:}n}$};
\node[below, dimension] at (1.75+1.5*\hrt, 3.8) {$\smash{f{:}n}$};
\node[above left, dimension] at (0.25-\hrt, -2*\yd) {$\smash{b{:}n}$};
\node[above left, dimension] at (0.25,-0.1) {$\smash{c{:}n}$};
\node[above right, dimension]at (1.75,-0.1) {$\smash{g{:}n}$};
\node[above, dimension] at (1.75+1.5*\hrt, -0.1) {$\smash{h{:}n}$};
\%middle index
\node[below left, dimension] at (0.25-\hrt, 3.85-3*\yd) {$\smash{a{:}n}$};
\node[dimension] at ( 1.02 ,1.43) {$\smash{r{:}n}$};
\node[dimension] at (1.02, 2.43) {$\smash{s{:}n}$};
\end{tz}
\\[-2pt]
\nonumber
\text{(a)}\, \HAD^{n^2,n^2}_{n^2} + 2{\times}\QLS_{n^2} + \UEB_n {\,\leadsto\,} \UEB_{n^3}
&
\text{(b)} \,4{\times}\HAD_n + 2{\times}\HAD_n^n +2{\times}\QLS_n {\,\leadsto\,}\UEB_{n^2}
\end{calign}
\end{minipage}}
\caption{Some higher-arity constructions of unitary error bases.\label{fig:higher}}
\hspace{2.5cm}
\end{figure}%
\begin{cor}[$\HAD_{n^2}^{n^2,n^2}+2{\times} \QLS_{n^2}+\UEB_n\leadsto\UEB_{n^3}$]
\label{cor:4aryconstruction}
The construction in \autoref{fig:higher}(a) produces an $n^3$\-dimensional UEB 
\begin{equation}
\label{eq:8aexplicit}
U_{abc,de,fg} = \sum_{r\in [n^2]} \, H_{a,r}^{b,c} \,  P_{c,r,d}^{\vphantom{b}} \, Q_{r,b,f}^{\vphantom{b}} \,V_{r,e,g}^{\vphantom{b}} 
\end{equation}
from the following data, with $a,b,c,d,f\in [n^2]$ and $e,g\in[n]$:
\begin{itemize}
\item $H_{a,r}^{b,c}\in \HAD_{n^2}^{n^2,n^2}$, an $(n^2,n^2)$\-controlled family of $n^2$\-dimensional Hadamard matrices;
\item $P_{c,r,d}, Q_{r,b,f}\in \QLS_{n^2}$, $n^2$\-dimensional quantum Latin squares; 
\item $V_{r,e,g}\in\UEB_n$, an $n$\-dimensional unitary error bases.
\end{itemize}
\end{cor}

\noindent
In \autoref{sec:newUEB}, we will use this construction  to produce a new unitary error basis that cannot be obtained by the most general previously known methods.

We now turn to the 8-ary construction of \autoref{fig:higher}(b).
\begin{cor}[$4{\times}\HAD_{n} + 2{\times}\HAD^n_{n}+2{\times} \QLS_n\leadsto\UEB_{n^2}$]
The construction in \autoref{fig:higher}(b) produces an $n^2$\-dimensional UEB 
\begin{equation}
 U_{abcd,ef,gh} = \frac{1}{n} \,\sum_{r,s\in [n]} A_{f,h} \, B_{s,f} \, C_{r,h} \, D_{s,r} \, H_{a,s}^d \, K_{b,r}^c \, Q_{d,s,e} \, P_{r,c,g} 
\end{equation}
from the following data, with $a,b,c,d,e,f,g,h\in [n]$:
\begin{itemize}
\item $A_{f,h},B_{s,f},C_{r,h},D_{s,r}\in \HAD_n$, $n$\-dimensional Hadamard matrices;
\item $H^d_{a,s}, K^c _{b,r}\in \HAD^n_n$, $n$\-controlled families of $n$\-dimensional Hadamard matrices;
\item $Q_{d,s,e},P_{r,c,g}\in \QLS_n$, $n$\-dimensional quantum Latin squares.
\end{itemize}
\end{cor}

\noindent
The factor $\frac{1}{n}$ arises from using two rotated Hadamard matrices $A$ and $D$; see the discussion after \autoref{thm:H+H=Q}.

\subsection{An infinity of constructions}
\label{sec:infinity}
\begin{figure}[t]
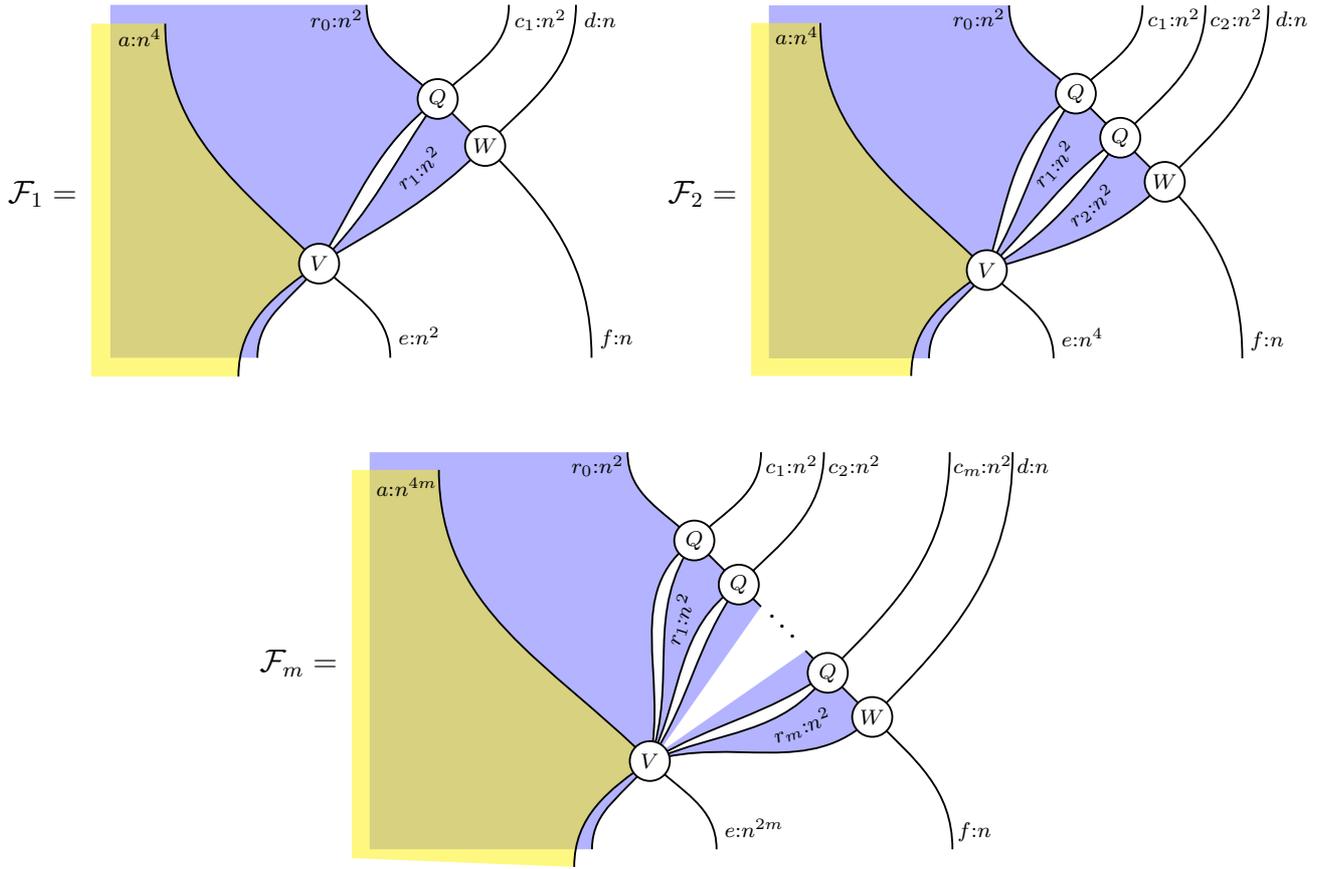

\def\angledelta{7}
\begin{calign}
\nonumber
\def\sidew {1.55} 
\def\dist{3*\hrt} 
\def \xd{0.2}
\def\yd{0.2}
\hspace{-30pt}
\mathcal{F}_1=
\begin{tz}[string,scale=4/3.75*1.17]
\path[blueregion] (0.25,2) to [out=down, in=135] (1,1) to [out=-135-\angledelta, in=45+2*\angledelta] (-0.25,-0.75) to [out=-125, in=up] (-0.9, -1.75)  to (-0.9-\sidew, -1.75) to (-0.9-\sidew, 2);
\draw[bnd]        (0.25,2) to [out=down, in=135] (1,1) to [out=-135-\angledelta, in=45+2*\angledelta] (-0.25,-0.75) to [out=-125, in=up] (-0.9, -1.75) ;
\path[yellowregion] (0.25-\dist, 2-\yd) to [out=down, in=135] (-0.25,-0.75) to [out=-145, in=up] (-0.9-\xd, -1.75-\yd) to (-0.9-\xd-\sidew,-1.75-\yd) to (-0.9-\xd-\sidew,2-\yd);
\draw[bnd](0.25-\dist, 2-\yd) to [out=down, in=135] (-0.25,-0.75) to [out=-145, in=up] (-0.9-\xd, -1.75-\yd);
\path[blueregion,bnd] (-0.25,-0.75) to [out=45, in=-125] (1,1) to (1.5,0.5) to [out=-135, in=45-2*\angledelta] (-0.25,-0.75);
\draw (-0.25,-0.75) to [out=-45, in=up] (0.5,-1.75);
\draw (1.5,0.5) to [out=-45, in=up] (0.5+3*\hrt, -1.75);
\draw (1,1) to [out=45, in=down] (1.75,2);
\draw (1.5,0.5) to [out=45, in=down] (1.75+\hrt, 2);
\node[blob] at (1,1) {$Q$};
\node[blob] at (1.5,0.5) {$W$};
\node[blob] at (-0.25,-0.75) {$V$};
\node[above right,dimension] at (0.55,-1.65) {$\smash{e{:}n^2}$};
\node[above right,dimension] at (0.55+3*\hrt,-1.65) {$\smash{f{:}n}$};
\node[below left,dimension] at (0.25-\dist,1.6) {$\smash{a{:}n^4}$};
\node[below left,dimension] at (0.25,1.8) {$\smash{r_0{:}n^2}$};
\node[below right, dimension] at (1.77,1.8) {$\smash{c_1{:}n^2}$};
\node[below right, dimension] at (1.8+\hrt ,1.8) {$\smash{d{:}n}$};
\node[left, dimension, rotate=40] at (1.1,0.4) {$\smash{r_1{:}n^2}$};
\end{tz}
&
\mathcal{F}_2=
\def\sidew {1.8} 
\def\dist{3*\hrt} 
\def \xd{0.2}
\def \yd{0.2}
\begin{tz}[string,scale=1.17]
\path[blueregion] (-0.25, 2.5) to [out=down, in=135] (0.5,1.5) to [out=-145, in=45+4*\angledelta] (-0.5,-0.5) to [out=-125, in=up] (-1.15,-1.5) to (-1.15-\sidew, -1.5) to (-1.15-\sidew,2.5);
\draw[bnd] (-0.25, 2.5) to [out=down, in=135] (0.5,1.5) to [out=-145, in=45+4*\angledelta] (-0.5,-0.5) to [out=-125, in=up] (-1.15,-1.5);
\path[blueregion,bnd] (0.5,1.5) to (1,1) to [out=-145, in=45] (-0.5,-0.5) to [out=45+2*\angledelta, in=-125] (0.5,1.5);
\path[blueregion,bnd] (1,1) to (1.5,0.5) to [out=-135, in=45-4*\angledelta] (-0.5,-0.5) to [out=45-2*\angledelta, in=-125] (1,1);
\path[yellowregion] (-0.25-\dist, 2.5-\yd) to [out=down, in=135] (-0.5,-0.5) to [out=-145, in=up] (-1.35,-1.5-\yd) to (-1.15-\xd-\sidew,-1.5-\yd) to (-1.15-\xd-\sidew,2.5-\yd);
\draw[bnd] (-0.25-\dist, 2.5-\yd) to [out=down, in=135] (-0.5,-0.5) to [out=-145, in=up] (-1.35,-1.5-\yd);
\draw (-0.5,-0.5) to [out=-45, in=up] (0.25,-1.5);
\draw (1.5,0.5) to [out=-45, in=up] (0.25+3*\hrt, -1.5);
\draw (1.5,0.5) to [out= 45, in=down] (1.25 + 2*\hrt, 2.5);
\draw (1,1) to [out=45, in=down] (1.25 + \hrt, 2.5);
\draw (0.5,1.5) to [out=45, in=down] (1.25,2.5);
\node[blob] at (0.5,1.5) {$Q$};
\node[blob] at (-0.5,-0.5) {$V$};
\node[blob] at (1,1) {$Q$};
\node[blob] at (1.5,0.5) {$W$};
\node[above right,dimension] at (0.3,-1.4) {$\smash{e{:}n^4}$};
\node[above right,dimension] at (0.3+3*\hrt,-1.4) {$\smash{f{:}n}$};
\node[below left,dimension] at (-2.35,2.1) {$\smash{a{:}n^4}$};
\node[below left,dimension] at (-0.25,2.3) {$\smash{r_0{:}n^2}$};
\node[below right, dimension] at (1.27,2.3) {$\smash{c_1{:}n^2}$};
\node[below right, dimension] at (1.97,2.3) {$\smash{c_2{:}n^2}$};
\node[below right, dimension] at (2+\hrt ,2.3) {$\smash{d{:}n}$};
\node[left, dimension, rotate=37] at (1.0,0.36) {$\smash{r_2{:}n^2}$};
\node[left, dimension, rotate=48] at (0.55,0.89) {$\smash{r_1{:}n^2}$};
\end{tz}
\end{calign}
\begin{calign}
\nonumber
\mathcal{F}_m= 
\def\sidew {2.5} 
\def\dist{3*\hrt} 
\def \xd{0.2}
\def \yd{0.2}
\begin{tz}[string,scale=1.17]
\path[blueregion] (-0.75, 3) to [out=down, in=135] (0,2) to [out=-145, in=50+4*\angledelta] (-0.5,-0.5) to [out= -125, in=up] (-1.15,-1.5) to (-1.15-\sidew, -1.5) to (-1.15-\sidew,3);
\draw[bnd] (-0.75, 3) to [out=down, in=135] (0,2) to [out=-145, in=50+4*\angledelta] (-0.5,-0.5) to [out= -125, in=up] (-1.15,-1.5);
\path[blueregion,bnd] (0,2) to (0.5,1.5) to [out=-145, in=50+2*\angledelta] (-0.5,-0.5) to [out=50+3*\angledelta, in=-125] (0,2);
\path[blueregion,bnd] (0.75, 1.25) to (0.5,1.5)  to [out=-125, in=50+\angledelta] (-0.5,-0.5) to (-0.47,-0.5);
\path[blueregion,bnd] (1.25, 0.75) to (1.5,0.5) to [out=-145, in=45-2*\angledelta] (-0.5,-0.5) to (-0.53,-0.5);
\path[blueregion,bnd] (1.5,0.5) to (2,0) to [out=-135, in=45-4*\angledelta] (-0.5,-0.5) to [out=45-3*\angledelta, in= -125] (1.5,0.5);
\path[yellowregion] (-0.75-\dist, 3-\yd) to [out=down, in=135]  (-0.5,-0.5) to [out=-145, in=up] (-1.35, -1.5-\yd) to (-1.15-\xd-\sidew, -1.6) to (-1.15-\xd-\sidew, 3-\yd);
\draw[bnd]  (-0.75-\dist, 3-\yd) to [out=down, in=135]  (-0.5,-0.5) to [out=-145, in=up] (-1.35, -1.5-\yd);
\draw (-0.5,-0.5) to [out=-45, in=up] (0.25,-1.5);
\draw (2,0) to [out=-45, in=up] (0.25+3.75*\hrt, -1.5);
\draw (2,0) to [out= 45, in=down] (0.75+4*\hrt,3);
\draw (1.5,0.5) to [out= 45, in=down] (0.75+ 3*\hrt, 3);
\draw (0.5,1.5) to [out=45, in=down] (0.75+\hrt, 3);
\draw (0,2) to [out=45, in=down] (0.75,3);
\node[blob] at (0,2) {$Q$};
\node[blob] at (0.5,1.5) {$Q$};
\node[blob] at (-0.5,-0.5) {$V$};
\node[rotate=-45] at (1,1) {$\cdots$};
\node[blob] at (1.5,0.5) {$Q$};
\node[blob] at (2,0) {$W$};
\node[above right,dimension] at (0.3,-1.4) {$\smash{e{:}n^{2m}}$};
\node[above right,dimension] at (0.8+3*\hrt,-1.4) {$\smash{f{:}n}$};
\node[below left,dimension] at (-2.85,2.55) {$\smash{a{:}n^{4m}}$};
\node[below left,dimension] at (-0.75,2.8) {$\smash{r_0{:}n^2}$};
\node[below right, dimension] at (0.76,2.8) {$\smash{c_1{:}n^2}$};
\node[below right, dimension] at (0.76+\hrt,2.8) {$\smash{c_2{:}n^2}$};
\node[below right, dimension] at (0.75+3*\hrt,2.8) {$\smash{c_m\hspace{-0.5pt} {:}n^2}$};
\node[below right, dimension] at (0.76+4*\hrt,2.8) {$\smash{d{:}n}$};
\node[left, dimension, rotate=20] at (1.6,-0.05) {$\smash{r_m{:}n^2}$};
\node[left, dimension,  rotate=72] at (0.03,1.4) {$\smash{r_1{:}n^2}$};
\end{tz}\end{calign}
\vspace{-10pt}
\caption{An infinite sequence of unitary error basis constructions.\label{fig:infinitude}}
\end{figure}

We have seen several examples of unitary error basis constructions which do not factor through compositions of lower arity only involving Hadamard matrices, unitary error bases, quantum Latin squares and controlled families thereof. We now argue that such constructions can be found for all arities, and hence that our methods lead to infinitely many conceptually distinct constructions.

Consider the sequence of unitary error basis constructions presented in \autoref{fig:infinitude}. The construction $\mathcal F_m$ produces an $n^{2m+1}$\-dimensional UEB 
\begin{equation} U_{ar_0,c_1\cdots c_{m}d,ef} := \sum_{r_1\in [n^2]} \cdots \sum_{r_m\in [n^2]} V_{a,r_1 \cdots r_{m},e}^{r_0} \, 
\left( \prod _{i \in [m]} Q_{r_{i-1},r_{i},c_i}^{\vphantom{b}}  \right)
 \, W_{r_m,d,f}^{\vphantom{b}} 
\end{equation}
from the following data, with $a\in [n^{4m}]$, $e \in [n^{2m}]$, $r_i, c_i \in [n^2]$ and $d,f\in [n]$:
\begin{itemize}
\item $V_{a,r_1, \cdots , r_m,e}^{r_0}\in\UEB^{n^2}_{n^{2m}}$, an $n^2$\-controlled family of $n^{2m}$\-dimensional unitary error bases;
\item $Q_{r_{i-1},r_i,c_i}\in \QLS_{n^2}$, an $n^2$\-dimensional quantum Latin square;
\item $W_{r_m,d,f}\in\UEB_n$, an $n$\-dimensional unitary error basis.
\end{itemize}
By inspection, for each $m>1$, the construction $\mathcal{F}_m$ does not factor through any simpler construction between Hadamards, unitary error bases, quantum Latin squares or controlled families thereof.

\section{Equivalence of biunitaries}
\label{sec:equivalence}

In the planar algebra literature, two biunitaries $V$ and $U$ are said to be \textit{gauge equivalent} if there are unitaries $A,B,C$ and $D$ such that the following holds~\cite[Definition 2.11.10]{Jones:1999}:
 \begin{equation}\label{eq:gaugeEquivalent}
\begin{tz}[string,scale=0.75*1.5]
\path[redregion] (0.25,0) to [out=90, in=-135] (0.5,0.5) to (1,1) to (0.5,1.5) to [out= 135, in=-90] (0.25,2) to (0.25-0.66*\side,2) to (0.25-0.66*\side,0);
\path[blueregion] (1.75,0) to [out=90, in=-45] (1.5,0.5) to  (1,1) to (1.5,1.5) to [out= 45, in=-90] (1.75,2) to (1.75+0.66*\side,2) to (1.75+0.66*\side,0);
\path[greenregion,draw] (0.25,0) to [out=90, in=-135] (0.5,0.5)to (1,1) to (1.5,0.5) to [out=-45, in=90] (1.75,0);
\path[yellowregion,draw] (0.25,2) to [out=-90, in=135] (0.5,1.5) to (1,1) to  (1.5,1.5) to [out=45, in=-90] (1.75,2);
\node[blob] at (1,1) {$V$};
\end{tz}~
=~
\begin{tz}[string,scale=0.75]
\path[redregion] (-0.25,-0.5) to [out=90, in=-135] (0.5,0.5) to (1,1) to (0.5,1.5) to [out= 135, in=-90]  (-0.25, 2.5) to (-0.25-\side,2.5) to (-0.25-\side,-0.5);
\path[blueregion] (2.25,-0.5) to [out=90, in=-45] (1.5, 0.5) to  (1,1) to (1.5,1.5) to [out= 45, in=-90] (2.25,2.5) to (2.25+\side,2.5) to (2.25+\side,-0.5);
\path[greenregion,draw] (-0.25,-0.5) to [out=90, in=-135] (0.5,0.5)to (1,1) to (1.5,0.5) to [out=-45, in=90] (2.25,-0.5);
\path[yellowregion,draw] (-0.25,2.5) to [out=-90, in=135] (0.5,1.5) to (1,1) to  (1.5,1.5) to [out=45, in=-90] (2.25,2.5);
\node[blob] at (1,1) {$U$};
\node[blob] at (0.25,1.75) {$A$};
\node[blob] at (1.75,1.75) {$B$};
\node[blob] at (0.25,0.25){$C$};
\node[blob] at (1.75,0.25) {$D$};
\end{tz}
\end{equation} 
However, this notion of equivalence does not coincide with the usual notions  of equivalence of Hadamard matrices, unitary error bases and quantum Latin squares~\cite{Tadej:2006,Klappenecker:2003,Musto:2015}. For example, $n$\-dimensional Hadamard matrices $H$ and $H'$ are said to be equivalent if there are scalars $\lambda_a, \mu_b\in U(1) $ and permutations $\sigma,\tau \in S_n$ such that $H'_{a,b} = \lambda_a\lambda_b H_{\sigma(a), \tau(b)}$ \cite[Definition 2.2]{Tadej:2006}. Applying condition \eqref{eq:gaugeEquivalent} to Hadamard matrices \eqref{eq:typeHadamard} accounts for the scalars $\lambda_a$ and $\mu_b$ but not for the permutations $\sigma$ and $\tau$.

To remedy this, we suggest that two biunitaries should be considered equivalent if each can be obtained from the other by composition with biunitaries. (Note that~\eqref{eq:gaugeEquivalent} arises from biunitary composition of $U$ with the unitaries $A,B,C,D$.) In \autoref{sec:equivalencefoundations}, we make this precise. In \autoref{sec:equivalenceapplications}, we verify that this gives the correct notion of equivalence for quantum structures, and investigate the consequences for some of the construction rules explored in \autoref{sec:composition}.

\subsection{Mathematical foundation}
\label{sec:equivalencefoundations}
\begin{definition} We say that a biunitary $F$ is \textit{minor reversible} if there exists a biunitary $G$ and unitaries $A,B,C,D$ such that the following hold:
\tikzset{zig/.style={}}
 \begin{calign}
\begin{tz}[string]
\path[greenregion] (2.25,2.5) to [out=-90, in=45] (1.5,1.5) 
decorate[zig]{ to [out=-45, in=90] (1.75+\hrt,0)}
to (1.75+\hrt+\side,0) to (1.75+\hrt+\side,2.5);
\path[redregion] (0.25,0) to [out=90, in=-135] (1,1) 
decorate[zig]{to [out=135, in=-90] (0.75-\hrt, 2.5)}
 to (0.75-\hrt-\side,2.5) to (0.75-\hrt-\side,0);
\path[greenregion] 
decorate[zig]{(1.75,0) to [out=90, in=-45] (1,1)}
to [out=-135, in=90] (0.25,0);
\path[redregion] decorate[zig]{(0.75,2.5) to [out=-90, in=135] (1.5,1.5)} to [out=45, in=-90] (2.25,2.5);
\path[ blueregion,draw] 
decorate[zig]{(1.75,0) to [out=90, in=-45] (1,1)}
 to (1.5,1.5) 
 decorate[zig]{ to [out=-45, in=90]  (\hrt+1.75,0)};
\path[yellowregion,draw] 
decorate[zig]{(0.75,2.5) to [out=-90, in=135] (1.5,1.5)}
 to (1,1)
 decorate[zig]{to [out=135, in=-90] (0.75-\hrt,2.5)};
 \draw (0.25,0) to [out=90, in=-135] (1,1) to (1.5,1.5) to [out=45, in=-90] (2.25,2.5);
\node[blob] at (1,1) {$F$};
\node[blob] at (1.5,1.5) {$G$};
\node[scale=0.8] at (0.75-\hrt, 2.8) {$\Sigma\vphantom{\Sigma'}$};
\node[scale=0.8] at (0.75, 2.8) {$\Sigma'$};
\node[scale=0.8] at (1.75, -0.3) {$\Delta\vphantom{\Delta'}$};
\node[scale=0.8] at (1.75+\hrt, -0.3) {$\Delta'$};
\end{tz} 
\eqgap =\eqgap
\begin{tz}[string]
\path[redregion] (0.25,0) to [out=90, in=-135] (1,1) to (1.5,1.5) to [out=45, in=-90] (2.25,2.5) to (0.75-\hrt-\side,2.5) to (0.75-\hrt-\side,0);
\path[greenregion] (0.25,0) to [out=90, in=-135] (1,1) to  (1.5,1.5) to [out=45, in=-90] (2.25,2.5) to (1.75+\hrt+\side,2.5) to (1.75+\hrt+\side,0);
\draw(0.25,0) to [out=90, in=-135] (1,1) to(1.5,1.5) to [out=45, in=-90]  (2.25,2.5);
\path[fill=white,zig, decorate] (0.75-\hrt,2.5) to (0.75-\hrt,1.8) to (0.75,1.8) to (0.75,2.5);
\path[yellowregion,zig,decorate,draw] (0.75-\hrt,2.5) to (0.75-\hrt,1.8) to (0.75,1.8) to (0.75,2.5);
\path[fill=white,draw] (0.65- \hrt ,1.6) rectangle (0.85,2);
\node[scale=0.8] at (0.75-0.5*\hrt,1.8) {$A$};
\path[fill=white,zig, decorate] (1.75,0) to (1.75,0.7) to (1.75+\hrt,0.7) to (1.75+\hrt,0);
\path[blueregion,zig,decorate,draw]  (1.75,0) to (1.75,0.7) to (1.75+\hrt,0.7) to (1.75+\hrt,0.);
\path[fill=white,draw] (1.65 ,0.5) rectangle (1.85+\hrt,0.9);
\node[scale=0.8] at (1.75+0.5*\hrt,0.7) {$B$};
\node[scale=0.8] at (0.75-\hrt, 2.8) {$\Sigma\vphantom{\Sigma'}$};
\node[scale=0.8] at (0.75, 2.8) {$\Sigma'$};
\node[scale=0.8] at (1.75, -0.3) {$\Delta\vphantom{\Delta'}$};
\node[scale=0.8] at (1.75+\hrt, -0.3) {$\Delta'$};
\end{tz}
\\
\begin{tz}[string]
\path[blueregion] (2.25,2.5) to [out=-90, in=45] (1.5,1.5) 
decorate[zig]{ to [out=-45, in=90] (1.75+\hrt,0)}
to (1.75+\hrt+\side,0) to (1.75+\hrt+\side,2.5);
\path[yellowregion] (0.25,0) to [out=90, in=-135] (1,1) 
decorate[zig]{to [out=135, in=-90] (0.75-1/2^0.5, 2.5)}
 to  (0.75-\hrt-\side,2.5) to (0.75-\hrt-\side,0);
\path[blueregion] 
decorate[zig]{(1.75,0) to [out=90, in=-45] (1,1)}
to [out=-135, in=90] (0.25,0);
\path[yellowregion] decorate[zig]{(0.75,2.5) to [out=-90, in=135] (1.5,1.5)} to [out=45, in=-90] (2.25,2.5);
\path[ greenregion,draw] 
decorate[zig]{(1.75,0) to [out=90, in=-45] (1,1)}
 to (1.5,1.5) 
 decorate[zig]{ to [out=-45, in=90]  (\hrt+1.75,0)};
\path[redregion,draw] 
decorate[zig]{(0.75,2.5) to [out=-90, in=135] (1.5,1.5)}
 to (1,1)
 decorate[zig]{to [out=135, in=-90] (0.75-\hrt,2.5)};
 \draw (0.25,0) to [out=90, in=-135] (1,1) to (1.5,1.5) to [out=45, in=-90] (2.25,2.5);
\node[blob] at (1,1) {$G$};
\node[blob] at (1.5,1.5) {$F$};
\node[scale=0.8] at (0.75-\hrt, 2.8) {$\Sigma'$};
\node[scale=0.8] at (0.75, 2.8) {$\Sigma\vphantom{\Sigma'}$};
\node[scale=0.8] at (1.75, -0.3) {$\Delta'$};
\node[scale=0.8] at (1.75+\hrt, -0.3) {$\Delta\vphantom{\Delta'}$};
\end{tz} 
\eqgap = \eqgap
\begin{tz}[string]
\path[yellowregion] (0.25,0) to [out=90, in=-135] (1,1) to (1.5,1.5) to [out=45, in=-90] (2.25,2.5) to (0.75-\hrt-\side,2.5) to (0.75-\hrt-\side,0);
\path[blueregion] (0.25,0) to [out=90, in=-135] (1,1) to  (1.5,1.5) to [out=45, in=-90] (2.25,2.5) to (1.75+\hrt+\side,2.5) to (1.75+\hrt+\side,0);
\draw(0.25,0) to [out=90, in=-135] (1,1) to(1.5,1.5) to [out=45, in=-90]  (2.25,2.5);
\path[fill=white,zig, decorate] (0.75-\hrt,2.5) to (0.75-\hrt,1.8) to (0.75,1.8) to (0.75,2.5);
\path[redregion,zig,decorate,draw] (0.75-\hrt,2.5) to (0.75-\hrt,1.8) to (0.75,1.8) to (0.75,2.5);
\path[fill=white,draw] (0.65- \hrt ,1.6) rectangle (0.85,2);
\node[scale=0.8] at (0.75-0.5*\hrt,1.8) {$C$};
\path[fill=white,zig, decorate] (1.75,0) to (1.75,0.7) to (1.75+\hrt,0.7) to (1.75+\hrt,0);
\path[greenregion,zig,decorate,draw]  (1.75,0) to (1.75,0.7) to (1.75+\hrt,0.7) to (1.75+\hrt,0.);
\path[fill=white,draw] (1.65 ,0.5) rectangle (1.85+\hrt,0.9);
\node[scale=0.8] at (1.75+0.5*\hrt,0.7) {$D$};
\node[scale=0.8] at (0.75-\hrt, 2.8) {$\Sigma'$};
\node[scale=0.8] at (0.75, 2.8) {$\Sigma\vphantom{\Sigma'}$};
\node[scale=0.8] at (1.75, -0.3) {$\Delta'$};
\node[scale=0.8] at (1.75+\hrt, -0.3) {$\Delta\vphantom{\Delta'}$};
\end{tz}
\end{calign}
\end{definition}

\noindent
(Note that our usage of ``unitary'' is that of \autoref{def:unitary2morphism}.) That is, a biunitary is minor reversible if it is invertible with respect to biunitary composition along the minor diagonal direction $\begin{tz}[string,scale=0.25]
\draw (0,0 ) to (0.75,1);
\end{tz}$. Similarly, we say that a biunitary is \textit{major reversible} if it is invertible with respect to biunitary composition along the major diagonal direction $\begin{tz}[string,scale=0.25,xscale=-1]
\draw (0,0 ) to (0.75,1);
\end{tz}$.

For example, a unitary $U$ can be seen both as a minor reversible biunitary and a major reversible biunitary, respectively, depending on the chosen partition of the input and output wires:
\begin{calign}\begin{tz}[string]
\path[redregion] (0.25,0) to [out=90, in=-135] (1,1) to [out=45, in=-90] (1.75,2) to (0.25-\side,2) to (0.25-\side,0);
\path[greenregion] (0.25,0) to [out=90, in=-135] (1,1) to [out=45, in=-90] (1.75,2) to (1.75+\side,2) to (1.75+\side,0);
\draw(0.25,0) to [out=90, in=-135] (1,1) to [out=45, in=-90] (1.75,2);
\node[blob] at (1,1) {$U$};
\end{tz}
&
\begin{tz}[string,xscale=-1]
\path[redregion] (0.25,0) to [out=90, in=-135] (1,1) to [out=45, in=-90] (1.75,2) to (0.25-\side,2) to (0.25-\side,0);
\path[greenregion] (0.25,0) to [out=90, in=-135] (1,1) to [out=45, in=-90] (1.75,2) to (1.75+\side,2) to (1.75+\side,0);
\draw(0.25,0) to [out=90, in=-135] (1,1) to [out=45, in=-90] (1.75,2);
\node[blob] at (1,1) {$U$};
\end{tz}
\end{calign}
If $F$ is minor reversible, then since $A$ and $C$ are invertible, it follows that the non-negative integer-valued matrix $\sigma:= (\dim(\Sigma_{a,b}) )_{a,b}$ is invertible with non-negative integer-valued inverse $\sigma' = (\dim(\Sigma'_{b,a}))_{b,a}$. In this case $\Sigma$ defines a bijection on the label sets of the two adjacent regions, and we say that $\Sigma$ is an \textit{equivalence}, drawing it as follows:

\begin{equation}
\begin{tz}[string,scale=0.8]
\path[redregion] decorate[zig]{(1,0) to(1,2)} to (0,2) to (0,0);
\path[yellowregion] decorate[zig] {(1,0) to(1,2)} to (2,2) to (2,0);
\draw[zig,decorate] (1,0) to (1,2);
\node[scale=0.8] at (1,2.3) {$\Sigma$};
\node[dimension, below right] at (0,1.75) {$\smash{a{:}n}$};
\node[dimension, below left] at (2,1.75) {$\smash{b{:}n}$};
\end{tz} 
\hspace{30pt}\leftrightsquigarrow \hspace{30pt} \delta_{\sigma(a),b}
\end{equation}

It follows that a minor-reversible biunitary with no shaded region is simply a unitary
\begin{equation} \begin{tz}[string]
\draw (0.25,0) to [out= up, in =-135] (1,1) to [out= 45,in = down] (1.75,2);
\node[blob] at (1,1) {$U$};
\end{tz}
\end{equation}

\noindent
with $\Sigma$ and $\Delta$ being the identity bijections between 1\-element sets. A minor-reversible biunitary of the form 
\begin{equation}\begin{tz}[string]
\path[blueregion] (0.25,0) to[out = up, in =-135] (1,1) to [out= 45, in = down] (1.75,2) to (0.25-\side,2) to (0.25-\side,0);
\path[bnd] (0.25,0) to[out = up, in =-135] (1,1) to [out= 45, in = down] (1.75,2);
\path[draw,zig,decorate] (1,1) to [out= 135, in = down] (0.25,2);

\node[blob] at (1,1) {$\lambda$};
\node[scale=0.8] at (0.25, 2.3) {$\Sigma$};
\end{tz}\end{equation}
corresponds to a controlled family of scalars $\{\lambda_a\in U(1)\,|\,a\in \mathcal{A}\}$ with $\Sigma$ acting as a permutation on the index set.

We make the following observation.
\begin{proposition}
\label{lem:bijectionbiunitary}
If $\Sigma$ is an equivalence, then the following vertex is biunitary:
\begin{equation}
\begin{tz}[string]
\path[redregion] 
decorate[zig] {(0.25,0) to [out=90, in=-90] (1,1) to [out=90, in=-90] (0.25,2)}
to (0.25-\side,2) to (0.25-\side,0);
\path[yellowregion]
decorate[zig]{(0.25,0) to [out=90, in=-90] (1,1) to [out=90, in=-90] (0.25,2)}
to (1.75+\side,2) to (1.75+\side,0);
\draw[zig,decorate] (0.25,0) to [out=90, in=-90] (1,1) to [out=90, in=-90] (0.25,2);
\node[scale=0.8] at (0.25,-0.3) {$\Sigma$};
\node[scale=0.8] at (0.25,2.3) {$\Sigma$};
\end{tz}
\end{equation}
\end{proposition}
\begin{proof}

Vertical unitarity is immediate, and horizontal unitarity follows from these calculations:
\begin{align*}
\hspace{-30pt}
\begin{tz}[string,scale=0.9]
\path[yellowregion] (0.25-\side, 0) rectangle (1.75+\side, 2);
\node[dimension, below right] at (0.25-\side,1.75) {$\smash{i{:}n}$};
\end{tz}
\,&=\,
\begin{tz}[string,scale=0.9]
\path[yellowregion] (0.25-\side, 0) rectangle (1.75+\side, 2);
\path[fill=white,zig,decorate] (1,1) circle (0.75);
\path[redregion,draw,zig,decorate] (1,1) circle (0.75);
\node[dimension, below right] at (0.25-\side,1.75) {$\smash{i{:}n}$};
\node[dimension] at (1,1) {$r{:}n$};
\node[scale=0.8,right] at (1.75,1) {$\Sigma$};
\end{tz}\hspace{10pt} &&\leftrightsquigarrow&\hspace{10pt}
&1\,=\,\sum_{r\in[n]}\delta_{\sigma(r),i}\hspace{10pt} \forall i\in[n]
\\
\begin{tz}[string,scale=0.9]
\clip (0.25-\side, -0.5) rectangle (1.75+\side, 2.5);
\path[yellowregion] (0.25-\side, 0) rectangle (1.75+\side, 2);
\path[fill=white]decorate[zig]{ (0.25,0)  to +(0,2)} to (1.75,2) 
decorate[zig]{ to (1.75,0)};
\path[redregion]decorate[zig]{ (0.25,0)  to +(0,2)} to (1.75,2) 
decorate[zig]{ to (1.75,0)};
\path[draw,zig,decorate] (0.25,0) to +(0,2);
\path[draw,zig,decorate] (1.75,2) to +(0,-2);
\node[dimension, below right] at (0.25-\side,1.75) {$\smash{i{:}n}$};
\node[dimension, below left] at (1.75+\side,1.75) {$\smash{j{:}n}$};
\node[scale=0.8]at  (1.75,-0.3) {$\Sigma^{\vphantom{-1}}$};
\node[scale=0.8]at  (0.25,-0.3) {$\Sigma^{-1}$};
\node[scale=0.8] at (0.25,2.3) {$\Sigma^{-1}$};
\node[scale=0.8] at (1.75,2.3) {$\Sigma^{\vphantom{-1}}$};
\node[dimension, below] at (1,1.75) {$\smash{k{:}n}$};
\node[dimension, below] at (1,0.15) {$\smash{r{:}n}$};
\end{tz}
\,&=\,
\begin{tz}[string,scale=0.9]
\clip (0.25-\side, -0.5) rectangle (1.75+\side, 2.5);
\path[yellowregion] (0.25-\side, 0) rectangle (1.75+\side, 2);
\path[fill=white,zig,decorate] (0.25,0) to [out= up, in =left] (1,0.75) to [out= right, in = up] (1.75,0);
\path[redregion,draw,zig,decorate] (0.25,0) to [out= up, in =left] (1,0.75) to [out= right, in = up] (1.75,0);
\path[fill=white,zig,decorate] (0.25,2) to [out= down, in =left] (1,1.25) to [out= right, in = down] (1.75,2);
\path[redregion,draw,zig,decorate] (0.25,2) to [out= down, in =left] (1,1.25) to [out= right, in = down] (1.75,2);
\node[dimension, below right] at (0.25-\side,1.75) {$\smash{i{:}n}$};
\node[dimension, below left] at (1.75+\side,1.75) {$\smash{j{:}n}$};
\node[dimension, below] at (1,1.75) {$\smash{k{:}n}$};
\node[dimension, below] at (1,0.15) {$\smash{r{:}n}$};
\node[scale=0.8] at (1.75,-0.3) {$\Sigma^{\vphantom{-1}}$};
\node[scale=0.8] at (1.75,2.3) {$\Sigma^{\vphantom{-1}}$};
\node[scale=0.8] at (0.25,-0.3) {$\Sigma^{-1}$};
\node[scale=0.8] at (0.25,2.3) {$\Sigma^{-1}$};
\end{tz}\hspace{10pt} &&\leftrightsquigarrow&\hspace{10pt}
&\delta_{i,\sigma(k)} \delta_{k,r} \delta_{\sigma(k),j}\,
\\[-1.1cm]
&&&&&\hspace{0cm}=\,\delta_{i,j} \delta_{\sigma(k),j}\delta_{\sigma(r),j}\hspace{10pt} \forall i,j,k,r\in[n]
\end{align*}
\end{proof}

We are now ready to state the definition of equivalence of biunitaries, in which the unitaries in \eqref{eq:gaugeEquivalent} are replaced by minor- and major-reversible biunitaries.

\begin{definition} Two biunitaries $U,V$ are \textit{equivalent} if there exist minor-reversible biunitaries $B,C$ and major-reversible biunitaries $A,D$
\begin{calign}
\hspace{-30pt}
\begin{tz}[string]
\path[yellowregion] decorate[zig]{(0.25,2) to [out=-90, in=135] (0.5,1.5) to (1,1)} to (0.5,0.5) to [out= -135, in=90] (0.25,0) to (0.25-\side,0) to (0.25-\side,2);
\path[blueregion] (1.75,2) to [out=-90, in=45] (1.5,1.5) to  (1,1) 
decorate[zig] { to (1.5,0.5) to [out= -45, in=90] (1.75,0)}
 to  (1.75+\side,0) to (1.75+\side,2);
\path[blueregion,draw] 
(0.25,0) to [out=90, in=-135] (0.5,0.5)to (1,1)
decorate[zig]{ to (1.5,0.5) to [out=-45, in=90] (1.75,0)};
\path[yellowregion,draw] decorate[zig]{(0.25,2) to [out=-90, in=135] (0.5,1.5) to (1,1)}
 to  (1.5,1.5) to [out=45, in=-90] (1.75,2);
\node[blob] at (1,1) {$B$};
\node[scale=0.8] at (1.75,-0.3) {$\Delta^{-1}$};
\node[scale=0.8] at (0.25,2.3) {$\Sigma ^{-1}$};
\end{tz}
&
\begin{tz}[string]
\path[redregion] decorate[zig]{(0.25,2) to [out=-90, in=135] (0.5,1.5) to (1,1)} to (0.5,0.5) to [out= -135, in=90] (0.25,0) to (0.25-\side,0) to (0.25-\side,2);
\path[greenregion] (1.75,2) to [out=-90, in=45] (1.5,1.5) to  (1,1) 
decorate[zig] { to (1.5,0.5) to [out= -45, in=90] (1.75,0)}
 to (1.75+\side,0) to (1.75+\side,2);
\path[greenregion,draw] 
(0.25,0) to [out=90, in=-135] (0.5,0.5)to (1,1)
decorate[zig]{ to (1.5,0.5) to [out=-45, in=90] (1.75,0)};
\path[redregion,draw] decorate[zig]{(0.25,2) to [out=-90, in=135] (0.5,1.5) to (1,1)}
 to  (1.5,1.5) to [out=45, in=-90] (1.75,2);
\node[blob] at (1,1) {$C$};
\node[scale=0.8] at (1.75,-0.3) {$\Lambda $};
\node[scale=0.8] at (0.25,2.3) {$\Theta$};
\end{tz}
&
\begin{tz}[string,xscale=-1]
\path[yellowregion] decorate[zig]{(0.25,2) to [out=-90, in=135] (0.5,1.5) to (1,1)} to (0.5,0.5) to [out= -135, in=90] (0.25,0) to (0.25-\side,0) to (0.25-\side,2);
\path[redregion] (1.75,2) to [out=-90, in=45] (1.5,1.5) to  (1,1) 
decorate[zig] { to (1.5,0.5) to [out= -45, in=90] (1.75,0)}
 to  (1.75+\side,0) to (1.75+\side,2);
\path[redregion,draw] 
(0.25,0) to [out=90, in=-135] (0.5,0.5)to (1,1)
decorate[zig]{ to (1.5,0.5) to [out=-45, in=90] (1.75,0)};
\path[yellowregion,draw] decorate[zig]{(0.25,2) to [out=-90, in=135] (0.5,1.5) to (1,1)}
 to  (1.5,1.5) to [out=45, in=-90] (1.75,2);
\node[blob] at (1,1) {$A$};
\node[scale=0.8] at (1.75,-0.3) {$\Theta$};
\node[scale=0.8] at (0.25,2.3) {$\Sigma$};
\end{tz}
&
\begin{tz}[string,xscale=-1]
\path[blueregion] decorate[zig]{(0.25,2) to [out=-90, in=135] (0.5,1.5) to (1,1)} to (0.5,0.5) to [out= -135, in=90] (0.25,0) to  (0.25-\side,0) to (0.25-\side,2);
\path[greenregion] (1.75,2) to [out=-90, in=45] (1.5,1.5) to  (1,1) 
decorate[zig] { to (1.5,0.5) to [out= -45, in=90] (1.75,0)}
 to (1.75+\side,0) to (1.75+\side,2);
\path[greenregion,draw] 
(0.25,0) to [out=90, in=-135] (0.5,0.5)to (1,1)
decorate[zig]{ to (1.5,0.5) to [out=-45, in=90] (1.75,0)};
\path[blueregion,draw] decorate[zig]{(0.25,2) to [out=-90, in=135] (0.5,1.5) to (1,1)}
 to  (1.5,1.5) to [out=45, in=-90] (1.75,2);
\node[blob] at (1,1) {$D$};
\node[scale=0.8] at (1.75,-0.3) {$\Lambda^{-1}$};
\node[scale=0.8] at (0.25,2.3) {$\Delta^{-1}$};
\end{tz}\end{calign}
such that the following equation holds: 
\begin{equation}
\label{eq:biunitaryequivalence}
\begin{tz}[string,scale=2.5/2*1.5]
\path[redregion] (0.25,0) to [out=90, in=-135] (0.5,0.5) to (1,1) to (0.5,1.5) to [out= 135, in=-90] (0.25,2) to (0.25-0.8*\side,2) to (0.25-0.8*\side,0);
\path[blueregion] (1.75,0) to [out=90, in=-45] (1.5,0.5) to  (1,1) to (1.5,1.5) to [out= 45, in=-90] (1.75,2) to (1.75+0.8*\side,2) to (1.75+0.8*\side,0);
\path[greenregion,draw] (0.25,0) to [out=90, in=-135] (0.5,0.5)to (1,1) to (1.5,0.5) to [out=-45, in=90] (1.75,0);
\path[yellowregion,draw] (0.25,2) to [out=-90, in=135] (0.5,1.5) to (1,1) to  (1.5,1.5) to [out=45, in=-90] (1.75,2);
\node[blob] at (1,1) {$V$};
\end{tz}
\eqgap = \eqgap
 \begin{tz}[string,scale=1.5]
\path[redregion] (0.5,0.25) to [out=90, in=-135] (1,1) to (1.5,1.5) to (1,2) to [out=135, in=-90] (0.5,2.75) to (0.5-\side,2.75) to (0.5-\side,0.25);
\path[blueregion] (2.5,0.25) to [out=90, in=-45] (2,1) to (1.5,1.5) to (2,2) to [out=45, in=-90] (2.5,2.75) to (2.5+\side,2.75) to (2.5+\side,0.25);
\path[greenregion,draw] (0.5,0.25) to [out=90, in=-135] (1,1) to (1.5,1.5) to (2,1) to [out=-45, in=90] (2.5,0.25);
\path[yellowregion,draw] (0.5,2.75) to [out=-90, in=135] (1,2) to (1.5,1.5) to (2,2) to [out=45, in=-90] (2.5,2.75);

\draw[zig,decorate] (1,1) to [out=-45, in=180] (1.5,0.65) to [out=0, in=-135] (2,1) to [out=45, in=-90] (2.35,1.5) to [out=90, in=-45] (2,2) to [out=135, in=0] (1.5,2.35) to [out=180, in=45] (1,2) to [out=-135, in=90] (0.65,1.5) to [out=-90, in=135] (1,1);
\node[blob] at (1.5,1.5) {$U$};
\node[blob] at (1,1) {$C$};
\node[blob] at (1,2) {$A$};
\node[blob] at (2,1) {$D$};
\node[blob] at (2,2){$B$};
\node[scale=0.8] at (1.5, 2.5) {$\Sigma$};
\node[scale=0.8] at (0.5,1.5) {$\Theta$};
\node[scale=0.8] at (2.5,1.5) {$\Delta$};
\node[scale=0.8] at (1.5,0.5) {$\Lambda$};
\end{tz}
\end{equation}
It is easy to check that this defines an equivalence relation on the set of biunitaries. Note that the right-hand side of \eqref{eq:biunitaryequivalence} is a composite of 9 biunitaries, thanks to \autoref{lem:bijectionbiunitary}.
\end{definition}

\subsection{Equivalence for quantum structures}
\label{sec:equivalenceapplications}

This leads to the following notions of equivalence of Hadamard matrices, unitary error bases and quantum Latin squares, agreeing with the respective notions proposed in the literature  \cite{Tadej:2006,Klappenecker:2003,Musto:2015}.

\paragraph{Hadamard matrices.} Two Hadamard matrices $H$ and $W$ are equivalent if the following equation holds: 
\begin{equation}\begin{tz}[string,scale=2.5/2*1.25]
\path[blueregion] (0.25,0) to [out=90, in=-135] (0.5,0.5) to (1,1) to (0.5,1.5) to [out= 135, in=-90] (0.25,2) to (0.25-0.8*\side,2) to (0.25-0.8*\side,0);
\path[blueregion] (1.75,0) to [out=90, in=-45] (1.5,0.5) to  (1,1) to (1.5,1.5) to [out= 45, in=-90] (1.75,2) to (1.75+0.8*\side,2) to (1.75+0.8*\side,0);
\path[bnd] (0.25,0) to [out=90, in=-135] (0.5,0.5)to (1,1) to (1.5,0.5) to [out=-45, in=90] (1.75,0);
\path[bnd] (0.25,2) to [out=-90, in=135] (0.5,1.5) to (1,1) to  (1.5,1.5) to [out=45, in=-90] (1.75,2);
\node[blob] at (1,1) {$W$};
\end{tz}\eqgap=\eqgap
 \begin{tz}[string,scale=1.25]
\path[blueregion] (0.5,0.25) to [out=90, in=-135] (1,1) to (1.5,1.5) to (1,2) to [out=135, in=-90] (0.5,2.75) to (0.5-\side,2.75) to (0.5-\side,0.25);
\path[blueregion] (2.5,0.25) to [out=90, in=-45] (2,1) to (1.5,1.5) to (2,2) to [out=45, in=-90] (2.5,2.75) to (2.5+\side,2.75) to (2.5+\side,0.25);
\path[bnd] (0.5,0.25) to [out=90, in=-135] (1,1) to (1.5,1.5) to (2,1) to [out=-45, in=90] (2.5,0.25);
\path[bnd] (0.5,2.75) to [out=-90, in=135] (1,2) to (1.5,1.5) to (2,2) to [out=45, in=-90] (2.5,2.75);

\draw[zig,decorate]  (2,1) to [out=45, in=-90] (2.35,1.5) to [out=90, in=-45] (2,2) ;
\draw[zig,decorate] (1,2) to [out=-135, in=90] (0.65,1.5) to [out=-90, in=135] (1,1);
\node[blob] at (1.5,1.5) {$H$};
\node[blob] at (1,1) {$\mu$};
\node[blob] at (1,2) {$\lambda$};
\node[blob] at (2,1) {$\beta$};
\node[blob] at (2,2){$\alpha$};
\node[scale=0.8] at (0.45,1.5){$\sigma$};
\node[scale=0.8] at (2.55,1.5) {$\tau$};
\end{tz}
\end{equation}
Thus, $H$ and $W$ are equivalent if there are scalars $\lambda_a,\mu_a,\alpha_b$ and $\beta_b$ and permutations $\sigma,\tau \in S_n$ such that 
\begin{equation}W_{a,b} = \lambda_a\mu_a H_{\sigma(a),\tau(b)}\alpha_b\beta_b.
\end{equation}
Redefining $c_a:= \lambda_a\mu_a$ and $d_b:= \alpha_b\beta_b$, this becomes equivalent to the usual notion of equivalence of Hadamard matrices $H$ and $W$: there are scalars $c_a, d_b \in U(1)$ and permutations $\sigma,\tau \in S_n$ such that 
\begin{equation}W_{a,b} = c_a d_b H_{\sigma(a),\tau(b)}.\end{equation}

\paragraph{Unitary error bases.} Two unitary error bases
\begin{calign}
\mathcal{U} = \left\{ U_i~|~i\in [n^2] \right\}
&
\mathcal{V} = \left\{ V_i~|~i\in[ n^2] \right\}
\end{calign}
are equivalent if the following holds:
\begin{equation}\label{eq:equivUEB}\begin{tz}[string,scale=2.5/2*1.25]
\path[blueregion] (0.25,0) to [out=90, in=-135] (0.5,0.5) to (1,1) to (0.5,1.5) to [out= 135, in=-90] (0.25,2) to (0.25-0.8*\side,2) to (0.25-0.8*\side,0);
\path[bnd] (0.25,0) to [out= up, in = -135] (0.5,0.5) to (1,1) to (0.5,1.5) to [out= 135, in = down] (0.25,2);
\path[draw] (1.75,0) to [out= up, in = -45] (1.5,0.5) to (1,1) to (1.5,1.5) to [out= 45, in = down] (1.75,2);
\node[blob] at (1,1) {$V$};
\end{tz}\eqgap=\eqgap
 \begin{tz}[string,scale=1.25]
\path[blueregion] (0.5,0.25) to [out=90, in=-135] (1,1) to (1.5,1.5) to (1,2) to [out=135, in=-90] (0.5,2.75) to (0.5-\side,2.75) to (0.5-\side,0.25);
\path[bnd] (0.5,0.25) to [out=90, in=-135] (1,1) to (1.5,1.5) to (1,2) to [out=135, in=-90] (0.5,2.75) ;

\path[draw] (2.5,0.25) to [out = up, in = -45] (2,1) to (1.5,1.5) to (2,2) to[out=45, in = down] (2.5,2.75);

\draw[zig,decorate] (1,2) to [out=-135, in=90] (0.65,1.5) to [out=-90, in=135] (1,1);
\node[blob] at (1.5,1.5) {$U$};
\node[blob] at (1,1) {$\mu$};
\node[blob] at (1,2) {$\lambda$};
\node[blob] at (2,1) {$B$};
\node[blob] at (2,2){$A$};
\node[scale=0.8] at (0.45,1.5){$\sigma$};
\end{tz},
\end{equation} That is, they are equivalent if there are unitary matrices $A,B$, scalars $c_i \in U(1)$ and a permutation $\sigma \in S_{n^2}$ such that
\begin{equation}V_i = c_i A U_{\sigma(i)} B.\end{equation}

\paragraph{Quantum Latin squares.} Two quantum Latin squares $Q$ and $P$ are equivalent if the following holds:
\begin{equation}
 \begin{tz}[string,scale=1.5]
\path[blueregion] (0.575,0.25) to [out=90, in=-145]  (1.5,1.5) to(1,2) to [out=135, in=-90] (0.5,2.75) to (0.5-0.8*\side,2.75) to (0.5-0.8*\side,0.25);
\draw[bnd](0.575,0.25) to [out=90, in=-145]  (1.5,1.5) to(1,2) to [out=135, in=-90] (0.5,2.75);
\path[blueregion,bnd] (0.675,0.25) to [out=90, in=-125]  (1.5,1.5) to [out=-55, in=90] (2.5,0.25);
\path[draw]  (1.5,1.5) to (2,2) to [out=45, in=-90] (2.5,2.75);

\node[blob] at (1.5,1.5) {$P$};
\end{tz}\eqgap=\eqgap
 \begin{tz}[string,scale=1.5]
\path[blueregion] (0.575,0.25) to [out=90, in=-145]  (1.5,1.5) to(1,2) to [out=135, in=-90] (0.5,2.75) to (0.5-0.8*\side,2.75) to (0.5-0.8*\side,0.25);
\draw[bnd](0.575,0.25) to [out=90, in=-145]  (1.5,1.5) to(1,2) to [out=135, in=-90] (0.5,2.75);

\path[blueregion,bnd] (0.675,0.25) to [out=90, in=-125]  (1.5,1.5) to [out=-55, in=90] (2.5,0.25);
\path[draw] (1.5,1.5) to (2,2) to [out=45, in=-90] (2.5,2.75);

\draw[zig,decorate] (1,1) to [out=-45, in=180] (1.5,0.65) to [out=0, in=-135] (2,1);
\draw[zig,decorate] (1,2) to [out=-135, in=90] (0.65,1.5) to [out=-90, in=135] (1,1);
\node[blob] at (1.5,1.5) {$Q$};
\node[blob] at (1,1) {$\lambda$};
\node[blob] at (1,2) {$\alpha$};
\node[blob] at (2,1) {$\beta$};
\node[blob] at (2,2){$U$};
\node[scale=0.8] at (0.45,1.5){$\sigma$};
\node[scale=0.8] at (1.5,0.45){$\tau$};
\end{tz}
\end{equation}
That is, they are equivalent if there is a unitary matrix $U$, scalars $c_{a,b}\in U(1)$ and permutations $\sigma,\tau \in S_n$ such that the following holds:
\begin{equation}\ket{P_{a,b}} = c_{a,b}U\ket{Q_{\sigma(a),\tau(b)}}\end{equation}
\\

Equivalence of controlled families of quantum structures can be defined in a similar way.\\
It is instructive to consider how the notion of equivalence interacts with composition of biunitaries. Consider two pairs of equivalent biunitaries of the following type:
\begin{calign}
\begin{tz}[string,scale=2.5/2*0.85]
\path[redregion] (0.25,0) to [out=90, in=-135] (0.5,0.5) to (1,1) to (0.5,1.5) to [out= 135, in=-90] (0.25,2) to (0.25-0.8*\side,2) to (0.25-0.8*\side,0);
\path[blueregion] (1.75,0) to [out=90, in=-45] (1.5,0.5) to  (1,1) to (1.5,1.5) to [out= 45, in=-90] (1.75,2) to (1.75+0.8*\side,2) to (1.75+0.8*\side,0);
\path[greenregion,draw] (0.25,0) to [out=90, in=-135] (0.5,0.5)to (1,1) to (1.5,0.5) to [out=-45, in=90] (1.75,0);
\path[yellowregion,draw] (0.25,2) to [out=-90, in=135] (0.5,1.5) to (1,1) to  (1.5,1.5) to [out=45, in=-90] (1.75,2);
\node[blob] at (1,1) {$U$};
\end{tz}
\,=\,
 \begin{tz}[string,scale=0.85]
\path[redregion] (0.5,0.25) to [out=90, in=-135] (1,1) to (1.5,1.5) to (1,2) to [out=135, in=-90] (0.5,2.75) to (0.5-\side,2.75) to (0.5-\side,0.25);
\path[blueregion] (2.5,0.25) to [out=90, in=-45] (2,1) to (1.5,1.5) to (2,2) to [out=45, in=-90] (2.5,2.75) to (2.5+\side,2.75) to (2.5+\side,0.25);
\path[greenregion,draw] (0.5,0.25) to [out=90, in=-135] (1,1) to (1.5,1.5) to (2,1) to [out=-45, in=90] (2.5,0.25);
\path[yellowregion,draw] (0.5,2.75) to [out=-90, in=135] (1,2) to (1.5,1.5) to (2,2) to [out=45, in=-90] (2.5,2.75);
\draw[zig,decorate] (1,1) to [out=-45, in=180] (1.5,0.65) to [out=0, in=-135] (2,1) to [out=45, in=-90] (2.35,1.5) to [out=90, in=-45] (2,2) to [out=135, in=0] (1.5,2.35) to [out=180, in=45] (1,2) to [out=-135, in=90] (0.65,1.5) to [out=-90, in=135] (1,1);
\node[blob] at (1.5,1.5) {$U'$};
\node[blob] at (1,1) {$C$};
\node[blob] at (1,2) {$A$};
\node[blob] at (2,1) {$D$};
\node[blob] at (2,2){$B$};
\end{tz}
&
\begin{tz}[string,scale=2.5/2*0.85]
\path[yellowregion] (0.25,0) to [out=90, in=-135] (0.5,0.5) to (1,1) to (0.5,1.5) to [out= 135, in=-90] (0.25,2) to (0.25-0.8*\side,2) to (0.25-0.8*\side,0);
\path[orangeregion] (1.75,0) to [out=90, in=-45] (1.5,0.5) to  (1,1) to (1.5,1.5) to [out= 45, in=-90] (1.75,2) to  (1.75+0.8*\side,2) to (1.75+0.8*\side,0);
\path[blueregion,draw] (0.25,0) to [out=90, in=-135] (0.5,0.5)to (1,1) to (1.5,0.5) to [out=-45, in=90] (1.75,0);
\path[cyanregion,draw] (0.25,2) to [out=-90, in=135] (0.5,1.5) to (1,1) to  (1.5,1.5) to [out=45, in=-90] (1.75,2);
\node[blob] at (1,1) {$V$};
\end{tz}\,=\,
 \begin{tz}[string,scale=0.85]
\path[yellowregion] (0.5,0.25) to [out=90, in=-135] (1,1) to (1.5,1.5) to (1,2) to [out=135, in=-90] (0.5,2.75) to  (0.5-\side,2.75) to (0.5-\side,0.25);
\path[orangeregion] (2.5,0.25) to [out=90, in=-45] (2,1) to (1.5,1.5) to (2,2) to [out=45, in=-90] (2.5,2.75) to (2.5+\side,2.75) to (2.5+\side,0.25);
\path[blueregion,draw] (0.5,0.25) to [out=90, in=-135] (1,1) to (1.5,1.5) to (2,1) to [out=-45, in=90] (2.5,0.25);
\path[cyanregion,draw] (0.5,2.75) to [out=-90, in=135] (1,2) to (1.5,1.5) to (2,2) to [out=45, in=-90] (2.5,2.75);
\draw[zig,decorate] (1,1) to [out=-45, in=180] (1.5,0.65) to [out=0, in=-135] (2,1) to [out=45, in=-90] (2.35,1.5) to [out=90, in=-45] (2,2) to [out=135, in=0] (1.5,2.35) to [out=180, in=45] (1,2) to [out=-135, in=90] (0.65,1.5) to [out=-90, in=135] (1,1);
\node[blob] at (1.5,1.5) {$V'$};
\node[blob] at (1,1) {$\widetilde{B}$};
\node[blob] at (1,2) {$\widetilde{A}$};
\node[blob] at (2,1) {$\widetilde{D}$};
\node[blob] at (2,2){$\widetilde{C}$};
\end{tz}
\end{calign}
Then in general, when $B$ and $\widetilde B$ are not inverses with respect to composition along the minor diagonal, the following composites are not equivalent:
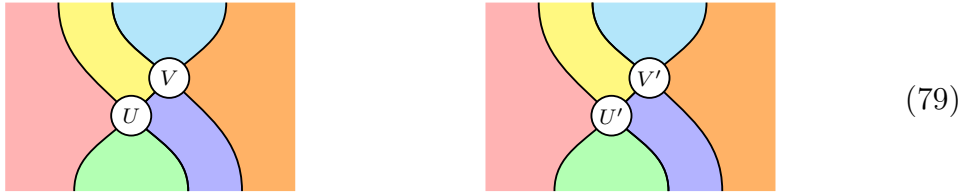
\begin{calign}
\begin{tz}[string]
\path[orangeregion] (1.75+\hrt,0) to [out=90, in=-45] (1.5,1.5) to [out=45, in=-90] (2.25,2.5) to (1.75+\hrt+\side,2.5) to (1.75+\hrt+\side,0);
\path[redregion] (0.25,0) to [out=90, in=-135] (1,1) to [out=135, in=-90] (0.75-\hrt, 2.5) to (0.75-\hrt-\side,2.5) to (0.75-\hrt-\side,0);
\path[greenregion,draw] (0.25,0) to [out=90, in=-135] (1,1) to [out=-45, in=90] (1.75,0);
\path[ blueregion,draw] (1.75,0) to [out=90, in=-45] (1,1) to (1.5,1.5) to [out=-45, in=90]  (1.75+\hrt,0);
\path[cyanregion, draw] (0.75,2.5) to [out=-90, in=135] (1.5,1.5) to [out=45, in=-90] (2.25,2.5);
\path[yellowregion,draw] (0.75,2.5) to [out=-90, in=135] (1.5,1.5) to (1,1)to [out=135, in=-90] (0.75-\hrt,2.5);
\node[blob] at (1,1) {$U$};
\node[blob] at (1.5,1.5) {$V$};
\end{tz}
&
\begin{tz}[string]
\path[orangeregion] (1.75+\hrt,0) to [out=90, in=-45] (1.5,1.5) to [out=45, in=-90] (2.25,2.5) to (1.75+\hrt+\side,2.5) to (1.75+\hrt+\side,0);
\path[redregion] (0.25,0) to [out=90, in=-135] (1,1) to [out=135, in=-90] (0.75-\hrt, 2.5) to (0.75-\hrt-\side,2.5) to (0.75-\hrt-\side,0);
\path[greenregion,draw] (0.25,0) to [out=90, in=-135] (1,1) to [out=-45, in=90] (1.75,0);
\path[ blueregion,draw] (1.75,0) to [out=90, in=-45] (1,1) to (1.5,1.5) to [out=-45, in=90]  (1.75+\hrt,0);
\path[cyanregion, draw] (0.75,2.5) to [out=-90, in=135] (1.5,1.5) to [out=45, in=-90] (2.25,2.5);
\path[yellowregion,draw] (0.75,2.5) to [out=-90, in=135] (1.5,1.5) to (1,1)to [out=135, in=-90] (0.75-\hrt,2.5);
\node[blob] at (1,1) {$U'$};
\node[blob] at (1.5,1.5) {$V'$};
\end{tz} 
\end{calign}
This behaviour was recognized by Werner~\cite{Werner:2001}, who observed that it is possible to construct inequivalent shift-and-multiply unitary error bases even when all Hadamard matrices and Latin squares come from the same equivalence classes.

It is also exploited in \Dita's construction~\cite{Dita:2004}. Consider the following equivalence transformation on the family of Hadamard matrices $K$ in \autoref{fig:binarydetailed1}(d):
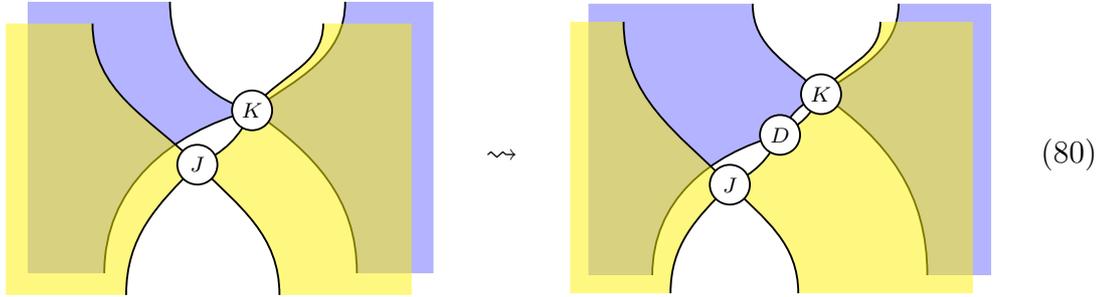
\begin{equation}\begin{tz}[string,scale=3/2.5*1.2]
\path[blueregion] (0.75,2.5) to [out=-90, in=165] (1.5,1.5) to [out=-165, in=90] (0.15,0) to (0.15-\side,0) to (0.15-\side,2.5);
\draw[bnd]   (0.75,2.5) to [out=-90, in=165] (1.5,1.5) to [out=-165, in=90] (0.15,0);
\path[yellowregion] (0.35,-\ydelta) to [out=90, in=-135] (1,1) to [out=135, in=-90] (0.75-\hrt, 2.5-\ydelta) to (-0.05-\side,2.5-\ydelta) to (-0.05-\side,-\ydelta);
\draw[bnd](0.35,-\ydelta) to [out=90, in=-135] (1,1) to [out=135, in=-90] (0.75-\hrt, 2.5-\ydelta);
\path[blueregion] (2.35,2.5) to [out=-90, in=35] (1.5,1.5) to [out=-35, in=90] (1.75 + \hrt,0) to (1.75+\hrt+\side,0) to (1.75+\hrt+\side,2.5);
\draw[bnd] (2.35,2.5) to [out=-90, in=35] (1.5,1.5) to [out=-35, in=90] (1.75 + \hrt,0);
\path[ yellowregion] (1.75,-\ydelta) to [out=90, in=-45] (1,1) to [out=25, in=-115] (1.5,1.5) to [out=55, in=-90] (2.15,2.5-\ydelta) to (1.55+\hrt+\side,2.5-\ydelta) to (1.55+\hrt+\side,-\ydelta); 
\draw[bnd](1.75,-\ydelta) to [out=90, in=-45] (1,1) to [out=25, in=-115] (1.5,1.5) to [out=55, in=-90] (2.15,2.5-\ydelta);
\node[blob] at (1,1) {$J$};
\node[blob] at (1.5,1.5) {$K$};
\end{tz}
\hspace{0.5cm}\rightsquigarrow \hspace{0.5cm}
\def \xd {0.2}
\begin{tz}[string,scale=1.2]
\node[blob] (D) at (1.55,1.55) {$D$};
\path[blueregion] (1.25,3) to [out=-90, in=135] (2,2) to [out=-155, in=65] (D.center) to [out=-165, in=up]  (0.15,0) to (0.15-\side,0) to (0.15-\side,3);
\draw[bnd]   (1.25,3) to [out=-90, in=135] (2,2) to [out=-155, in=65] (D.center) to [out=-165, in=up]  (0.15,0);
\path[yellowregion] (0.35,-\xd) to [out=90, in=-135] (1,1) to [out=135, in=-90] (1.25-2*\hrt, 3-\xd) to (0.15-\xd-\side,3-\xd) to (0.15-\xd-\side,-\xd);
\draw[bnd] (0.35,-\xd) to [out=90, in=-135] (1,1) to [out=135, in=-90] (1.25-2*\hrt, 3-\xd) ;
\path[blueregion] (2.85,3) to [out=-90, in=35] (2,2) to [out=-35, in=90] (1.75 + 2*\hrt,0) to (1.75+2*\hrt+\side,0) to (1.75+2*\hrt+\side,3);
\draw[bnd] (2.85,3) to [out=-90, in=35] (2,2) to [out=-35, in=90] (1.75 + 2*\hrt,0);
\path[ yellowregion] (1.75,-\xd) to [out=90, in=-45] (1,1) to [out=25, in=-115] (D.center) to [out=25, in=- 115] (2,2) to [out=55, in=-90] (2.65,3-\xd) to (1.75-\xd+2*\hrt+\side,3-\xd) to (1.75-\xd+2*\hrt+\side,-\xd); 
\draw[bnd] (1.75,-\xd) to [out=90, in=-45] (1,1) to [out=25, in=-115] (D.center) to [out=25, in=- 115] (2,2) to [out=55, in=-90] (2.65,3-\xd);
\node[blob] at (1,1) {$J$};
\node[blob] at (2.,2.) {$K$};
\end{tz}\end{equation}
This allows us to introduce a controlled family of free scalars $D$ in the resulting Hadamard matrix, a technique used to construct continuously-parameterized families of Hadamard matrices in~\cite[Section 4]{Dita:2004}. This construction is one of the reasons why continuous families of Hadamard matrices are comparatively better understood in composite dimensions. In fact, it was conjectured by Popa~\cite{Popa:1983} that there are only finitely many inequivalent Hadamard matrices (and in particular no continuous families) in prime dimensions. This conjecture was disproven by Petrescu~\cite{Petrescu:1997} in 1997, who constructed several continuous families of Hadamard matrices in certain prime dimensions.

\section{A new unitary error basis}
\label{sec:newUEB}

Construction techniques for unitary error bases have been widely studied~\cite{Knill:1996, Klappenecker:2003, Musto:2015}. The methods proposed in the literature fall into the following two classes:
\begin{itemize}
\item Quantum shift-and-multiply (\textbf{QSM})~\cite{Musto:2015}. Requires a quantum Latin square and a family of Hadamard matrices. Generalizes the earlier shift-and-multiply (\textbf{SM)}~\cite{Werner:2001} and Hadamard (\textbf{HAD}) methods.
\item Algebraic (\textbf{ALG})~\cite{Knill:1996_2}. Requires a finite group equipped with a projective representation satisfying certain requirements.
\end{itemize}
As shown in \autoref{cor:QSM}, the quantum shift-and-multiply method is a special case of our biunitary composition method (\textbf{BC}). We thus arrive at the following Venn diagram summarising all known constructions of unitary error bases, extending a Venn diagram in~\cite{Musto:2015}:

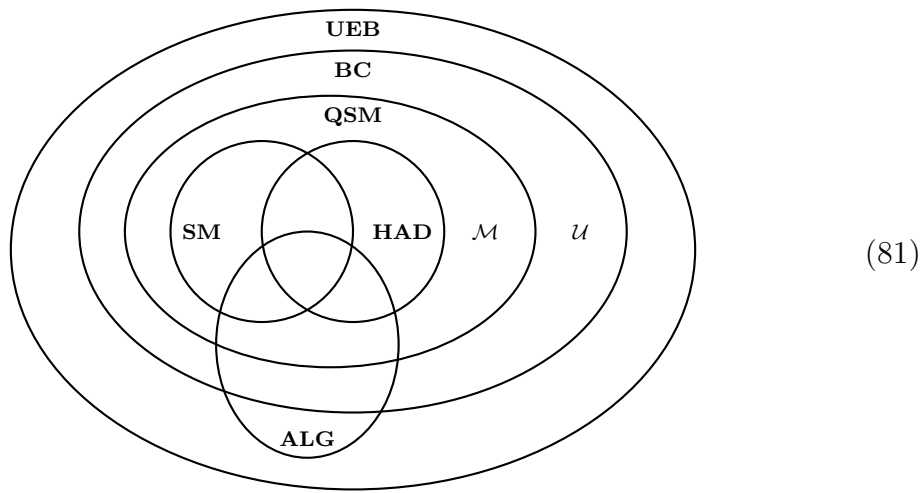
\begin{equation}
\def\s{1.2}
\begin{aligned}
\begin{tikzpicture}[thick, scale=\s, font=\scriptsize, xscale=1, yscale=1]

\node [ellipse, draw, minimum width=\s*2cm, minimum height=\s*2cm] at (0,0) {};
\node [right] at (-1,0) {\textbf{SM}};

\node [ellipse, draw, minimum width=\s*2cm, minimum height=\s*2cm] at (1,0) {};
\node [left] at (2,0) {\textbf{HAD}};

\node [ellipse, draw, minimum width=\s*2cm, minimum height=\s*2.5cm] at (0.5,-1.25) {};
\node [above] at (0.5,-2.5) {\textbf{ALG}};

\node [ellipse, draw, minimum width=\s*4.5cm, minimum height=\s*3cm] at (0.75,0) {};
\node [below] at (1,1.5) {\textbf{QSM}};
\node at (2.45,0) {$\mathcal M$};

\node [ellipse, draw, minimum width=\s*7.5cm, minimum height=\s*5.3cm] at (1,-0.20) {};
\node [below] at (1,2.45) {\textbf{UEB}};

\node [ellipse, draw, minimum width=\s*6cm, minimum height=\s*4cm] at (1,0) {};
\node [below] at (1,2) {\textbf{BC}};
\node at (3.5,0) {$\mathcal U$};

\end{tikzpicture}
\end{aligned}
\end{equation}
In \cite{Musto:2015}, a unitary error basis $\mathcal M$ was constructed which lies in \textbf{QSM}, but outside \textbf{SM}, \textbf{HAD} and \textbf{ALG}. In this section, we construct a unitary error basis $\mathcal U$ which lies in \textbf{BC}, but outside \textbf{QSM} and \textbf{ALG}. It follows that our biunitary composition techniques are able to produce genuinely new quantum structures.

In \autoref{sec:constructingu}, we give the construction of $\mathcal U$. In \autoref{sec:notnice} we show that it is not equivalent to a UEB arising from the algebraic construction, and in \autoref{sec:notqsm} we show it is not equivalent to one arising from the quantum shift-and-multiply construction. An accompanying \textit{Mathematica} notebook is available at \href{https://arxiv.org/abs/1609.07775}{arXiv:1609.07775}.

\subsection{Constructing $\mathcal U$}
\label{sec:constructingu}

We employ the construction of \autoref{fig:higher}(a)  and \autoref{cor:4aryconstruction} for $n=2$, with the following definitions for the  (constant family consisting of the) Hadamard matrix $H$, the quantum Latin squares $P$ and $Q$, and the unitary error basis $\mathcal{V}$:
\begin{align}
H&:= \left(
\begin{array}{cccc}
 1 & 1 & 1 & 1 \\
 1 & i & \minus 1 & \minus i \\
 1 & \minus 1 & 1 & \minus 1 \\
 1 & \minus i & \minus 1 & i \\
\end{array}
\right)
\\[7pt]
P &:= \grid{\ket{1} & \ket{2} & \ket{3} & \ket{4}
\\\hline
\frac{1}{\sqrt{2}}(\ket{2}-\ket{3})
& \frac{1}{\sqrt{5}}(i\ket{1}+2\ket{4})
& \frac{1}{\sqrt{5}}(2\ket{1}+i\ket{4})
& \frac{1}{\sqrt{2}}(\ket{2}+\ket{3})
\\\hline
\frac{1}{\sqrt{2}}(\ket{2}+\ket{3})
& \frac{1}{\sqrt{5}}(2\ket{1}+i\ket{4})
& \frac{1}{\sqrt{5}}(i\ket{1}+2\ket{4})
& \frac{1}{\sqrt{2}}(\ket{2}-\ket{3})
\\\hline
\ket{4} & \ket{3} & \ket{2} & \ket{1}}~
\\[7pt]
Q &:= \grid{\ket{1} & \ket{4} & \ket{2} & \ket{3}
\\\hline
\ket{4}&\ket{1}& \ket{3}&\ket{2}
\\\hline
\ket{3}&\ket{2}&\ket{1}&\ket{4}
\\\hline
\ket{2} & \ket{3} & \ket{4} & \ket{1}}
\\[7pt]
\mathcal{V}&:= \left\{ \left(\begin{array}{cc}
 1 & 0 \\
 0 & 1 \end{array}\right),\hspace{0.25cm}\left(
\begin{array}{cc}
 1 & 0 \\
 0 & \minus 1 \\
\end{array}
\right),\hspace{0.25cm}\left(
\begin{array}{cc}
 0 & 1 \\
 1 & 0 \\
\end{array}
\right),\hspace{0.25cm}
\left(
\begin{array}{cc}
 0 & 1 \\
 \minus 1 & 0 \\
\end{array}
\right)\right\}
\end{align}

The resulting unitary error basis $\widetilde {\mathcal U}$ is calculated according to formula \eqref{eq:8aexplicit}. We then define an equivalent UEB $\mathcal U$ as follows:
\begin{equation}
\label{eq:theUEBU}
\mathcal{U} = \left\{U_{abc}:= \widetilde{U}_{111}^\dagger \widetilde{U}_{abc}\,|\, a,b,c \in [4] \right
\}
\end{equation}
We choose $\mathcal U$ in this way to ensure that $U_{111} = \mathbbm{1}$. The full UEB is presented in \autoref{appendix:UEB}, and the commutativity structure of its elements is visualized in \autoref{fig:commutativity}.

\begin{figure}[h]
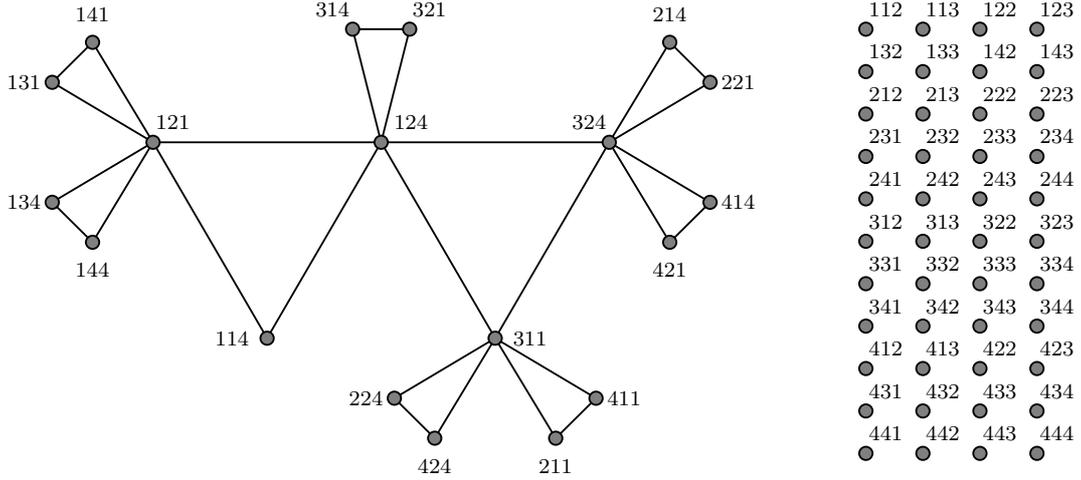

\newcounter{indexcount}
\tikzset{vertex/.style={draw, circle, fill=gray, inner sep=1pt, minimum width=5pt, font=\scriptsize, node on layer=front,  line width=0.7pt}}
\def \vdelta{0.35} 
\def \rt {3.464} 
\newcommand{\tri}[4]{     
\begin{scope}[rotate around={#2:(#1)}]
\draw[string] (#1.center) to +(2,0.5) to +(2,-0.5) to (#1.center);
\node[vertex] (V1) at ($(#1)+(2,0.5)$){};
\node[vertex] (V2) at ($(#1)+(2,-0.5)$){};
\node[dimension] at ($(V1)+(\vdelta,\vdelta)$){#3};
\node[dimension] at ($(V2)+(\vdelta,-\vdelta)$){#4};
\end{scope}
}
\begin{calign}\nonumber \begin{tz}[string,scale=0.75]
\node[vertex] (L) at (0,0){};
\node[dimension] at ($(L)+(\vdelta,\vdelta)$){$121$};
\node[vertex] (M) at (4,0){};
\node[dimension] at ($(M)+(1.5*\vdelta,\vdelta)$){$124$};
\node[vertex] (R) at (8,0){};
\node[dimension] at ($(R)+(-\vdelta,\vdelta)$){$324$};
\node[vertex] (BL) at (2,-\rt){};
\node[dimension] at ($(BL)+(-1.75*\vdelta,0)$){$114$};
\node[vertex] (BR) at (6,-\rt){};
\node[dimension] at ($(BR)+(1.75*\vdelta,0)$){$311$};
\draw[string] (L.center) to (M.center) to (R.center) to (BR.center) to (M.center) to (BL.center) to (L.center);
\tri{R}{45}{$214$}{$221$}
\tri{R}{-45}{$414$}{$421$}
\tri{M}{90}{$314$}{$321$}
\tri{L}{135}{$131$}{$141$}
\tri{L}{-135}{$144$}{$134$}
\tri{BR}{-135}{$424$}{$224$}
\tri{BR}{-45}{$411$}{$211$}
\foreach \y in {0,...,10}
     \foreach \x in {0,...,3}
        { \node[vertex] at (12.5+\x, 2-0.75*\y){};
        \node (\arabic{indexcount}) at (12.5+\x+\vdelta,2-0.75*\y +\vdelta){};
         \stepcounter{indexcount}}     
\setcounter{indexcount}{0}      
\foreach \z in {112,113,122,123,132,133,142,143,212,213,222,223,231,232,233,234,241,242,243,244,312,313,322,323,331,332,333,334,341,342,343,344,412,413,422,423,431,432,433,434,441,442,443,444}      
      {\node[dimension] at (\arabic{indexcount}){$\z$};
       \stepcounter{indexcount}} 
\end{tz}\end{calign}
\vspace{-15pt}
\caption{The graph with vertices given by elements of $\mathcal U$, and edges between commuting elements. The element $U_{111}= \mathbbm{1}$ is omitted.\label{fig:commutativity}}
\end{figure}

\subsection{Nice error bases}
\label{sec:notnice}

In this subsection we define nice error bases, and show that $\mathcal U$ is not equivalent to a nice error basis.

\begin{definition}[Knill~\cite{Knill:1996_2}]A \textit{nice error basis} is a unitary error basis $\mathcal{U} = \left\{ U_i~|~ i\in I\right\}$ with $\mathbbm{1}\in \mathcal{U}$ that is (up to phases) closed under multiplication. In other words, for each $a,b\in I$, there exists a scalar $\omega(a,b)\in U(1)$ and an index $a*b\in I$, such that
\begin{equation}U_aU_b = \omega(a,b) U_{a*b}.\end{equation}
\end{definition}

\noindent
Nice unitary error bases correspond to certain projective representations of finite groups.

The following is a strong property of nice error bases.
\begin{proposition}[Musto \& V.~{\cite[Proposition 43]{Musto:2015}}] \label{thm:notnice}Let $\mathcal{V}$ be a unitary error basis containing the identity matrix, such that $\mathcal{V}$ is equivalent to a nice error basis. Then up to multiplication by a phase, $\mathcal{V}$ is closed under taking adjoints.
\end{proposition}

\noindent
It follows that our new unitary error basis $\mathcal U$ is not equivalent to a nice error basis.

\begin{theorem} The unitary error basis $\mathcal{U}$ of \autoref{appendix:UEB} is not equivalent to a nice unitary error basis.
\end{theorem}
\begin{proof} We note that $U_{112}^\dagger$ is not proportional to any matrix in $\mathcal{U}$, and hence by \autoref{thm:notnice} the result follows.
\end{proof}

\subsection{Quantum shift-and-multiply bases}
\label{sec:notqsm}

Quantum shift-and-multiply UEBs were defined in \cite{Musto:2015}, and the construction exactly matches the biunitary composite \autoref{fig:binarydetailed2}(b).  In this subsection we demonstrate that all quantum shift-and-multiply UEBs have a particular commutativity property, and therefore demonstrate that our new UEB $\mathcal U$ does not arise in this way.

The following proposition gives a strong constraint on the structure of quantum shift-and-multiply bases.

\begin{proposition}\label{thm:notQSM} Let $\mathcal{V}$ be an $m$\-dimensional UEB which contains the identity matrix, such that $\mathcal{V}$ is equivalent to a quantum shift-and-multiply UEB. Then $\mathcal{V}$ contains $m$ pairwise-commuting matrices.
\end{proposition}

\begin{proof} Quantum shift-and-multiply UEBs are of the form $V_{ab}= Q^{}_a D_a^b$ for unitary matrices $Q_a$ and unitary diagonal matrices $D_a^b$. Using the definition of equivalence of unitary error bases from \autoref{sec:equivalenceapplications}, it follows that $\mathcal{V}$ is of the following form:
\begin{align*}
\mathcal{V} &= \left\{ c_{ab} A Q_aD_a^b B \,|\, a,b\in[m]\right\}
\intertext{Since $\mathbbm{1}\in \mathcal{V}$, there are indices $a_0, b_0$ such that $ c_{a_0b_0} A Q_{a_0} D_{a_0}^{b_0} B= \mathbbm{1}$. Defining the diagonal matrix $D:= \left( c_{a_0b_0} D_{a_0}^{b_0}\right)^\dagger$, this means that}
A &= B^\dagger D Q_{a_0}^\dagger
\intertext{and hence that}
\mathcal{V} &= \left \{ c_{ab} B^\dagger D Q_{a_0}^\dagger Q_a D_a^b B\,|\,a,b\in[m]\right\}.
\end{align*}
All matrices with $a=a_0$ pairwise commute, and there are $m$ of these.
\end{proof}

\noindent
The desired result follows.

\begin{theorem} The unitary error basis $\mathcal{U}$ of \autoref{appendix:UEB} is not equivalent to a quantum shift-and-multiply basis.
\end{theorem}
\begin{proof} The commutativity graph of $\mathcal{U}$ is shown in \autoref{fig:commutativity}. It is clear by inspection that every pairwise-commuting subset contains at most $4$ elements (including the element $U_{111} = \mathbbm{1}$, which is omitted from the graph.) The result then follows from \autoref{thm:notQSM}.
\end{proof}

{\small
\bibliographystyle{plainurl}
\bibliography{references}
}

\appendix

\pagebreak
\section{The UEB $\mathcal{U}$}\label{appendix:UEB}
The following is the UEB $\mathcal{U}$ constructed in \autoref{sec:newUEB}.
\input{example_ueb}

\end{document}

%% file: example_ueb.tex
%


\newcounter{uebcounter}
\setcounter{uebcounter}{1}
\newcounter{linelength}
\setcounter{linelength}{0}

\newcounter{ueba}
\newcounter{uebb}
\newcounter{uebc}
\newcommand\updatecounter{ 
\setcounter{ueba}{
\intcalcDiv{\value{uebcounter}-1}{16}+1
} 
\setcounter{uebb}{
\intcalcDiv{\value{uebcounter}-1 - 16*(\value{ueba}-1)}{4}+1
}
\setcounter{uebc}{
\value{uebcounter}-16*(\value{ueba}-1) - 4*(\value{uebb}-1)
}
}

\renewcommand{\arraystretch}{1.6}

\newcolumntype{C}{>{\centering\arraybackslash$}p{0.42cm}<{$}}
\newcommand\mat[1]{
\ifnum \value{linelength}=3 \setcounter{linelength}{0}\\[5pt] \fi
\ifnum \value{linelength}<3 \addtocounter{linelength}{1}\fi
\updatecounter
U_{\arabic{ueba}\arabic{uebb}\arabic{uebc}}\addtocounter{uebcounter}{1}
&=
\left(\raisebox{2.7pt}{\hspace{-3pt}\makebox[3.2cm][l]{{$\tiny\begin{array}{@{}C@{}C@{}C@{}C@{}C@{}C@{}C@{}C@{}}#1\end{array}$}}}\hspace{-2pt}\right)&
}
\newcommand\entry[1]{\makebox[0.40cm][c]{\ensuremath{#1}}}

